\documentclass{amsart}
\usepackage{amsthm,amsmath}
\usepackage{graphicx}
\usepackage{mathptmx}
\usepackage{latexsym}
\usepackage{amsfonts}
\newtheorem{thm}{Theorem}[section]

\newtheorem{lem}[thm]{Lemma}
\newtheorem{rem}[thm]{Remark}
\newtheorem{cor}[thm]{Corollary}
\numberwithin{equation}{section}

\newcommand{\cF}{{\mathcal F}}

\newcommand{\cL}{{\mathcal L}}

\newcommand{\cV}{{\mathcal V}}

\newcommand{\bg}{{\mathfrak b}}
\newcommand{\te}{{\theta}}

\newcommand{\Om}{{\Omega}}

\newcommand{\ve}{{\varepsilon}}
\newcommand{\del}{{\delta}}

\newcommand{\gam}{{\gamma}}

\newcommand{\sig}{{\sigma}}

\newcommand{\be}{{\beta}}

\newcommand{\bbN}{{\mathbb N}}

\newcommand{\bbR}{{\mathbb R}}

\newcommand{\bbI}{{\mathbb I}}

\begin{document}
\title[]{Binomial Approximations for Barrier Options of Israeli Style.}%
 \vskip 0.1cm
 \author{ Yan Dolinsky and Yuri Kifer\\
 \vskip 0.1cm
Institute of Mathematics\\
Hebrew University\\
Jerusalem, Israel}%

\address{
 Institute of Mathematics, The Hebrew University, Jerusalem 91904, Israel}%
 \email{ yann1@math.huji.ac.il, kifer@math.huji.ac.il}%
\thanks{ Partially supported by the ISF grant no. 130/06}%
\subjclass[2000]{Primary: 91B28 Secondary: 60F15, 91A05}%
\keywords{barrier game options, Dynkin games, shortfall risk, binomial
approximations, Skorokhod embedding.}%

 \date{\today}
\begin{abstract}
We show that prices and shortfall risks of game (Israeli) barrier options
in a sequence of binomial approximations of the Black--Scholes (BS) market
converge to the corresponding quantities for similar game barrier options
in the BS market with path dependent payoffs and the speed of convergence
is estimated, as well. The results are new also for usual American style
options and they are interesting from the computational point of view,
as well, since in binomial markets these quantities can be obtained via
dynamical programming algorithms. The paper continues the study of \cite{Ki2}
and \cite{DK2} but requires substantial additional arguments in view of
pecularities of barrier options which, in particular, destroy the regularity
 of payoffs needed in the above papers.
\end{abstract}
\maketitle
\markboth{Ya. Dolinsky and Yu. Kifer}{Barrier options of Israeli style}%
\renewcommand{\theequation}{\arabic{section}.\arabic{equation}}
\pagenumbering{arabic}

\section{Introduction}\label{sec:1}\setcounter{equation}{0}
This paper deals with knock--out and knock--in double barrier
options of the game (Israeli) type sold in a standard securities
market consisting of a nonrandom component $b_t$ representing the
value of a savings account at time $t$ with an interest rate $r$ and
of a random component $S_t$ representing the stock price at time
$t$. As usual, we view $S_t,t>0$ as a stochastic process on a
probability space $(\Omega,\mathcal{F},P)$ and we assume that it
generates a right continuous filtration $\{\mathcal{F}_{t}\}$. The
setup includes also two right continuous with left limits
(\textit{cadlag}) stochastic payoff processes $X_t\geq{Y_t}\geq{0}$
adapted to the above filtration. Recall, that a game contingent
claim (GCC) or a game option was defined in \cite{Ki1} as a contract
between the seller and the buyer of the option such that both have
the right to exercise it at any time up to a maturity date (horizon)
$T$ which in this paper assumed to be finite. If the buyer exercises
the contract at time $t$ then he receives the payment $Y_t$, but if
the seller exercises (cancels) the contract before the buyer then
the latter receives $X_t$. The difference $\Delta_t=X_t-Y_t$ is the
penalty which the seller pays to the buyer for the contract
cancellation. In short, if the seller will exercise at a stopping
time $\sigma\leq{T}$ and the buyer at a stopping time $\tau\leq{T}$
then the former pays to the latter the amount
$H(\sigma,\tau)=X_\sig\mathbb{I}_{\sig<\tau}+
Y_\tau\mathbb{I}_{\tau\leq{\sig}}$ where we set $\mathbb{I}_{A}=1$
 if an event $A$ occurs and $\mathbb{I}_{A}=1$ if not.

A hedge (for the seller) against a GCC is defined here as a pair
$(\pi,\sigma)$ which consists of a self financing strategy $\pi$
(i.e. a trading strategy with no consumption and no infusion of
capital) and a stopping time $\sigma$ which is the cancellation time
for the seller. A hedge is called perfect if no matter what exercise
time the buyer chooses, the seller can cover his liability to the
buyer (with probability one). The option price $\cV^*$ is defined as
the minimal initial capital which is required for a perfect hedge,
i.e. for any $x>\cV^*$ there is a perfect hedge with an initial
capital $x$. Recall, (see \cite {Ki1}) that pricing a GCC in a
complete market leads to the value of a zero sum optimal stopping
(Dynkin's) game with discounted payoffs $\tilde{X_t}=
b_0\frac{X_t}{b_t}$, $\tilde{Y_t}=b_0\frac{Y_t}{b_t}$ considered
under the unique martingale measure $\tilde{P}\sim{P}$.

We consider a double knock--out barrier option with a two constant
barriers $L,R$ such that $0\leq L<S_0<R\leq\infty$ which means that
the option is worthless to its holder (buyer) at the first time
$\tau_I$ the stock price $S_t$ exits the open interval $I=(L,R)$.
Thus for $t\geq\tau_{(L,R)}$ the payoff is $X_t=Y_t=0$. For
$t<\tau_{(L,R)}$ we consider path dependent payoffs. Such a contract
is of potential value to a buyer who believes that the stock price
will not exit the interval $I$ up to a maturity date and to a seller
who believes otherwise and does not want to have to worry about
hedging if the stock price will reach one of the barriers $L,R$.
Double knock--in barrier options which start when $S_t$ exits an
interval $I$ will be considered, as well. Observe, that we view
barrier game options as a generalization of regular game options
where $L=0$ and $R=0$ which provides a way of their simultaneous
treatment.

The Cox, Ross and Rubinstein (CRR) binomial model which was
introduced in \cite{CRR} is an efficient tool to approximate
derivative securities in a Black--Scholes (BS) market. We will show
that for a double barrier options in the BS model the option price
can be approximated by a sequence of option prices of a barrier
options (with the same barriers) in appropriate CRR $n$--step models
with errors bounded by $Cn^{-1/4}(\ln{n})^{3/4}$ where $C$ is a
constant which does not depend on the value of the barriers. These
both  generalize the results from \cite{Ki2} which were obtained
 for regular (without barriers) game options with path
dependent payoffs and provide an algorithm for computation of this
important class of derivative securities since pricing of game
options in CRR markets can be done by dynamical programming (see
\cite{Ki1}).

Pricing of European and American type barrier options was studied in
several papers (see, for instance, \cite{GY} and \cite{KW}) and a
number of papers dealt with error estimates for discrete
approximations of barrier European options (see, for instance,
\cite{BGK1}, \cite{BGK2}, \cite{RS} and references there). On the
other hand, binomial approximations and their error estimates for
look back American style, let alone for Israeli style, barrier
options were not studied rigorously before.

We also deal with partial hedging (under the same assumption on the
payoffs) which becomes relevant if for instance, an investor
(seller) is not willing for various reasons to tie in a hedging
portfolio the full initial capital required for a perfect hedge. In
this case the seller is ready to accept a risk that his portfolio
value at an exercise time may be less than his obligation to pay and
he will need additional funds to fullfil the contract. Thus a
portfolio shortfall comes into the picture and by this reason we
distinguish here between hedges and perfect hedges.

In this paper we deal with certain type of risk called the shortfall
risk (cf. for instance, \cite{CK}, \cite{DK1}, \cite{FL}, \cite{M})
which was defined for game options in \cite{DK1} by the formulas
\begin{equation*}
R(\pi,\sigma)=\sup_{\tau}E(Q(\sigma,\tau)-b_0\frac{V^\pi_{\sigma\wedge\tau}}
{b_{\sigma\wedge\tau}})^+\,\,\mbox{and}\,\,
R(x)=\inf_{(\pi,\sigma)}R(\pi,\sigma)
\end{equation*}
where the supremum is taken over all stopping times not exceeding a
horizon $T$, the infimum is taken over all hedges with an initial
capital $x$, $Q(\sigma,\tau)=\tilde{X}_s\mathbb{I}_{s<t}+\tilde{Y}_t
\mathbb{I}_{t\leq{s}}$ is the discounted payoff, $V^\pi_t$ is the
portfolio value of $\pi$ at time $t$ and $E$ denotes the expectation
with respect to the objective probability measure $P$. An investor
(seller) whose initial capital $x$ is less than the option price
still wants to compute the minimal possible shortfall risk and to
find a hedge with the initial capital $x$ which minimizes or
"almost" minimizes the shortfall risk.

In \cite{DK1} we proved that for a game option in the multinomial
model with general payoffs there exists a hedge which minimizes the
shortfall risk under constraint on the initial capital, and the
above hedge together with the corresponding shortfall risk can be
computed via a dynamical programming procedure. For game option in
the BS model the problem of finding an optimal hedge is more
complicated and for now remains open even for regular payoffs. We
will prove that in the BS model the shortfall risk $R(x)$ of a
seller with initial capital $x$ for double barrier options is a
limit of the shortfall risks $R_n(x)$ for double barrier options in
the CRR markets with the same barriers and initial capital as in the
BS model. Here we are able to provide only a one sided error
estimate $R(x)-R_n(x)\leq \tilde{C}n^{-1/4}(\ln{n})^{-3/4}$ where
$\tilde{C}>0$ is a constant which does not depend on the value of
the barriers. These results generalize the ones which were obtained
in \cite{DK2} for regular game options with path dependent payoffs
and again provide a way of computation of the shortfall risk for
barrier game options. Binomial approximations of shortfall risks for
barrier options were not studied before even for European options.

For a given initial capital $x$ we will use hedges which minimize
the shortfall risk in CRR markets under the above constraint on the
initial capital, in order to construct hedges which "almost"
minimize the shortfall risk in the BS model under the same
constraint on the initial capital. Furthermore we will see that the
corresponding portfolios are managed on a finite set of random times
as it was done in \cite{DK2} for regular game options. We consider
also another situation where the seller of a game option in the BS
model has an initial capital which is a little bit larger than the
option price. In this case we use perfect hedges in CRR markets in
order to build explicitly hedges with small shortfall risks in the
BS model where the corresponding portfolios are managed on a finite
set of random times as it was done in \cite{Ki2} for regular game
options.

Our main tool is the Skorohod type embedding of sums of i.i.d.
random variables into a Brownian motion with a constant drift. This
tool was employed for a regular options in \cite{DK2} and \cite{Ki2}
in order to obtain error estimates for approximation of shortfall
risks and for approximation of option prices, respectively. However,
in the barrier options case the payoffs lose their Lipschitz
continuity which was crucial in \cite{Ki2} and \cite{DK2}, and so
this case requires substantial additional arguments and estimates
leading to a generalization of our previous results. Moreover,
observe that discontinuities of payoffs occur at random times since
they depend on the stock behavior. Since the discretisation does not
necessarily adjusted to the barrier value where discontinuities
occur we have to estimate the deviation of the option price as the
barrier value changes a bit which is the key additional part of the
proof in comparison to \cite{DK2} and \cite{Ki2} (see Lemmas
\ref{lem3.3}, \ref{lem3.4} and \ref{lem5.1}).

Main results of this paper are formulated in the next section where
we discuss also the Skorohod type embedding. In Section \ref{sec3}
we introduce recursive formulas which enable us to compare various
option prices and risks. In this section we also derive auxiliary
estimates for option prices and risks. In Section \ref{sec4} we
complete the proof of main results of the paper for knock--out
options while in Section \ref{sec5} we deal with the knock--in case
which requires a somewhat different definitions and a separate
treatment yielding a bit worse error estimates. Some definitions and
estimates in this paper are similar to \cite{DK2} and \cite{Ki2} but
for the sake of the reader and in order to keep the paper relatively
self-contained we repeat them here with needed modifications. On the
other hand, the reader may benefit reading this paper consulting
occasionally for more details also \cite{DK2} and \cite{Ki2}.

\section{Preliminaries and main results}\label{sec:2}\setcounter{equation}{0}
First, we describe the setup. Denote by $M[0,t]$ the space of Borel
measurable functions on [0,t] with the uniform metric
$d_{0t}(\upsilon,\tilde{\upsilon})=\sup_{0\leq{s}\leq{t}}|\upsilon_s
-\tilde{\upsilon}_s|$. For each $t>0$ let $F_t$ and $\Delta_t$ be
nonnegative functions on $M[0,t]$ such that for some constant
$\cL\geq{1}$ and for any $t\geq{s}\geq{0}$ and
$\upsilon,\tilde{\upsilon}\in{M[0,t]}$,
\begin{equation}\label{2.1}
|F_s(\upsilon)-F_s(\tilde{\upsilon})|+|\Delta_s(\upsilon)-
\Delta_s(\tilde{\upsilon})|\leq{\cL(s+1)d_{0s}(\upsilon,\tilde{\upsilon})},
\end{equation}
and
\begin{eqnarray}\label{2.2}
&|F_t(\upsilon)-F_s({\upsilon})|+|\Delta_t(\upsilon)-
\Delta_s({\upsilon})|\\
&\leq{\cL(|t-s|(1+\sup_{u\in{[0,t]}}|\upsilon_u|)
+\sup_{u\in{[s,t]}}|\upsilon_u-\upsilon_s|)}.\nonumber
\end{eqnarray}
By (\ref{2.1}), $F_0(\upsilon)=F_0(\upsilon_0)$ and
$\Delta_0(\upsilon)=\Delta_0(\upsilon_0)$ are functions of
$\upsilon_0$ only and by (\ref{2.2}),
\begin{equation}\label{2.3}
 F_t(\upsilon)+\Delta_t({\upsilon})\leq{
F_0(\upsilon_0)+\Delta_0({\upsilon_0})+\cL(t+2)(1+\sup_{0\leq{s}\leq{t}}
|\upsilon_s|)}.
\end{equation}
Next we consider a complete probability space ($\Omega_B$,
$\mathcal{F}^{B}$, $P^{B}$) together with a standard one-dimensional
continuous in time Brownian motion \{$B_t\}_{t=0}^\infty$, and the
filtration $\mathcal{F}^{B}_t=\sigma{\{B_s|s\leq{t}\}}$. A BS
financial market consists of a savings account and a stock whose
prices $b_t$ and $S^B_t$ at time $t$, respectively, are given by the
formulas
\begin{equation}\label{2.4}
b_t=b_0e^{rt}\,\,\mbox{and}\,\, S^B_t=S_0e^{rt+\kappa{B^{*}_t}},\,\,
b_0,S_0>0
\end{equation}
where
\begin{equation}\label{2.5}
B^{*}_t=(\frac{\mu}{\kappa}-\frac{\kappa}{2})t+B_t,\, t\geq{0},
\end{equation}
$r$ is the interest rate, $\kappa>0$ is called volatility and $\mu$
is another parameter. Denote by $\tilde{S}^B_t=e^{-rt}S^{B}_t$ the
discounted stock price.

For any open interval $I=(L,R)$ such that $0\leq{L}<S_0<R\leq\infty$
let
\begin{equation}\label{2.6}
\tau_I=\inf\{t\geq{0}|S^{B}_t\notin{I}\}
\end{equation}
be the first time the stock price exit from the interval $I$.
Clearly $\tau_I$ is a stopping time (not necessary finite since we
allow the cases $L=0$ and $R=\infty$). In this paper we assume that
either $L>0$ or $R<\infty$ while the case $L=0$ and $R=\infty$ of
regular options is treated in \cite{Ki2} and \cite{DK2}. Consider a
game option with the payoffs
\begin{equation}\label{2.7}
Y^I_t=F_t(S^B)\mathbb{I}_{t<\tau_I} \ \  \mbox{and} \ \
X^I_t=G_t(S^B)\mathbb{I}_{t<\tau_I}, \ \ t\geq{0}
\end{equation}
where $G_t=F_t+\Delta_t$ with $F$ and $\Delta$ satisfying
(\ref{2.1}) and (\ref{2.2}), $S^B=S^B(\omega)\in{M[0,\infty)}$ is a
random function taking the value $S^B_t=S^B_t(\omega)$ at
$t\in{[0,\infty)}$. When considering
 $F_t(S^B),\, G_t(S^B)$ for $t<\infty$ we take the
restriction of $S^B$ to the interval $[0,t]$. Denote by $T$ the
horizon of our game option assuming that $T<\infty$. Observe that
the contract is "knocked--out" (i.e. becomes worthless to the buyer)
at the first time that the stock price exit from the interval $I$.
The case of knock--in options will be considered in Section
\ref{sec5}. The discounted payoff function is given by
\begin{equation}\label{2.8}
Q^{B,I}(s,t)=\tilde{X}^I_s\mathbb{I}_{s<t}+\tilde{Y}^I_t\mathbb{I}_{t\leq{s}},
\end{equation}
where $\tilde{Y}^I_t=e^{-rt}Y^I_t$ and $\tilde{X}^I_t=e^{-rt}X^I_t$
are the discounted payoffs.

Among examples of barrier options which fit our setup are put or
call barrier options given by
\begin{equation*}
\Delta\equiv\delta, \  \ F_t(\upsilon)=(K-\upsilon_t)^{+}\ \
\mbox{or} \ \ F_t(\upsilon)=(\upsilon_t-K)^{+},
\end{equation*}
 respectively, Russian type barrier options given by
\begin{equation*}
 F_t(\upsilon)=\max(m,\sup_{[0,t]}\upsilon_t) \ \ \mbox{and} \
\ \Delta_t(\upsilon)=\delta\upsilon_t,
\end{equation*}
and integral put or call barrier options given by
\begin{equation*}
\Delta_t(\upsilon)=\int_{0}^t{\delta_u(\upsilon_u)du}, \
F_t(\upsilon)=(K-\int_{0}^t{f_u(\upsilon_u)du})^{+} \ \ \mbox{or} \
\ F_t(\upsilon)=(\int_{0}^t{f_u(\upsilon_u)du}-K)^{+},
\end{equation*}
 respectively, where we assume that for all $x,y,u\geq{0}$,
\begin{equation*}
|f_u(x)-f_u(y)|+|\delta_u(x)-\delta_u(y)|\leq{\cL|x-y|} \ \
\mbox{and} \ \ f_u(x)+\delta_u(x)\leq{\cL x}
\end{equation*}
where $\cL$ is the same constant as in (\ref{2.1}) and (\ref{2.2}).

Denote by $\tilde{P}^{B}$ the unique martingale measure for the BS
model. Using standard arguments it follows that the restriction of
the probability measure $\tilde{P}^{B}$ to the $\sigma$--algebra
$\mathcal{F}^{B}_t$ satisfies
\begin{equation}\label{2.8+}
Z_t=\frac{dP^{B}}{d\tilde{P}^{B}}|\mathcal{F}^{B}_t=e^{\frac{\mu}{\kappa}B_t
+\frac{1}{2}(\frac{\mu}{\kappa})^2t}.
\end{equation}
Denote by $\mathcal{T}^B$ the set of all stopping times with respect
to the Brownian filtration $\mathcal{F}^{B}_t, t\geq{0}$ and let
$\mathcal{T}^{B}_{0T}$ be the set of all stopping times with values
in $[0,T]$. From Theorem 3.1 in \cite{Ki1} we obtain the fair price
of a game option in the BS model by
\begin{equation}\label{2.8++}
\cV^I=\inf_{\sigma\in\mathcal{T}^{B}_{0T}}\sup_{\tau\in\mathcal{T}^{B}_{0T}}
\tilde{E}^BQ^{B,I}(\sigma,\tau)
\end{equation}
where $\tilde{E}^{B}$ is the expectation with respect to
$\tilde{P}^{B}$.

Recall, (see, for instance, \cite{Sh}, Section 7.1) that a self
financing strategy $\pi$ with a (finite) horizon $T$ and an initial
capital $x$ is a process $\pi=\{(\beta_t,\gamma_t)\}_{t=0}^{T}$ of
pairs
 where $\beta_t$ and $\gamma_t$ are progressively measurable with respect
 to the filtration $\mathcal{F}^{B}_t$, $t\geq{0}$ and satisfy
\begin{equation}\label{2.9}
\int_{0}^{T}e^{rt}|\beta_t|dt<\infty\,\,\mbox{and}\,\,
\int_{0}^{T}(\gamma_t{S}^{B}_t)^2dt<\infty.
\end{equation}
The portfolio value $V^\pi_t$ for a strategy $\pi$ at time
$t\in{[0,T]}$
 is given by
\begin{equation}\label{2.10}
{V}^{\pi}_t=\beta_tb_t+\gamma_tS^{B}_t=x+\int_{0}^{t}\beta_ud
b_u+\int_{0}^{t} \gamma_ud{{S}^B_u}.
\end{equation}
Denote by $\tilde{V}^\pi_t=e^{-rt}V^\pi_t$ the discounted portfolio
value at time $t$. Then it is easy to see that (see, for instance,
\cite{Sh}),
\begin{equation}\label{2.11}
\tilde{V}^{\pi}_t=x+\int_{0}^{t}\gamma_ud{\tilde{S}^B_u} \ \
\mbox{and} \ \
\beta_t=(x+\int_{0}^{t}\gamma_ud{\tilde{S}^B_u}-\gamma_t\tilde{S}^{B}_t)/b_0.
\end{equation}
Observe that the discounted portfolio value depends only on the
process ${\{\gamma_t\}}_{t=0}^T$. Thus in order to determine a self
financing strategy it suffices to fix a process
${\{\gamma_t\}}_{t=0}^T$ and to obtain the process
${\{\beta_t\}}_{t=0}^T$ by (\ref{2.11}). A self financing strategy
$\pi$ is called \textit{admissible} if ${V}^{\pi}_t\geq{0}$ for all
$t\in{[0,{T}]}$ and the set of such strategies with an initial
capital $x$ will be denoted by $\mathcal{A}^B(x)$. Set also
$\mathcal{A}^B=\bigcup_{x\geq{0}}\mathcal{A}^B(x)$. A pair
$(\pi,\sigma)\in\mathcal{A}^B\times\mathcal{T}^B_{0T}$ of an
\textit{admissible} self financing strategy $\pi$ and of a stopping
time $\sigma$ will be called a hedge. For a hedge $(\pi,\sigma)$ the
shortfall risk is given by (see \cite{DK1}),
\begin{equation}\label{2.11+}
R^{I}(\pi,\sigma)=\sup_{\tau\in\mathcal{T}^B_{0T}}E^B[(Q^{B,I}(\sigma,\tau)-
\tilde{V}^{\pi}_{\sigma\wedge\tau})^+],
\end{equation}
which is the maximal possible expectation with respect to the
probability measure $P^B$ of the discounted shortfall. The shortfall
risks for a portfolio $\pi\in{\mathcal{A}^{B}}$ and for an initial
capital $x$ are given by
\begin{equation}\label{2.11++}
R^{I}(\pi)=\inf_{\sigma\in\mathcal{T}^{B}_{0T}}R^{I}(\pi,\sigma) \ \
\mbox{and} \ \ R^{I}(x)=\inf_{\pi\in\mathcal{A}^B(x)} R^{I}(\pi),
\end{equation}
respectively.

As in \cite{DK2} and \cite{Ki2} we consider a sequence of CRR
markets on a complete probability space such that for each
$n=1,2,...$ the bond prices $b^{(n)}_t$ at time $t$ are
\begin{equation}\label{2.12}
b^{(n)}_t=b_0e^{r[nt/T]T/n}=b_0(1+r_n)^{[nt/T]},\,\, r_n=e^{rT/n}-1
\end{equation}
and stock prices $S^{(n)}_t$ at time $t$ are given by the formulas
$S^{(n)}_t=S_0$ for $t\in{[0,T/n)}$ and
\begin{equation}\label{2.13}
S^{(n)}_t=S_0\exp(\sum_{k=1}^{[nt/T]}(\frac{rT}{n}+\kappa(\frac{T}{n})^{1/2}
\xi_k))=S_0\prod_{k=1}^{[nt/T]}(1+\rho^n_k)\,\,\mbox{if}\,\,
t\geq{T/n}
\end{equation}
where $\rho^n_k=\exp(\frac{rT}{n}+\kappa(\frac{T}{n})^{1/2}\xi_k)-1$
and $\xi_1,\xi_2,...$ are i.i.d. random variables taking values 1
and -1 with probabilities
$p^{(n)}=(\exp((\kappa-\frac{2\mu}{\kappa})\sqrt{\frac{T}{n}})+1)^{-1}$
and
$1-p^{(n)}=(\exp((\frac{2\mu}{\kappa}-\kappa)\sqrt{\frac{T}{n}})+1)^{-1}$,
respectively. Let $P^\xi_n={\{p^{(n)},1-p^{(n)}\}}^\infty$ be the
corresponding product probability measure on the space of sequences
$\Om_\xi=\{-1,1\}^\infty$. Namely, for each $n$ we consider a CRR
market with horizon $n$ on the probability space $(\Om_\xi,\,
P^\xi_n)$ with bond prices $b_m=b^{(n)}_{\frac {mT}n}$ and stoch
prices $S_m=S^{(n)}_{\frac {mT}n}$. We view
$S^{(n)}=S^{(n)}(\omega)$ as a random function on [0,T], so that
$S^{(n)}(\omega)\in{M[0,T]}$ takes the value
$S^{(n)}_t=S^{(n)}_t(\omega)$ at $t\in{[0,T]}$. For $k\leq{n}$
denote the discounted stock price at the moment $kT/n$ by
$\tilde{S}^{(n)}_{\frac{kT}{n}}=(1+r_n)^{-k}S^{(n)}_{\frac{kT}{n}}$.
Let $\mathcal{F}^{\xi}_k=\sigma{\{\xi_1,...,\xi_k}\}$ and
$\mathcal{F}^{\xi}=\bigcup_{k\geq{1}}\mathcal{F}^{\xi}_k$. Denote by
$\mathcal{T}^{\xi}$ the set of all stopping times with respect to
the filtration $\mathcal{F}^{\xi}_k$ and let
$\mathcal{T}^{\xi}_{0n}$ be the set of all stopping times with
values in ${\{0,1,...,n\}}$. Similarly to (\ref{2.6}), given an open
interval $I$ introduce a stopping time (with respect to the
filtration ${\{\mathcal{F}^{\xi}_k\}}_{k=0}^\infty)$
\begin{equation}\label{2.13+}
\tau^{(n)}_I=\min\{k\geq{0}|S^{(n)}_{\frac{kT}{n}}\notin{I}\}
\end{equation}
together with barrier options having the payoffs
\begin{equation}\label{2.13++}
Y^{I,n}_k=F_{\frac{kT}{n}}(S^{(n)})\mathbb{I}_{k<\tau^{(n)}_I} \ \
\mbox{and} \ \
X^{I,n}_k=G_{\frac{kT}{n}}(S^{(n)})\mathbb{I}_{k<\tau^{(n)}_I}.
\end{equation}
The corresponding discounted payoff function is given by
\begin{equation} \label{2.13+++}
Q^{I,n}(s,k)=\tilde{X}^{I,n}_s\mathbb{I}_{s<k}+
\tilde{Y}^{I,n}_k\mathbb{I}_{k\leq{s}},\ \ k,s\leq{n}
\end{equation}
where $\tilde{X}^{I,n}_k=(1+r_n)^{-k}X^{I,n}_k$ and
$\tilde{Y}^{I,n}_k=(1+r_n)^{-k}Y^{I,n}_k$ are the discounted
payoffs. Let $\tilde{P}^{\xi}_n$ be a probability measure on the
$\Omega_{\xi}$ such that $\xi_1,\xi_2...$ is a sequence of i.i.d.
random variables taking on the values $1$ and $-1$ with
probabilities
$\tilde{p}^{(n)}=(\exp(\kappa\sqrt{\frac{T}{n}})+1)^{-1}$ and
$1-\tilde{p}^{(n)}=(\exp(-\kappa\sqrt{\frac{T}{n}})+1)^{-1}$,
respectively (with respect to $\tilde{P}^{\xi}_n$). Observe that for
any $n$ the process ${\{\tilde{S}^{(n)}_{\frac{mT}{n}}\}}_{m=0}^{n}$
is a martingale with respect to $\tilde{P}^{\xi}_n$, and so we
conclude that $\tilde{P}^{\xi}_n$ is the unique martingale measure
for the above CRR markets. Thus from Theorem 2.1 in \cite{Ki1} it
follows that the fair price of the game option in the $n$--step CRR
market is given by
\begin{equation}\label{2.13++++}
\cV^{I}_n=\min_{\zeta\in\mathcal{T}^{\xi}_{0n}}\max_{\eta\in
\mathcal{T}^{\xi}_{0n}}\tilde{E}^{\xi}_nQ^{I,n}(\zeta,\eta).
\end{equation}
where $\tilde{E}^{\xi}_n$ is the expectation with respect to
$\tilde{P}^{\xi}_n$. The following theorem provides an estimate for
the error term in approximations of the fair price of a knock--out
game option in the BS model by fair prices of the sequence of knock
out game options in the CRR markets defined above. This result is a
generalization of Theorem 2.1 in \cite{Ki2} which deals with regular
game options.
\begin{thm}\label{thm2.1}
There exists a constant $C_1$ such that for any open interval $I$
and $n\in\mathbb{N}$,
\begin{equation}\label{2.13+++++}
|\cV^I-\cV^{I}_n|\leq{C_1n^{-\frac{1}{4}}(\ln{n})^{\frac{3}{4}}}.
\end{equation}
\end{thm}

Denote by $\mathcal{A}^{\xi,n}(x)$ the set of all
\textit{admissible} self financing strategies with an initial
capital $x$ and set
$\mathcal{A}^{\xi,n}=\bigcup_{x\geq{0}}\mathcal{A}^{\xi,n}(x)$.
Recall (see \cite{SKKM}) that a self financing strategy $\pi$ with
an initial capital $x$ and a horizon $n$ is a sequence
$(\pi_1,...,\pi_n)$ of pairs $\pi_k=(\beta_k,\gamma_k)$ where
$\beta_k,\gamma_k$ are $\mathcal{F}^{\xi}_{k-1}$-measurable random
variables representing the number of bond and stock units,
respectively, at time $k$. Thus the portfolio value $V^\pi_k$,
$k=0,1,...,n$ is given by
\begin{equation}\label{2.14}
V^\pi_0=x,\,\,
V^\pi_k=\beta_kb^{(n)}_{\frac{kT}{n}}+\gamma_kS^{(n)}_{\frac{kT}{n}},
\ \ 1\leq{k}\leq{n}.
\end{equation}
Denote by $\tilde{V}^\pi_k=(1+r_n)^{-k}V^\pi_k$ the discounted
portfolio value at time $k$. Since $\pi$ is self financing then
\begin{equation}\label{2.15}
\beta_kb^{(n)}_{\frac{kT}{n}}+\gamma_kS^{(n)}_{\frac{kT}{n}}=
\beta_{k+1}b^{(n)}_{\frac{kT}{n}}+\gamma_{k+1}S^{(n)}_{\frac{kT}{n}},
\end{equation}
and so (see \cite{Sh} and \cite{SKKM}),
\begin{equation}\label{2.16}
\tilde{V}^{\pi}_{k}=x+\sum_{i=0}^{k-1}{\gamma_{i+1}
(\tilde{S}^{(n)}_{\frac{(i+1)T}{n}}-\tilde{S}^{(n)}_{\frac{iT}{n}})}
\ \ \mbox{and} \ \
\beta_k=(x+\sum_{i=0}^{k-1}{\gamma_{i+1}(\tilde{S}^{(n)}_{\frac{(i+1)T}{n}}-
\tilde{S}^{(n)}_{\frac{iT}{n}})}-\gamma_k\tilde{S}^{(n)}_{\frac{kT}{n}})/b_0.
\end{equation}
Hence, as before, in order to determine a self financing strategy it
suffices to introduce a process ${\{\gamma_k\}}_{k=0}^n$ and to
obtain the process ${\{\beta_k\}}_{k=0}^n$ by (\ref{2.16}). We call
a self financing strategy $\pi$ \textit{admissible} if
$V^\pi_k\geq{0}$ for any $k\leq{n}$. A hedge with an initial capital
$x$ is an element in the set
$\mathcal{A}^{\xi,n}(x)\times\mathcal{T}^\xi_{0n}$. The definitions
for the shortfall risks in the CRR markets are similar to the
definitions in the BS model. Thus for the $n$--step CRR market the
shortfall risks are given by
\begin{eqnarray}\label{2.17}
&R^I_n(\pi,\sigma)=\max_{\tau\in\mathcal{T}^\xi_{0n}}E^{\xi}_n(Q
^{I,n}(\sigma,\tau)-\tilde{V}^{\pi}_{\sigma\wedge\tau})^+, \\
&R^I_n(\pi)=\min_{\sigma\in\mathcal{T}^{\xi}_{0n}}R^I_n(\pi,\sigma)
\ \ \mbox{and} \ \ R^I_n(x)=
\inf_{\pi\in\mathcal{A}^{\xi,n}(x)}R^I_n(\pi), \nonumber
\end{eqnarray}
where $E^{\xi}_n$ is the expectation with respect to $P^{\xi}_n$.
\begin{thm}\label{thm2.2}
For any open interval $I$
\begin{equation}\label{2.18}
lim_{n\rightarrow\infty}R^I_n(x)=R^I(x).
\end{equation}
Furthermore, there exists a constant $C_2$ (which does not depend on
the interval $I$) such that for any $n\in\mathbb{N}$
\begin{equation}\label{2.19}
R^I(x)\leq{R^I_n(x)+C_2n^{-\frac{1}{4}}(\ln{n})^{3/4}}.
\end{equation}
\end{thm}

The above result says that the shortfall risk $R^{I}(x)$ for double
barrier options in the BS model can be approximated by a sequence of
shortfall risks with an initial capital $x$ for a similar options in
the CRR markets and it provides also a one sided error estimate of
the approximation. This result is a generalization of Theorem 2.1 in
\cite{DK2} which deals with regular game options.

In order to compare the option prices and the shortfall risks in the
BS model with the corresponding quantities in the CRR markets, we
will use (a trivial form of) the Skorohod type embedding (see
\cite{B}) which allows us to consider the above objects on the same
probability space. Thus, define recursively
\begin{equation*}
\theta^{(n)}_0=0, \ \theta^{(n)}_{k+1}=\inf{\{t>\theta^{(n)}_k
:|B^{*}_t-B^{*}_{\theta^{(n)}_k}|=\sqrt{\frac{T}{n}}\}},
\end{equation*}
where, recall, $B^*_t=(\frac{\mu}{\kappa}-\frac{\kappa}{2})t+B_t$.
Using the same arguments as in \cite {Ki2} we obtain that for each
of the measures $P^{B},\tilde{P}^{B}$, the sequence
$\theta^{(n)}_k-\theta^{(n)}_{k-1}$, $k=1,2,...$ is a sequence of
i.i.d. random variables such that
$(\theta^{(n)}_{k+1}-\theta^{(n)}_k,
B^{*}_{\theta^{(n)}_{k+1}}-B^{*}_{\theta^{(n)}_k})$ are independent
of $\mathcal{F}^{B}_{\theta^{(n)}_k}$. Employing the exponential
martingale $\exp((\kappa-\frac{2\mu}{\kappa})B^{*}_t)$ for the
probability $P^B$ we obtain that
$E^{B}\exp((\kappa-\frac{2\mu}{\kappa})B^{*}_{\theta^{(n)}_1}) =1$
concluding that $B^{*}_ {\theta^{(n)}_1}=\sqrt{\frac{T}{n}}$ or
$-\sqrt{\frac{T}{n}}$ with probability $p^{(n)}$ or $1-p^{(n)}$,
respectively. Using the martingale
$\tilde{S}^B_t=S_0\exp({\kappa}B^{*}_t)$ for the probability $\tilde
P^B$ we obtain
$\tilde{E}^{B}\exp({\kappa}B^{*}_{\theta^{(n)}_1})=1$, and so
$B^{*}_ {\theta^{(n)}_1}=\sqrt{\frac{T}{n}}$ or
$-\sqrt{\frac{T}{n}}$ with probability $\tilde{p}^{(n)}$ or
$1-\tilde{p}^{(n)}$ respectively.

The Skorohod embedding also allows us to define mappings (introduced
in \cite{DK2} and \cite{Ki2}) which map hedges in CRR markets to
hedges in the BS model and which will play a decisive role in
Theorems \ref{thm2.3} and \ref{thm2.4} below. For readers
convenience we review the definitions. For any $n\in\mathbb{N}$ set
$\bg^{(n)}_i=B^{*}_{\theta^{(n)}_i} -B^{*}_{\theta^{(n)}_{i-1}},\,
i=1,2,...$ and following \cite{Ki2} introduce for each $k=1,2,...$
the finite $\sigma$--algebra
$\mathcal{G}^{B,n}_k=\sigma\{\bg^{(n)}_1,...,\bg^{(n)}_k\}$ with
$\mathcal{G}^{B,n}_0={\{\emptyset,\Omega_B\}}$.
 Let $\mathcal{S}^{B,n}_{0,n}$ be the set of all
stopping times with respect to the filtration $\mathcal{G}^{B,n}_k,
k=0,1,2...$ with values in ${\{0,1...,n\}}$. Observe that for any
$n$ we have a natural bijection
$\Pi_n:L^{\infty}(\mathcal{F}^{\xi}_n,P^{\xi}_n)\rightarrow{L^{\infty}
(\mathcal{G}^{B,n}_n,P^B})$ which is given by $\Pi_n(Z)=\tilde Z$ so
that if $Z=f(\xi_1,...,\xi_n)$ for a function $f$ on $\{-1,1\}^n$
then $\tilde
Z=f(\sqrt{\frac{n}{T}}\bg^{(n)}_1,...,\sqrt{\frac{n}{T}}\bg^{(n)}_n)$.
Notice that if we restrict $\Pi_n$ to
$L^{\infty}(\mathcal{F}^{\xi}_k,P^{\xi}_n)$ we obtain a bijection
$\Pi_{n,k}:L^{\infty}(\mathcal{F}^{\xi}_k,P^{\xi}_n)\rightarrow{L^{\infty}
(\mathcal{G}^{B,n}_k,P^B})$ and if we restrict $\Pi_n$ to
$\mathcal{T}^{\xi}_{0n}$ we get a bijection
$\Pi_n:\mathcal{T}^{\xi}_{0n}\rightarrow\mathcal{S}^{B,n}_{0,n}$. In
addition to the set $\mathcal{S}^{B,n}_{0,n}$ consider also the set
$\mathcal{T}^{B,n}_{0,n}$ of stopping times with respect to the
filtration ${\{\mathcal{F}^{B}_{\theta^{(n)}_k}\}}_{k=0}^n$ with
values in ${\{0,1,...n\}}$. Clearly
$\mathcal{S}^{B,n}_{0,n}\subset\mathcal{T}^{B,n}_{0,n}$. Next, we
define a function $\phi_n:\mathcal{T}^{\xi}_{0n}\rightarrow
\mathcal{T}^{B}_{0T}$ which maps stopping times in CRR markets to
stopping times in the BS model by
\begin{equation}\label{2.20}
\phi_n(\sigma)=T\wedge\theta^{(n)}_{\Pi_n(\sigma)} \ \ \mbox{if} \ \
\Pi_n(\sigma)<n  \ \ \mbox{and} \ \ \phi_n(\sigma)=T \ \ \mbox{if} \
\  \Pi_n(\sigma)=n.
\end{equation}
It is easy to see that $\phi_n(\sigma)\in\mathcal{T}^{B}_{0T}$ (see
(2.28) in \cite{DK2}). For each $n$ and $x>0$ let
$\mathcal{A}^{B,n}(x)$ be the set of all \textit{admissible} self
financing strategies with an initial capital $x$ in the BS model
which can be managed only on the set
${\{0,\theta^{(n)}_1,...,\theta^{(n)}_n\}}$, such that the
discounted portfolio value remains constant after the moment
$\theta^{(n)}_n$ and set
$\mathcal{A}^{B,n}=\bigcup_{x\geq{0}}\mathcal{A}^{B,n}(x)$. Thus if
$\pi={\{(\beta_t,\gamma_t)\}}_{t=0}^\infty\in\mathcal{A}^{B,n}$ then
$\beta_t=\beta_{\theta^{(n)}_k}$ and
$\gamma_t=\gamma_{\theta^{(n)}_k}$ for any $k<n$ and
$t\in{[\theta^{(n)}_k,\theta^{(n)}_{k+1})}$.
 Furthermore, in order to keep the
discounted portfolio constant after $\theta^{(n)}_n$ the investor
should sell all his stocks at the moment $\theta^{(n)}_n$ and buy
bonds for all money, and so $\gamma_t=0$ for $t\geq\theta^{(n)}_n$.
From (\ref{2.11}) it follows that for
$\pi={\{(\beta_t,\gamma_t)\}}_{t=0}^\infty\in\mathcal{A}^{B,n}$ the
corresponding discounted portfolio value is given by
\begin{equation}\label{2.21}
\tilde{V}^{\pi}_t=\tilde{V}^{\pi}_{\theta^{(n)}_k}+\gamma_{\theta^{(n)}_k}
(\tilde{S}^{B}_{t}-\tilde{S}^{B}_{\theta^{(n)}_{k}}),\,
 t\in{[\theta^{(n)}_k,\theta^{(n)}_{k+1}]}\, \ \mbox{and}\,\
\tilde{V}^{\pi}_t=\tilde{V}^{\pi}_{\theta^{(n)}_n},\,
t>\theta^{(n)}_n.
\end{equation}
Finally, we define a function
$\psi_n:\mathcal{A}^{\xi,n}(x)\rightarrow\mathcal{A}^{B,n}(x)$ which
maps \textit{admissible} self financing strategies in the CRR
$n$--step model to the set of the above self financing strategies in
the BS model. For
$\pi=\{(\beta_k,\gamma_k)\}_{k=1}^n\in{\mathcal{A}^{\xi,n}(x)}$
define $\psi_n(\pi)\in\mathcal{A}^{B,n}(x)$ by
\begin{eqnarray}\label{2.22}
&\tilde{V}^{\psi_n(\pi)}_t=\tilde{V}^{\psi_n(\pi)}_{\theta^{(n)}_k}+
\Pi_n(\gamma_{k+1})
(\tilde{S}^{B}_{t}-\tilde{S}^{B}_{\theta^{(n)}_{k}}),\,
 t\in{[\theta^{(n)}_k,\theta^{(n)}_{k+1}]},\\
&\mbox{and} \ \
\tilde{V}^{\psi_n(\pi)}_t=\tilde{V}^{\psi_n(\pi)}_{\theta^{(n)}_n},
\ \  t>\theta^{(n)}_n. \nonumber
\end{eqnarray}
Observe that $\Pi_n(\tilde{S}^{(n)}_{\frac{kT}{n}})=
\tilde{S}^{B}_{\theta^{(n)}_k}$ for any $k\leq{n}$, and so we obtain
from (\ref{2.16}) and (\ref{2.21}) that
$\tilde{V}^{\psi_n(\pi)}_{\theta^{(n)}_k}=\Pi_n(\tilde{V}^{\pi}_k)\geq{0}$
for any $k\leq{n}$. Since the process $\tilde{V}^{\psi_n(\pi)}_t$,
$t\geq{0}$ is a martingale with respect to the martingale measure
$\tilde{P}^B$ and it remains constant for $t\geq\theta^{(n)}_n$ we
get that the portfolio $\psi_n(\pi)$ is \textit{admissible}
concluding that $\psi_n(\pi)\in\mathcal{A}^{B,n}(x)$, as required.
Clearly, if we restrict the portfolio $\psi_n(\pi)$ to the interval
$[0,T]$ we can consider $\psi_n(\pi)$ as an element in
$\mathcal{A}^{B}(x)$.

Let $I=(L,R)$ be an open interval and set
$L_n=L\exp(-n^{-\frac{1}{3}})$, $R_n=R\exp(n^{-\frac{1}{3}})$ (with
 $R_n=\infty$ if $R=\infty$) and $I_n=(L_n,R_n)$. Let
$(\pi,\sigma)\in\mathcal{A}^{\xi,n}(\cV^{I_n}_n)\times\mathcal{T}^\xi_{0n}$
be a perfect hedge for a double barrier option in the $n$--step CRR
market with the barriers $L_n,R_n$, i.e. a hedge which satisfies
$\tilde{V}^{\pi}_{\sigma\wedge{k}}\geq{Q^{I,n}(\sigma,k)}$ for any
$k\leq{n}$. In general the construction of perfect hedges for game
options in CRR markets can be done explicitly (see \cite{Ki1},
Theorem 2.1). The following result shows that if we embed the
perfect hedge $(\pi,\sigma)$ into the BS model we obtain a hedge
with small shortfall risk for the barrier option with barriers
$L,R$.
\begin{thm}\label{thm2.3}
Let $I=(L,R)$ be an open interval. For any $n$ let
$(\pi^p_n,\sigma^p_n)\in\mathcal{A}^{\xi,n}(\cV^{I_n}_n)\times
\mathcal{T}^\xi_{0n}$ be a perfect hedge for a double barrier option
in the $n$--step CRR market with the barriers $L_n,R_n$. Define
$(\pi^B_n,\sigma^B_n)\in\mathcal{A}^{B}(\cV^{I_n}_n)\times\mathcal{T}^{B}_{0T}$
by $\pi^B_n=\psi_n(\pi^p_n)$ and $\sigma^B_n=\phi_n(\sigma^p_n)$.
There exists a constant $C_3$ (which does not depend on the interval
$I$) such that for any $n,$
\begin{equation}\label{2.23}
R^I(\pi^B_n,\sigma^B_n)\leq{C_3n^{-\frac{1}{4}}(\ln{n})^{\frac{3}{4}}}.
\end{equation}
\end{thm}
We will see (as a conclusion of (\ref{3.19}) and Theorem
\ref{thm2.1}) that there exists a constant $\tilde{C}$ (which does
not depend on the interval $I$) such that
$|\cV^I-\cV^{I_n}_n|\leq{\tilde{C}n^{-\frac{1}{4}}(\ln{n})^{\frac{3}{4}}}$
for any $n$. Since the above term is small then in practice a seller
of a double barrier
 game option with the barriers $L,R$ can invest the amount
$\cV^{I_n}_n$ in the portfolio and use the above hedges facing only
small shortfall risk.

Next, consider an investor in the BS market whose initial capital
$x$ which is less than the option price $\cV^I$. A hedge
$(\pi,\sigma)\in{\mathcal{A}^B(x)\times\mathcal{T}^B_{0T}}$ will be
called $\ve$-optimal if $R^{I}(\pi,\sigma)\leq{R^{I}(x)+\ve}$. For
$\ve=0$ the above hedge is called an optimal hedge. For the CRR
markets we have an analogous definitions. In the next section we
will follow \cite{DK1} and construct optimal hedges
$(\pi_n,\sigma_n)\in\mathcal{A}^{\xi,n}(x)\times\mathcal{T}^{\xi}_{0n}$
for double barrier options in the $n$--step CRR markets with
barriers $L_n,R_n$. By embedding this hedges into the BS model we
obtain a simple representation of $\ve$--optimal hedges for the the
BS model.
\begin{thm}\label{thm2.4}
For any $n$ let
$(\pi_n,\sigma_n)\in\mathcal{A}^{\xi,n}(x)\times\mathcal{T}^{\xi}_{0n}$
be the optimal hedge which is given by (\ref{3.15}) with $H=I_n$.
Then
\begin{equation}\label{2.24}
lim_{n\rightarrow\infty}R^{I}(\psi_n(\pi_n),\phi_n(\sigma_n))=R^I(x).
\end{equation}
\end{thm}
In Section \ref{sec5} we formulate and prove corresponding results
for knock--in Israely style barrier options.

\section{Auxiliary lemmas}\label{sec3}\setcounter{equation}{0}

First we introduce the machinery which enables us to reduce
optimization of the shortfall risk to optimal stopping problems for
Dynkin's games with appropriately chosen payoff processes so that on
the next stage we will be able to employ the Skorohod embedding in
order to compare values of the corresponding discrete and continuous
time games. This machinery was used in \cite{DK2} for similar
purposes in the case of regular game options. For any $n$ set
$a^{(n)}_1=e^{\kappa\sqrt{\frac{T}{n}}}-1,\,$
$a^{(n)}_2=e^{-\kappa\sqrt{\frac{T}{n}}}-1$ and observe that for any
$m\leq{n}$ the random variable
\[
\frac{\tilde{S}^{(n)}_{\frac{mT}{n}}}{\tilde{S}^{(n)}_{(m-1)T/n}}-1=
\exp(\kappa(\frac{T}{n})^{1/2}\xi_m)-1
\]
takes on only the values $a^{(n)}_1,a^{(n)}_2$. For each $y>0$ and
$n\in\bbN$ introduce the closed interval $K_n(y)=\big [-\frac
y{a_1^{(n)}}, -\frac y{a_2^{(n)}}\big ]$ and for $0\leq{k}<n$ and a
given positive $\mathcal{F}^{\xi}_k$-measurable random variable $X$
define
\begin{eqnarray}\label{3.1}
&\mathcal{A}^{\xi,n}_k(X)=\{Y|Y=X+\alpha(\exp(\kappa(\frac{T}{n})^{1/2}
\xi_{k+1})-1)\,\,\mbox{for some}\\
&\mathcal{F}^{\xi}_k-\mbox{measurable}\,\,\alpha\in K_n(X)\}.
\nonumber
\end{eqnarray}
Notice that if for $\pi={\{(\beta_k,\gamma_k)\}}_{k=1}^n$,
$\tilde{V}^\pi_k=X$ and $\tilde{V}^\pi_{k+1}=Y$ then by
(\ref{2.16}),
$Y=X+\alpha(\exp(\kappa(\frac{T}{n})^{1/2}\xi_{k+1})-1)$ where
$\alpha=\gamma_{k+1}\tilde{S}^{(n)}_{\frac{kT}{n}}$ is
$\mathcal{F}^{\xi}_k$-measurable. Since we allow only nonnegative
portfolio values, and so $Y\geq{0}$ which must be satisfied for all
possible values of $\exp(\kappa(\frac{T}{n})^{1/2}\xi_{k+1})-1$ we
conclude in view of independency of $\alpha$ and $\xi_{k+1}$ that
$\mathcal{A}^{\xi,n}_k(X)$ is the set of all possible discounted
portfolio values at the time $k+1$ provided that the discounted
portfolio value at the time $k$ is $X$.

Let $H$ be an open interval. For any $\pi\in\mathcal{A}^{\xi,n}$
define a sequence of random variables ${\{W^{H,\pi}_k\}}_{k=0}^n$
\begin{eqnarray}\label{3.2}
&W^{H,\pi}_n=(\tilde{Y}^{H,n}_n-\tilde{V}^\pi_n)^+, \ \ \mbox{and} \
\
W^{H,\pi}_k=\min\bigg((\tilde{X}^{H,n}_k-\tilde{V}^{\pi}_k)^{+},\\
&\max\bigg((\tilde{Y}^{H,n}_k-\tilde{V}^{\pi}_k)^{+},E^{\xi}_n(W^{H,\pi}_{k+1}
|\mathcal{F}^{\xi}_k)\bigg)\bigg) \ \ k<n. \nonumber
\end{eqnarray}
Applying the results for Dynkin's games from \cite {YO} for the
processes
\[
{\{(\tilde{X}^{H,n}_k-\tilde{V}^{\pi}_k)^{+}\}}_{k=0}^n,
{\{(\tilde{Y}^{H,n}_k-\tilde{V}^{\pi}_k)^{+}\}}_{k=0}^n
\]
we obtain
\begin{equation}\label{3.3}
W^{H,\pi}_0=\min_{\sigma\in{\mathcal{T}^\xi_{0n}}}
\max_{\tau\in{\mathcal{T}^\xi_{0n}}}E^{\xi}_n(Q
^{H,n}(\sigma,\tau)-\tilde{V}^{\pi}_{\sigma\wedge\tau})^+=R^H_n(\pi)=
R^H_n(\pi,\sigma(H,\pi))
\end{equation}
where
\begin{equation}\label{3.4}
\sigma(H,\pi)=\min{\{k|(\tilde{X}^{H,n}_k-\tilde{V}^{\pi}_k)^{+}=
W^{H,\pi}_k\}}\wedge{n}.
\end{equation}

On the Brownian probability space set
\begin{equation}\label{3.5}
S^{B,n}_t=S_0, \ \ t\in [0,T/n] \ \ \mbox{and} \ \
S^{B,n}_t=S_0\exp(\sum_{k=1}^{[nt/T]}(\frac{rT}{n}+
\kappa\bg^{(n)}_k)), \ \  t\in{[T/n,T]}.
\end{equation}
Define
\begin{equation}\label{3.6}
\tau^{B,n}_H=\min\{k\geq{0}|S^{B,n}_{\frac{kT}{n}}\notin{H}\}.
\end{equation}
Clearly $\tau^{B,n}_H$ is a stopping time with respect to the
filtration $\mathcal{G}^{B,n}_k$, $k\geq{0}$. Consider the new
payoffs
$Y^{B,H,n}_k=F_{\frac{kT}{n}}(S^{B,n})\mathbb{I}_{k<\tau^{B,n}_H}$
and
$X^{B,H,n}_k=G_{\frac{kT}{n}}(S^{B,n})\mathbb{I}_{k<\tau^{B,n}_H}$,
$k\leq{n}$. The corresponding payoff function is given by
\begin{equation}\label{3.7}
Q^{B,H,n}(k,l)=\tilde{X}^{B,H,n}_k\mathbb{I}_{k<l}+
\tilde{Y}^{B,H,n}_l\mathbb{I}_{l\leq{k}}, \ \ k,l\leq{n}
\end{equation}
where $\tilde{Y}^{B,H,n}_k=(1+r_n)^{-k}Y^{B,H,n}_k$ and
$\tilde{X}^{B,H,n}_k=(1+r_n)^{-k}X^{B,H,n}_k$ are the discounted
payoffs. For any $n$ we consider now hedges which are elements in
$\mathcal{A}^{B,n}\times\mathcal{T}^{B,n}_{0,n}$. Given a positive
$\mathcal{F}^{B}_{\theta^{(n)}_k}$--measurable random variable $X$
define $\mathcal{A}^{B,n}_k(X)$ by (\ref{3.1}) with
$\sqrt{\frac{T}{n}}\xi_{k+1}$ and $\cF^\xi_k$ replaced by
$\bg^{(n)}_{k+1}$ and $\cF^B_{\te_k^{(n)}}$, respectively. By
(\ref{2.21}) we conclude similarly to the above that
$\mathcal{A}^{B,n}_k(X)$ consists of all possible discounted values
at the time $\theta^{(n)}_{k+1}$ of portfolios managed only at
embedding times $\{\te^{(n)}_i\}$ with the discounted stock
evolution $\tilde S_t^B$, provided the discounted portfolio value at
the time $\theta^{(n)}_k$ is $X$.

Next, define the shortfall risk by
\begin{eqnarray}\label{3.8}
&R^{B,H}_n(\pi,\zeta)=\sup_{\eta\in\mathcal{T}^{B,n}_{0n}}E^B(Q^{B,H,n}
(\zeta,\eta)-\tilde{V}^{\pi}_{\theta^{(n)}_{\zeta\wedge\eta}}) ^+,\\
&R^{B,H}_n(\pi)=\inf_{\zeta\in\mathcal{T}^{B,n}_{0n}}R^{B,H}_n(\pi,\zeta)
\ \ \mbox{and} \ \
R^{B,H}_n(x)=\inf_{\pi\in\mathcal{A}^{B,n}(x)}R^{B,H}_n(\pi).
\nonumber
\end{eqnarray}
For any $\pi\in{\mathcal{A}^{B,n}}$ define a sequence of random
variables ${\{U^{H,\pi}_k\}}_{k=0}^n$,
\begin{eqnarray}\label{3.9}
&U^{H,\pi}_n=(\tilde{Y}^{B,H,n}_n-\tilde{V}^\pi_{\theta^{(n)}_n})^+
\ \mbox{and} \ U^{H,\pi}_k=\min\bigg((\tilde{X}^{B,H,n}_k-
\tilde{V}^\pi_{\theta^{(n)}_k})^+,\\
&\max\bigg((\tilde{Y}^{B,H,n}_k-\tilde{V}^{\pi}_{\theta^{(n)}_k})^+,
E^{B}(U^{H,\pi}_{k+1}|\mathcal{F}^{B}_{\theta^{(n)}_{k}})\bigg)\bigg),
\ \ k<n \nonumber
\end{eqnarray}
and a stopping time
\begin{equation}\label{3.10}
\zeta(H,\pi)=\min{\{k|(\tilde{X}^{B,H,n}_k-
\tilde{V}^{\pi}_{\theta^{(n)}_k})^{+}=U^{H,\pi}_{k}\}}\wedge{n}.
\end{equation}
Again, using the results on Dynkin's games from \cite {YO} for the
adapted (with respect to the filtration
$\mathcal{F}^{B}_{\theta^{(n)}_k}$, $k\geq{0}$) payoff processes
${\{(\tilde{Y}^{B,H,n}_k-
\tilde{V}^\pi_{\theta^{(n)}_k})^+\}}_{k=0}^n,\,$
${\{(\tilde{X}^{B,H,n}_k-\tilde{V}^\pi_{\theta^{(n)}_k})^+\}}_{k=0}^n$
we obtain that
\begin{eqnarray}\label{3.11}
&U^{H,\pi}_0=\inf_{\zeta\in\mathcal{T}^{B,n}_{0n}}\sup_{\eta
\in\mathcal{T}^{B,n}_{0,n}}E^B(Q^{B,H,n}(\zeta,
\eta)-\tilde{V}^{\pi}_{\theta^{(n)}_{\zeta\wedge\eta}})^+\\
&=R^{B,H}_n(\pi,\zeta(H,\pi))=R^{B,H}_n(\pi).\nonumber
\end{eqnarray}
For ${k}\leq{n}$ and $x_1,...,x_k\in{\mathbb{R}}$, consider the
function $\psi^{x_1,...,x_k}\in{M[0,\frac{kT}{n}]}$ given by
\begin{eqnarray*}
&\psi^{x_1,...,x_k}(t)=S_0\exp(\frac{rjT}{n}+\kappa\sum_{i=1}^jx_i),\,
t\in{[jt/n,(j+1)T/n)}, \ 1\leq{j}\leq{k}  \\
&\mbox{and} \ \ \psi^{x_1,...,x_k}(0)=S_0, \ \ t\in{[0,T/n)},
\end{eqnarray*}
there exist $f^n_k,g^n_k:\mathbb{R}^k\rightarrow \mathbb{R}$ such
that for any $x_1,...,x_k\in{\mathbb{R}}$,
\begin{eqnarray*}
&f^n_k(x_1,...,x_k)=(1+r_n)^{-k}F_{\frac{kT}{n}}(\psi^{x_1,...,x_k})=
e^{-rkT/n}F_{\frac{kT}{n}}(\psi^{x_1,
...,x_k}),\\
&\mbox{and} \ \
g^n_k(x_1,...,x_k)=(1+r_n)^{-k}G_{\frac{kT}{n}}(\psi^{x_1,...,x_k})=
e^{-rkT/n}G_{\frac{kT}{n}}(\psi^{x_1, ...,x_k}).
\end{eqnarray*}
Set
\begin{equation*}
q^{H,n}_k(x_1,...,x_k)=\mathbb{I}_{{[\min_{0\leq{i}\leq{k}}
\psi^{x_1,...,x_i}(\frac{iT}{n}),\max_{0\leq{i}\leq{k}}\psi^{x_1,...,x_i}
(\frac{iT}{n})] \subset{H}}}.
\end{equation*}
Observe that for the above functions,
\begin{eqnarray}\label{3.12}
&\tilde{Y}^{B,H,n}_k=f^n_k(\bg^{(n)}_1,...,\bg^{(n)}_k)
q^{H,n}_k(\bg^{(n)}_1,...,\bg^{(n)}_k),
\ \ \tilde{X}^{B,H,n}_k=g^n_k(\bg^{(n)}_1,...,\bg^{(n)}_k)\\
&\times q^{H,n}_k(\bg^{(n)}_1,...,\bg^{(n)}_k), \ \
\tilde{Y}^{(n)}_k=f^n_k(\sqrt{\frac{T}{n}}\xi_1,...,
\sqrt{\frac{T}{n}}\xi_k)q^{H,n}_k(\sqrt{\frac{T}{n}}\xi_1,...,
\sqrt{\frac{T}{n}}\xi_k)\nonumber\\
& \mbox{and} \ \
\tilde{X}^{(n)}_k=g^n_k(\sqrt{\frac{T}{n}}\xi_1,...,\sqrt{\frac{T}{n}}\xi_k)
q^{H,n}_k(\sqrt{\frac{T}{n}}\xi_1,...,\sqrt{\frac{T}{n}}\xi_k).
\nonumber
\end{eqnarray}
Finally, define a sequence ${\{J^{H,n}_k\}}_{k=0}^n$ of functions
$J^{H,n}_k:[0,\infty)\times\mathbb{R}^k\rightarrow \mathbb{R}$ by
the following backward recursion
\begin{eqnarray}\label{3.13}
&J^{H,n}_n(y,u_1,u_2...,u_n)=(f^n_n(u_1,...,u_n)q^{H,n}_n(u_1,...,u_n)-y)^+
\  \ \mbox{and}\\
&J^{H,n}_k(y,u_1,...,u_k)=\min\bigg((g^n_k(u_1,...,u_k)
q^{H,n}_k(u_1,...,u_k)-y)^+\nonumber,
\max\bigg((f^n_k(u_1,...,u_k)\nonumber\\
&\times q^{H,n}_k(u_1,...,u_k)-y)^+,\inf_{u\in K_n(y)}\big(p^{(n)}
J^{H,n}_{k+1}(y+ua^{(n)}_1,u_1,...,u_k,\sqrt{\frac{T}{n}})+\nonumber\\
&(1-p^{(n)})J^{H,n}_{k+1}(y+ua^{(n)}_2,u_1,...,u_k,-\sqrt{\frac{T}{n}})\big)
\bigg)\bigg)  \ \mbox{for}  \ k=n-1,n-2,...,0.\nonumber
\end{eqnarray}
Similarly to \cite{DK2} this dynamical programming relations will
enable us to compute shortfall risks defined in (\ref{2.17}) and
(\ref{3.8}).
\begin{lem}\label{lem3.1}
The function $J^{H,n}_k(y,u_1,...,u_k)$ is continuous and decreasing
with respect to $y$ for any $n$, ${k}\leq{n}$ and an open interval
$H$.
\end{lem}
\begin{proof}
The proof is the same as the proof of Lemma 3.2 in \cite{DK2}, just
replace $J^{H,n}_k$ by $J^{n}_k$.
\end{proof}\

For a given closed interval $K=[a,b]$ and a function
$f:K\times\mathbb{R}^k\to\bbR$ such that $f(\cdot,v)$ is continuous
for all $v\in\mathbb{R}^k$ define
$argmin_{{a}\leq{u}\leq{b}}f(u,v)=\min\{w\in{K}|f(w,v)=\min_{\beta\in{K}}
f(\beta,v)\}$. Lemma \ref{lem3.1} enables us to define the following
functions
\begin{eqnarray}\label{3.14}
&h^{H,n}_k(y,x_1,...,x_k)=argmin_{u\in K_n(y)}
\big(p^{(n)}J^{H,n}_{k+1}(y+ua^{(n)}_1,\\
&u_1, ...,u_k,\sqrt{\frac{T}{n}})+
(1-p^{(n)})J^{H,n}_{k+1}(y+ua^{(n)}_2,u_1,...,u_k,-\sqrt{\frac{T}{n}})\big),
\ \ k<n. \nonumber
\end{eqnarray}
Let $x$ be an initial capital. For any $n$ and an open interval $H$
there exists a hedge
$(\pi^{H}_n,\sigma^{H}_n)\in{\mathcal{A}^{\xi,n}(x)\times
\mathcal{T}^{\xi}_{0n}}$ such that
\begin{eqnarray}\label{3.15}
 &\tilde{V}^{\pi^{H}_n}_0=x \ \ \mbox{and} \ \  \tilde{V}^{\pi^{H}_n}_{k+1}=
 \tilde{V}^{\pi^{H}_n}_k+h^{H,n}_k(\tilde{V}^{\pi^{H}_n}_k,e^{\kappa\sqrt
 \frac{T}{n}
 \xi_1},...,e^{\kappa\sqrt\frac{T}{n}\xi_k})\\
 &\times(e^{\kappa\sqrt\frac{T}{n}\xi_{k+1}}-1) \ \ \mbox{for} \ \
 k>0 \ \
\mbox{and} \ \ \sigma^{H}_n=\sigma(H,\pi^{H}_n). \nonumber
\end{eqnarray}
From the arguments concerning $\mathcal{A}^{\xi,n}_k(X)$ at the
beginning of this section it follows that $\pi^{H}_n$ is an
$\textit{admissible}$ strategy. Let $(\pi^{B,H}_n,\zeta^{H}_n)
\in\mathcal{A}^{B,n}(x)\times\mathcal{T}^{B,n}_{0,n}$ be a hedge
which is given by $\pi^{B,H}_n=\psi_n(\pi^{H}_n)$ and
$\zeta^{H}_n=\Pi_n(\sigma^{H}_n)$ where, recall, the maps
$\psi_n,\Pi_n$ were defined in Section 2. Namely, we consider a
hedge which is determined by
\begin{eqnarray}\label{3.16}
 &\tilde{V}^{\pi^{B,H}_n}_0=x \ \ \mbox{and} \ \ \tilde{V}^{\pi^{B,H}_n}_{k+1}
 = \tilde{V}^{\pi^{B,H}_n}_k+h^{H,n}_k(\tilde{V}^{\pi^{B,H}_n}_k,
 e^{\bg^{(n)}_1},...,e^{\bg^{(n)}_k})\\
 &\times(e^{\bg^{(n)}_k+1}-1) \ \ \mbox{for} \ \ k>0 \ \
\mbox{and} \ \ \zeta^{H}_n=\zeta(H,\pi^{B,H}_n). \nonumber
\end{eqnarray}
The following lemma enables us to consider all relevant processes on
the Brownian probability space and to deal with stopping times with
respect to the same filtration.
\begin{lem}\label{lem3.2}
For any initial capital $x$, $n\in\mathbb{N}$ and an open interval
$H$.
\begin{equation}\label{3.17}
R^H_n(x)=R^H_n(\pi^{H}_n,\sigma^{H}_n)=J^{H,n}_0(x)=
R^{B,H}_n(\pi^{B,H}_n,\zeta^{H}_n)=R^{B,H}_n(x).
\end{equation}
\end{lem}
\begin{proof}
The proof is the same as in Lemma 3.3 of \cite{DK2}, just replace
$J^n_k$, $R_n$, $R^{B,n}$, $(\pi_n,\sigma_n)$ and
$(\tilde\pi_n,\zeta_n)$ by $J^{H,n}_k$, $R^H_n$, $R^{B,H}_n$,
$(\pi^{H}_n,\sigma^{H}_n)$ and $(\pi^{B,H}_n,\zeta^{H}_n)$,
respectively.
\end{proof}

Observe that if the initial capital $x$ is no less than $\cV^H_n$
then the hedge which is given by (\ref{3.15}) satisfy
$R^H_n(\pi^{H}_n,\sigma^{H}_n)=R^H_n(x)=0$. Namely,
$(\pi^{H}_n,\sigma^{H}_n)$ is a perfect hedge for a game option with
the payoffs $Y^{H,n}_k,X^{H,n}_k$, $k\geq{0}$. Thus, the dynamical
algorithm which is given by (\ref{3.13}) provides a way to find a
perfect hedge (when the initial capital is no less than the option
price) for CRR markets. Of course, in general a perfect hedge should
not be unique taking different versions of the term $argmin$ which
was defined before (\ref{3.14}) we will obtain other perfect hedges.
However, a more efficient way to find a perfect hedge is via the
Doob decomposition exactly as in Theorem 2.1 of \cite{Ki1}.

Next we deal with estimates for the BS model. Let $H=(L,R)$ be an
open interval. For any $\epsilon>0$ set
$H_\epsilon=(Le^{-\epsilon},Re^{\epsilon})$. Clearly,
 $\cV^{H_\epsilon}\geq \cV^H$ for any $\epsilon>0$ and
  $R^{H_\epsilon}(x)\geq R^H(x)$ for any initial capital $x$.
The following result provides an estimate from above of the term
$R^{H_\epsilon}(x)-R^H(x)$.
\begin{lem}\label{lem3.3}
There exists a constant $A_1$ such that for any initial capital $x$,
$\epsilon>0$ and an open interval $H$,
\begin{equation}\label{3.18}
R^{H_\epsilon}(x)-R^H(x)\leq A_1\epsilon^{3/4}.
\end{equation}
\end{lem}
\begin{proof} Before proving the lemma observe that if $P=\tilde{P}$
then the option price can be represented as the shortfall risk for
an initial capital $x=0$, i.e. if $\mu=0$ then $\cV^I=R^I(0)$ for
any open interval $I$. Hence, by (\ref{3.18}) there exists a
constant $A_2$ (which is equal to $A_1$ for the case $\mu=0$) such
that for any open interval $H$ and $\epsilon>0,$
\begin{equation}\label{3.19}
\cV^{H_\epsilon}-\cV^H(x)\leq A_2\epsilon^{3/4}.
\end{equation}
Next we turn to the proof of the lemma. Choose an initial capital
$x$, an open interval $H=(L,R)$, some $\epsilon>0$ and fix
$\delta>0$. There exists a $\pi_1\in\mathcal{A}^{B}(x)$ such that
$R^H(\pi_1)<R^H(x)+\delta$. According to (\ref{2.11}) the discounted
portfolio process ${\{\tilde{V}^{\pi_1}_t\}}_{t=0}^T$ is given by a
stochastic integral whose integrand in view of (\ref{2.9}) satisfies
the standard conditions assumed in the construction of stochastic
integrals, and so ${\{\tilde{V}^{\pi_1}_t\}}_{t=0}^T$ has a
continuous modification (see, for instance, Ch.2 in \cite{Mc} or
Ch.4 in \cite{LS}) which we take as the portfolio process. Observe
that $(Q^{B,H}(\sigma,\tau)-\tilde{V}^{\pi_1}_{\sigma\wedge\tau})^+=
(Q^{B,H}(\tau_H\wedge\sigma,\tau)-\tilde{V}^{\pi_1}_{\tau_H\wedge\sigma
\wedge\tau})^+$ for all stopping times
$\sigma,\tau\in\mathcal{T}^{B}_{0T}$. Thus, there exists a hedge
$(\pi_1,\sigma_1)\in\mathcal{A}^{B}(x)\times\mathcal{T}^{B}_{0T}$
such that
\begin{equation}\label{3.20}
R^H(\pi_1,\sigma_1)<R^H(x)+\delta \ \ \mbox{and} \ \
\sigma_1\leq{\tau_H}.
\end{equation}
Set
$\sigma_2=\sigma_1\mathbb{I}_{\sigma_1<\tau_H}+T\mathbb{I}_{\sigma_1
\geq\tau_H}$. Clearly,
$\{\sigma_2\leq{t}\}=\{\sigma_1\leq{t}\}\cap{\{\sigma_1<\tau_H\}}
\in\mathcal{F}^{B}_t$ for any $t<T$, and so we conclude that
$\sigma_2\in\mathcal{T}^{B}_{0T}$. Observe that if $\pi_1=
\{(\be_t,\gam_t)\}_{t=0}^T$ and
$\pi_2=\{(\tilde\be_t,\tilde\gam_t)\}_{t=0}^T$ with
$\tilde\gam_t=\gam_t\bbI_{\sig_1\leq t}$ and
$\tilde\be_t=(x+\int_0^t \tilde\gam_ud\tilde
S_u^B-\tilde\gam_t\tilde S_t^B)/b_0$ then $\pi_2$ is an admissible
self financing strategy and  $\tilde{V}^{\pi_2}_t=
\tilde{V}^{\pi_1}_{t\wedge\sigma_1}$. Consider the hedge
$(\pi_2,\sigma_2)\in\mathcal{A}^{B}(x)\times\mathcal{T}^{B}_{0T}$
then
$(Q^{B,H_\epsilon}(\sigma_2,\tilde\tau)-\tilde{V}^{\pi_2}_{\sigma_2\wedge
\tilde\tau})^+=
(Q^{B,H_\epsilon}(\sigma_2,\tilde\tau\wedge\tau_{H_\epsilon})-
\tilde{V}^{\pi_2}_{\tau_{H_\epsilon}\wedge\sigma_2\wedge\tilde\tau})^+$
for any $\tilde\tau\in\mathcal{T}^{B}_{0T}$. Thus, there exists a
stopping time $\tau\in\mathcal{T}^{B}_{0T}$ such that
\begin{equation}\label{3.21}
R^{H_\epsilon}(\pi_2,\sigma_2)<E^B[Q^{B,H_\epsilon}(\sigma_2,\tau)-
\tilde{V}^{\pi_2}_{\sigma_2\wedge\tau})^+]+\delta \ \ \mbox{and} \ \
\tau\leq \tau_{H_\epsilon}.
\end{equation}
For any $\alpha>0$ denote $J_\alpha=(Le^\alpha,Re^{-\alpha})$. Set
$U_{\alpha}=(Q^{B,H}(\sigma_1,\tau\wedge\tau_{J_{\alpha}})-
\tilde{V}^{\pi_1}_{\sigma_1\wedge\tau\wedge\tau_{J_{\alpha}}})^+$.
Clearly, $\tau\wedge\tau_{J_\alpha}\leq\tau\wedge\tau_H$ for any
$\alpha>0$ and $\tau\wedge\tau_{J_\alpha}
\uparrow\tau\wedge\tau_{H}$ as $\alpha\rightarrow{0}$. This together
with (\ref{2.2}) yields that
\begin{eqnarray*}
&\lim_{\alpha\rightarrow{0}}Q^{B,H}(\sigma_1,\tau\wedge\tau_{J_{\alpha}})=
\tilde{X}^{H}_{\sigma_1}\lim_{\alpha\rightarrow{0}}
\mathbb{I}_{\sigma_1<\tau\wedge\tau_{J_{\alpha}}}+
\lim_{\alpha\rightarrow{0}}\tilde{Y}^H_{\tau\wedge\tau_{J_{\alpha}}}
\mathbb{I}_{\sigma_1\geq\tau\wedge\tau_{J_{\alpha}}}\\
&=e^{-r\sigma_1}G_{\sigma_1}(S^B)\mathbb{I}_{\sigma_1<\tau\wedge\tau_H}+
\lim_{\alpha\rightarrow{0}}e^{-r(\tau\wedge\tau_{J_{\alpha}})}F_{\tau\wedge
\tau_{J_{\alpha}}}(S^B)\mathbb{I}_{\sigma_1\geq\tau\wedge\tau_{J_{\alpha}}}\\
&=e^{-r\sigma_1}G_{\sigma_1}(S^B)\mathbb{I}_{\sigma_1<\tau\wedge\tau_H}+
e^{-r(\tau\wedge\tau_H)}F_{\tau\wedge\tau_H}(S^B)\mathbb{I}_{\sigma_1\geq\tau
\wedge\tau_H}.
\end{eqnarray*}
Since the process ${\{\tilde{V}^{\pi_1}_t\}}_{t=0}^T$ is continuous
and $\sigma_1\leq{\tau_H}$ we obtain by the choice of $\pi_2$ that
\begin{eqnarray}\label{3.22}
&\lim_{\alpha\rightarrow{0}}U_{\alpha}=(e^{-r\sigma_1}
G_{\sigma_1}(S^B)\mathbb{I}_{\sigma_1<\tau\wedge\tau_H}+
e^{-r(\tau\wedge\tau_H)}F_{\tau\wedge\tau_H}(S^B)\mathbb{I}_{\sigma_1\geq\tau
\wedge\tau_H}
\\
&-\tilde{V}^{\pi_1}_{\sigma_1\wedge\tau})^+=(e^{-r\sigma_1}
G_{\sigma_1}(S^B)\mathbb{I}_{\sigma_1<\tau\wedge\tau_H}+
e^{-r(\tau\wedge\tau_H)}F_{\tau\wedge\tau_H}(S^B)\mathbb{I}_{\sigma_1\geq
\tau\wedge\tau_H}- \tilde{V}^{\pi_2}_\tau)^+.\nonumber
\end{eqnarray}
Observe that $R^{H}(\pi_1,\sigma_1)\geq{E^BU_{\alpha}}$ for any
$\alpha$. Thus from (\ref{3.22}) and the Fatou's lemma we obtain
\begin{eqnarray}\label{3.23}
&R^{H}(\pi_1,\sigma_1)\geq{E^B\lim_{\alpha\rightarrow{0}}U_{\alpha}}=E^B(
e^{-r\sigma_1}G_{\sigma_1}(S^B)\mathbb{I}_{\sigma_1<\tau\wedge\tau_{J_1}}+\\
&e^{-r(\tau\wedge\tau_H)}F_{\tau\wedge\tau_H}(S^B)\mathbb{I}_{\sigma_1\geq
\tau\wedge\tau_H}- \tilde{V}^{\pi_2}_{\tau})^+.\nonumber
\end{eqnarray}
Since $\sigma_2\geq\sigma_1$ a.s. then from the definition of
$\pi_2$ it follows that
$\tilde{V}^{\pi_2}_{\sigma_2\wedge{t}}=\tilde{V}^{\pi_1}_{\sigma_1\wedge
\sigma_2\wedge{t}}=
\tilde{V}^{\pi_1}_{\sigma_1\wedge{t}}=\tilde{V}^{\pi_2}_{{t}}$ for
all $t$. This together with (\ref{3.21}) gives
\begin{equation}\label{3.24}
R^{H_\epsilon}(\pi_2,\sigma_2)<E^B(e^{-r\sigma_2}G_{\sigma_2}(S^B)
\mathbb{I}_{\sigma_2<\tau}+
e^{-r\tau}F_{\tau}(S^B)\mathbb{I}_{\sigma_2\geq\tau}-
\tilde{V}^{\pi_2}_{\tau})^+ +\delta.
\end{equation}
Observe that if $\sigma_2<\tau$ then
$\sigma_2=\sigma_1<\tau\wedge\tau_H$ and if $\sigma_2\geq\tau$ then
$\sigma_1\geq\tau\wedge\tau_H$. And so from (\ref{3.20}),
(\ref{3.23}) and (\ref{3.24}) we obtain that
\begin{eqnarray}\label{3.25}
&R^{H_\epsilon}(x)-R^{H}(x)\leq
R^{H_\epsilon}(\pi_2,\sigma_2)-R^{H}(\pi_1,\sigma_1)+\delta\leq
{2\delta}+\\
&E^B|e^{-r\tau}F_{\tau}(S^B)-e^{-r(\tau\wedge\tau_H)}F_{\tau\wedge\tau_H}(S^B)|
\leq 2\delta+E^B\Gamma_1+E^B\Gamma_2 \nonumber
\end{eqnarray}
where
\begin{equation*}
\Gamma_1=|e^{-r\tau}-e^{-r(\tau\wedge\tau_H)}|F_{\tau}(S^B) \ \
\mbox{and} \ \ \Gamma_2=|F_{\tau}(S^B)-F_{\tau\wedge\tau_H}(S^B)|.
\end{equation*}
In order to estimate $E^B\Gamma_1$ and $E^B\Gamma_2$ introduce the
process $W_{t}=\frac{\ln{S^B_t}-\ln{S_0}}{\kappa}=
B_t+(\frac{r+\mu}{\kappa}-\frac{\kappa}{2})t$, $t\geq{0}$. From
Girsanov's theorem (see \cite{KS}) it follows that
${\{W_t\}}_{t=0}^T$ is a Brownian motion with respect to the measure
$P_W$ whose restriction to the $\sigma$--algebra $\mathcal{F}^{B}_t$
satisfies
\begin{equation}\label{3.26}
D_t=\frac{d{P}^{B}}{dP_W}|\mathcal{F}^{B}_t=
\exp\bigg((\frac{r+\mu}{\kappa}-\frac{\kappa}{2})B_t+
\frac{(\frac{r+\mu}{\kappa}-\frac{\kappa}{2})^2}{2}t\bigg).
\end{equation}
Denote the expectation with respect to $P_W$ by $E_W$ then by
(\ref{2.3}) and the H\"{o}lder inequality,
\begin{eqnarray}\label{3.27}
&E^B\Gamma_1\leq
E_W\bigg(r(\tau-\tau\wedge\tau_H)(F_0(S_0)+\cL(T+2)(1+\sup_{0\leq{t}
\leq{T}}S^B_t))D_T\bigg)\\
&\leq c_1(E_W(\tau-\tau\wedge\tau_H)^{4/3})^{3/4} \nonumber
\end{eqnarray}
for some constant $c_1$. From (\ref{2.2}) it follows that
$\Gamma_2\leq \Gamma_3+\Gamma_4$ where
\begin{equation*}
\Gamma_3=\cL(\tau-\tau\wedge\tau_H)(1+\sup_{0\leq{t}
\leq{T}}S^{B}_t)\ \ \mbox{and} \ \
\Gamma_4=\sup_{\tau\wedge\tau_H\leq{t}\leq\tau}\cL|S^{B}_t-
S^{B}_{\tau\wedge\tau_H}|.
\end{equation*}
By the H\"{o}lder inequality,
\begin{equation}\label{3.28}
E^B\Gamma_3=E_W(\cL(\tau-\tau\wedge\tau_H)(1+\sup_{0\leq{t}
\leq{T}}S^{B}_t)D_T)\leq c_2(E_W(\tau-\tau\wedge\tau_H)^{4/3})^{3/4}
\end{equation}
for some constant $c_2$. Set
$\Gamma_5=\sup_{\tau\wedge\tau_H\leq{t}\leq\tau}\kappa|W_t-W_{\tau\wedge
\tau_H}|$. Employing the inequality $|e^x-1|\leq{x}$ for $0\leq
x\leq 1$ it follows that $\Gamma_4\leq
\cL\sup_{0\leq{t}\leq{T}}S^B_t(\mathbb{I}_{\Gamma_5>1}+\Gamma_5)$
and together with the Markov and H\"{o}lder inequalities we obtain
that there exists a constant $c_3$ such that
\begin{eqnarray}\label{3.29}
&E^B\Gamma_4\leq
E_W(D_T\cL\sup_{0\leq{t}\leq{T}}S^B_t\mathbb{I}_{\Gamma_5>1})+
E_W(D_T\cL\sup_{0\leq{t}\leq{T}}S^B_t\Gamma_5)\leq\\
&c_3(P_W\{\Gamma_5>1\})^{3/4}+c_3(E_W\Gamma^{4/3}_5)^{3/4}\leq
2c_3(E_W\Gamma^{4/3}_5)^{3/4}\nonumber.
\end{eqnarray}
Using the Burkholder--Davis--Gandy inequality (see \cite{KS}) for
the martingale $W_t-W_{\tau\wedge\tau_H}$, $t\geq \tau\wedge\tau_H$
we obtain that there exists a constant $c_4$ such that
\begin{equation}\label{3.30}
E_W\Gamma^{4/3}_5\leq c_4 E_W(\tau-\tau\wedge\tau_H)^{2/3}.
\end{equation}
Since $\tau-\tau\wedge\tau_H\leq T$ then from
(\ref{3.27})-(\ref{3.30}) we obtain
\begin{equation}\label{3.31}
E^B(\Gamma_1+\Gamma_2)\leq c_5
(E_W(\tau-\tau\wedge\tau_H)^{2/3})^{3/4}
\end{equation}
for some constant $c_5$. Finally, we estimate the term
$E_W(\tau-\tau\wedge\tau_H)^{2/3}$. First assume that $L>0$ and
$R<\infty$. Set $x_1=({\ln{L}-\ln{S_0}})/{\kappa}$,
$x_2=({\ln{R}-\ln{S_0}})/{\kappa}$,
$y_1=x_1-\frac{\epsilon}{\kappa}$ and
$y_2=x_2+\frac{\epsilon}{\kappa}$. For any $x\in\mathbb{R}$ let
$\tau^{(x)}=\inf{\{t\geq{0}|W_t=x\}}$ be the first time the process
${\{W_{t}\}}_{t=0}^\infty$ hits the level $x$. Clearly $\tau^{(x)}$
is a finite stopping time with respect to $P_W$. By (\ref{3.21}) we
obtain that
\begin{eqnarray}\label{3.32}
&\tau-\tau\wedge\tau_H\leq
T\wedge(\tau_{H_\epsilon}-\tau_H)=T\wedge(\tau^{(y_1)}\wedge\tau^{(y_2)}-
\tau^{(x_1)}\wedge\tau^{(x_2)})\leq\\
&T\wedge(\tau^{(y_1)}-\tau^{(x_1)})+
T\wedge(\tau^{(y_2)}-\tau^{(x_2)})\nonumber.
\end{eqnarray}
From the strong Markov property of the Brownian motion it follows
that under $P_W$  the random variable $\tau^{(y_1)}-\tau^{(x_1)}$
has the same distribution as
$\tau^{(y_1-x_1)}=\tau^{(-\frac{\epsilon}{\kappa})}$ and the random
variable $\tau^{(y_2)}-\tau^{(x_2)}$ has the same distribution as
$\tau^{(y_2-x_2)}=\tau^{(\frac{\epsilon}{\kappa})}$. Recall, (see
\cite{KS}) that for any $z\in\mathbb{R}$ the probability density
function of $\tau^{(z)}$ (with respect to $P_W$) is
$f_{\tau^{(z)}}(t)=\frac{|z|}{\sqrt{2\pi{t^3}}}\exp(-\frac{z^2}{2t})$.
Hence, using the inequality $(a+b)^{2/3}\leq a^{2/3}+b^{2/3}$
together with (\ref{3.32}) we obtain that
\begin{eqnarray}\label{3.33}
&E_W(\tau-\tau\wedge\tau_H)^{2/3}\leq
E_W(T\wedge\tau^{(-\frac{\epsilon}{\kappa})})^{2/3}+E_W(T\wedge
\tau^{(\frac{\epsilon}{\kappa})})^{2/3}\leq\\
&\frac{2\epsilon}{\sqrt{2\pi}\kappa}\bigg(\int_{0}^{T} \frac{1}{
t^{3/2-2/3}}dt+T^{2/3} \int_{T}^\infty \frac{1}{
t^{3/2}}dt\bigg)=\frac{16\epsilon}{\sqrt{2\pi}\kappa}T^{1/6}.\nonumber
\end{eqnarray}
Observe that when either $L=0$ or $R=\infty$ (but not both) we
obtain either $\tau-\tau\wedge\tau_H\leq
T\wedge(\tau^{(y_2)}-\tau^{(x_2)})$ or $\tau-\tau\wedge\tau_H\leq
T\wedge(\tau^{(y_1)}-\tau^{(x_1)})$, respectively. Thus for these
cases (\ref{3.33}) holds true, as well. From (\ref{3.25}),
(\ref{3.31}) and (\ref{3.33}) we see that there exists a constant
$A_1$ such that
\begin{equation*}
R^{H_\epsilon}(x)-R^H(x)\leq 2\delta+A_1\epsilon^{3/4}
\end{equation*}
and since $\delta>0$ is arbitrary we complete the proof.
\end{proof}

The next result provides an estimate from above of the shortfall
risk when one of the barriers is close to the initial stock price
$S_0$.
\begin{lem}\label{lem3.4}
Let $I=(L,R)$ be an open interval which satisfy
$\min(\frac{R}{S_0},\frac{S_0}{L})\leq e^\epsilon$, where we set
$\frac{S_0}{0}=\frac{\infty}{S_0}=\infty$. There exists a constant
$A_3$ independent of $L,R$ such that for any $\epsilon>0$ and an
initial capital $x$
\begin{equation}\label{3.34}
R^I(x)\leq (F_0(S_0)-x)^{+}+A_3\epsilon^{3/4}.
\end{equation}
\end{lem}
\begin{proof}
Let $x$ be an initial capital. Consider the constant portfolio
$\pi\in\mathcal{A}^{B}(x)$ which satisfy $\tilde{V}^{\pi}_t=x$ for
all t. Using the same notations as in Lemma \ref{lem3.3} set
$\sigma=(\tau^{(\frac{\epsilon}{\kappa})}\vee
\tau^{(-\frac{\epsilon}{\kappa})})\wedge{T}$. Since
$\tau^{(\frac{\epsilon}{\kappa})}\vee\tau^{(-\frac{\epsilon}{\kappa})}
\geq \tau_I$ we obtain that
\begin{equation}\label{3.35}
R^I(x)\leq
R^I(\pi,\sigma)\leq\sup_{\tau\in\mathcal{T}^{B}_{0T}}E^B(e^{-r(\tau\wedge
(\tau^{(\frac{\epsilon}{\kappa})}\vee\tau^{(-\frac{\epsilon}{\kappa})}))}
F_{\tau\wedge
(\tau^{(\frac{\epsilon}{\kappa})}\vee\tau^{(-\frac{\epsilon}{\kappa})})}
(S^B)-x)^{+}.
\end{equation}
Similarly to (\ref{3.31}) (by letting $\tau_H$=0) we obtain that
\begin{eqnarray}\label{3.36}
&\sup_{\tau\in\mathcal{T}^{B}_{0T}}E^B|e^{-r(\tau\wedge
(\tau^{(\frac{\epsilon}{\kappa})}\vee\tau^{(-\frac{\epsilon}{\kappa})}))}
F_{\tau\wedge (\tau^{(\frac{\epsilon}{\kappa})}\vee
\tau^{(-\frac{\epsilon}{\kappa})})}(S^B)-F_0(S_0)|
\leq \\
&c_5 (E_W({T}\wedge (\tau^{(\frac{\epsilon}{\kappa})}\vee
\tau^{(-\frac{\epsilon}{\kappa})}))^{2/3})^{3/4}.\nonumber
\end{eqnarray}
In the same way as in (\ref{3.33}) we derive that
\begin{eqnarray}\label{3.37}
E_W({T}\wedge
(\tau^{(\frac{\epsilon}{\kappa})}\vee\tau^{(-\frac{\epsilon}{\kappa})}))^{2/3}
 \leq \frac{16\epsilon}{\sqrt{2\pi}\kappa}T^{1/6}
\end{eqnarray}
and combining (\ref{3.35})-(\ref{3.37}) we complete the proof.
\end{proof}

\section{Proving the main results}\label{sec4}\setcounter{equation}{0}
In this section we complete the proof of Theorems
\ref{thm2.1}--\ref{thm2.4}. We start with the proof of Theorem
\ref{thm2.2}. Though Theorem \ref{thm2.2} provides only one sided
estimates for shortfall risks we will see that Theorem \ref{thm2.1}
which provide two sided estimates for option prices follows from the
proof of Theorem \ref{thm2.2}. In order to provide second side
estimates in Theorem \ref{thm2.2} we should have more precise
information on optimal portfolios of shortfall risk in the BS model.
However, this problem does not arise when we are dealing with option
prices. Theorem \ref{thm2.4} will also follow from the proof of
Theorem \ref{thm2.2}. At the end of this section we prove Theorem
\ref{thm2.3}. The proof of (2.27) and (2.28) is necessarily rather
technical and it is marked by various risk comparisons via the
formulas (\ref{4.1}), (\ref{4.8}), (\ref{4.9}), (\ref{4.12}), then
estimates of terms in the right hand side of (\ref{4.12}), then
(\ref{4.26})--(\ref{4.31}), then (\ref{4.35}) and estimates of its
right hand side and, finally, (\ref{4.44+}) and (\ref{4.45}) so that
these formulas may serve as road posts for the reader going through
all these details.

Let $x>0$ be an initial capital and let $I=(L,R)$ be an open
interval as before. Fix $\epsilon>0$ and denote
$I_\epsilon=(Le^{-\epsilon},Re^{\epsilon})$. Choose $\delta>0$. For
any $z$ let $\mathcal{A}^{B,C}(z)\subset\mathcal{A}^{B}(z)$ be the
subset consisting of all $\pi\in\mathcal{A}^{B}(z)$ such that the
discounted portfolio process ${\{\tilde{V}^{\pi}_t\}}_{t=0}^T$ is a
right continuous martingale with respect to the martingale measure
$\tilde{P}^B$ and $\tilde{V}^{\pi}_T=f(B^*_{t_1},...,B^*_{t_k})$ for
some smooth function $f\in{C^{\infty}_0(\mathbb{R}^k)}$ with a
compact support and $t_1,...,t_k\in{[0,T]}$. Using the same
arguments as in Lemmas 4.1--4.3 in \cite{DK2} we obtain that there
exists $z<x$ and $\pi\in\mathcal{A}^{B,C}(z)$ such that
$R^{I_\epsilon}(\pi)<R^{I_\epsilon}(x)+\delta$. Thus there exist
$k$, $0<t_1<t_2...<t_k\leq{T}$ and
$0\leq{f_\delta}\in{C^{\infty}_0(\mathbb{R}^k)}$ such that the
portfolio $\pi\in\mathcal{A}^{B}$ with
$\tilde{V}^{\pi}_t=\tilde{E}(f_{\delta}(B^{*}_{t_1},...,
B^{*}_{t_k})|\mathcal{F}^B_t)$ satisfies
\begin{equation}\label{4.1}
R^{I_\epsilon}(\pi)<R^{I_\epsilon}(x)+\delta \ \ \mbox{and} \ \
V^\pi_0<x.
\end{equation}
Set
\begin{equation}\label{4.2}
\Psi_n=f_{\delta}(B^{*}_{\theta^{(n)}_{[nt_1/T]}},...,
B^{*}_{\theta^{(n)}_{[nt_k/T]}}),
\end{equation}
 $u_n=\max_{0\leq{k}\leq{n}}|\theta^{(n)}_k-\frac{kT}{n}|$ and
$w_n=\max_{1\leq{k}\leq{n}}|\theta^{(n)}_k-\theta^{(n)}_{k-1}|+
|T-\theta^{(n)}_n|$. Since $w_n\leq{3u_n}+\frac{T}{n}$ then from
(4.7) in \cite{Ki2} we obtain that for any $m\in\mathbb{R}_+$ there
exists a constant $K^{(m)}$ such that for all $n$,
\begin{equation}\label{4.3}
 E^{B}u^{2m}_n\leq{K^{(m)}n^{-m}}\,\,\mbox{and}\,\,
E^{B}w^{2m}_n\leq{K^{(m)}n^{-m}}.
\end{equation}
From the exponential moment estimates (4.8) and (4.25) of \cite
{Ki2} it follows that there exists a constant $K_1$ such that for
any natural $n$ and a real $a$,
\begin{equation}\label{4.4}
\begin{split}
E^{B}e^{|a|\theta^{(n)}_n\vee{T}}\leq{e^{|a|K_1T}} \ \mbox{and} \ \
E^B\sup_{0\leq{t}\leq{\theta^{(n)}_n\vee{T}}}\exp(aB_t)\leq{2e^{a^2K_1T}}.
\end{split}
\end{equation}
Clearly $(B^{*}_t-B^{*}_{\theta^{(n)}_{[nt/T]}})^2\leq{2(B_t-
B_{\theta^{(n)}_{[nt/T]}})^2+2((\frac{\mu}{\kappa}-\frac{\kappa}{2})
(t-\theta^{(n)}_{[nt/T]}))^2}$ and
$|t-\theta^{(n)}_{[nt/T]}|\leq{\frac{T}{n}+u_n}$. Hence, from
(\ref{4.3}) and It\^{o}'s isometry for the Brownian motion it
follows that there exists a constant $C^{(1)}$ such that
$E^{B}|B^{*}_t-B^{*}_{\theta^{(n)}_{[nt/T]}}|^2\leq{C^{(1)}n^{-1/2}}$
for all $t$. Let
$\cL(f_\delta)=\max_{1\leq{i}\leq{k}}\sup_{x\in\mathbb{R}^k}|
\frac{\partial{f_\delta}}{\partial{x_i}}(x_1,...,x_k)|$. Then by
(\ref{4.2}) and the inequality
$(\sum_{i=1}^ka_i)^2\leq{k\sum_{i=1}^ka^2_i}$ we obtain
\begin{eqnarray}\label{4.6}
&E^{B}(\Psi_n-\tilde{V}^{\pi}_T)^2\leq{\cL(f_\delta)^2
E^{B}(\sum_{i=1}^k|B^{*}_{t_k}-
B^{*}_{\theta^{(n)}_{[nt_k/T]}}|)^2}\leq\\
&{k\cL(f_\delta)^2}\sum_{i=1}^kE^{B}(B^{*}_{t_k}-
B^{*}_{\theta^{(n)}_{[nt_k/T]}})^2\leq{k^2\cL(f_\delta)^2C^{(1)}n^{-1/2}}.
\nonumber
\end{eqnarray}
By (\ref{4.4}) and the Cauchy-Schwarz inequality,
\begin{equation*}
\lim_{n\rightarrow\infty}\tilde{E}^{B}|\Psi_n-\tilde{V}^{\pi}_T|=
\lim_{n\rightarrow\infty}({E}^{B}|\Psi_n-
\tilde{V}^{\pi}_T|^2)^{1/2}(E^{B}Z^{-2}_{\theta^{(n)}_n\vee{T}})^{1/2}=0
\end{equation*}
where $Z_t$ is the Radon-Nikodim derivative given by (\ref{2.8+}).
Since $\tilde{E}^{B}\tilde{V}^{\pi}_T<x$ then for sufficiently large
$n$ we can assume that $v_n=\tilde{E}(\Psi_n)<x$. Observe that the
finite dimensional distributions of the sequence
$\sqrt{\frac{T}{n}}\xi_1,..,\sqrt{\frac{T}{n}}\xi_n$ with respect to
$\tilde{P}^{\xi}_n$ and the finite dimensional distributions of the
sequence $\bg^{(n)}_1,...,\bg^{(n)}_n$ with respect to
$\tilde{P}^{B}$ are the same, and so $v_n=\tilde{E}^{\xi}_n
f_{\delta}\bigg(\sqrt{\frac{T}{n}}\sum_{i=1}^{[nt_1/T]}\xi_i,
...,\sqrt{\frac{T}{n}}\sum_{i=1}^{[nt_k/T]}\xi_i\bigg)<x$ (for
sufficiently large $n$). Since CRR markets are complete we can find
a portfolio $\tilde{\pi}(n)\in\mathcal{A}^{\xi,n}(v_n)$ such that
\begin{equation}\label{4.7}
\tilde{V}^{\tilde{\pi}}_n=f_{\delta}\bigg(\sqrt{\frac{T}{n}}
\sum_{i=1}^{[nt_1/T]}\xi_i,
...,\sqrt{\frac{T}{n}}\sum_{i=1}^{[nt_k/T]}\xi_i\bigg).
\end{equation}
For a fixed $n$ let
$\pi'=\psi_n(\tilde{\pi})\in\mathcal{A}^{B,n}(v_n)$. From
(\ref{2.22}) it follows that
$\tilde{V}^{\pi'}_{\theta^{(n)}_n}=\Psi_n$. Since $R^I_n(\cdot)$ is
a non increasing function then by (\ref{4.1}) and Lemma 3.2,
\begin{equation}\label{4.8}
R^{I}_n(x)-R^{I_\epsilon}(x)\leq{R^I_n(v_n)-R^{I_\epsilon}(x)}
\leq{\delta+R^{B,I}_n(\pi')-R^{I_\epsilon}(\pi)}.
\end{equation}
There exists a stopping time $\sigma\in\mathcal{T}^{B}_{0T}$ such
that
\begin{equation}\label{4.9}
R^{I_\epsilon}(\pi)>{\sup_{\tau\in\mathcal{T}^{B}_{0T}}E^{B}(Q^{B,I_\epsilon}
(\sigma,\tau)-\tilde{V}^{\pi}_{\sigma\wedge\tau})^{+}}-\delta.
\end{equation}
Set
\begin{equation}\label{4.10}
\zeta=(n\wedge\min{\{i|\theta^{(n)}_i\geq{\sigma}\}})\mathbb{I}_{\sigma<T}+
n\mathbb{I}_{\sigma=T}.
\end{equation}
Clearly, $\zeta\leq{n}$ a.s. and
$\{\zeta\leq{i}\}={\{\sigma\leq{\theta^{(n)}_i}\}}\cap{\{\sigma<T\}}\in
\mathcal{F}^{B}_{\theta^{(n)}_i}$ for any $i<n$ implying that
$\zeta\in{\mathcal{T}^{B,n}_{0,n}}$. There exists a stopping time
$\eta\in\mathcal{T}^{B,n}_{0,n}$ such that
\begin{eqnarray}\label{4.11}
&E^{B}(Q^{B,I,n}(\frac{\zeta{T}}{n},\frac{{\eta}T}{n})-
\tilde{V}^{\pi'}_{\theta^{(n)}_{\zeta\wedge\eta}})^{+}>
\sup_{\tilde{\eta}\in\mathcal{T}^{B,n}_{0,n}}E^{B}(Q^{B,I,n}
(\frac{\zeta{T}}{n},\frac{{\tilde\eta}T}{n})\\
&-\tilde{V}^{\pi'}_{\theta^{(n)}_{\zeta\wedge\tilde{\eta}}})^{+}-\delta
\geq{R^{B,I}_n(\pi')-\delta}.\nonumber
\end{eqnarray}
From (\ref{4.9}) and (\ref{4.11}) we obtain that
\begin{eqnarray}\label{4.12}
&R^{B,I}_n(\pi')-R^{I_\epsilon}(\pi)<2\delta+E^{B}(Q^{B,I,n}
(\frac{\zeta{T}}{n},\frac{{\eta}T}{n})-
\tilde{V}^{\pi'}_{\theta^{(n)}_{\zeta\wedge\eta}})^{+}\\
&-E^{B}(Q^{B,I_\epsilon}(\sigma,\theta^{(n)}_{\eta}\wedge{T})-
\tilde{V}^{\pi}_{\sigma\wedge\theta^{(n)}_{\eta}})^{+} \leq 2\delta+
E^B(\Lambda_1+\Lambda_2+\Lambda_3) \nonumber
\end{eqnarray}
where
\begin{eqnarray}\label{4.13}
&\Lambda_1=|\tilde{V}^{\pi'}_{\theta^{(n)}_{\zeta\wedge\eta}}-
\tilde{V}^{\pi}_{\theta^{(n)}_{\zeta\wedge\eta}\wedge{T}}|, \ \
\Lambda_2=|\tilde{V}^{\pi}_{\theta^{(n)}_{\zeta\wedge\eta}\wedge{T}}-
\tilde{V}^{\pi}_{\theta^{(n)}_{\eta}\wedge{\sigma}}| \\
&\mbox{and} \ \
\Lambda_3=(Q^{B,I,n}(\frac{\zeta{T}}{n},\frac{{\eta}T}{n})-Q^{B,I_\epsilon}
(\sigma,\theta^{(n)}_{\eta}\wedge{T}))^+\nonumber.
\end{eqnarray}
Since the processes ${\{\tilde{V}^{\pi}_t\}}$, $t\geq{0}$ is a
martingale then
$\tilde{V}^{\pi}_{\theta^{(n)}_{\zeta\wedge\eta}\wedge{T}}=
\tilde{E}^B(\tilde{V}^{\pi}_T|\mathcal{F}^B_{\theta^{(n)}_{\zeta\wedge\eta}
\wedge{T}})=
\tilde{E}^B(\tilde{V}^{\pi}_T|\mathcal{F}^B_{\theta^{(n)}_{\zeta\wedge\eta}})$
taking into account that $\tilde{V}^{\pi}_T$ is $\mathcal{F}^B_T$
measurable. Since the processes ${\{\tilde{V}^{\pi'}_t\}}$,
$t\geq{0}$ is a martingale and
$\Psi_n=\tilde{V}^{\pi'}_{\theta^{(n)}_n}$ then
$\tilde{V}^{\pi'}_{\theta^{(n)}_{\zeta\wedge\eta}}=
\tilde{E}^B(\Psi_n|\mathcal{F}^B_{\theta^{(n)}_{\zeta\wedge\eta}})$.
Thus
\begin{equation}\label{4.14}
\tilde{V}^{\pi'}_{\theta^{(n)}_{\zeta\wedge\eta}}
-\tilde{V}^{\pi}_{\theta^{(n)}_{\zeta\wedge\eta}\wedge{T}}=
\tilde{E}^B(\Psi_n-\tilde{V}^{\pi}_T|\mathcal{F}^B_{\theta^{(n)}_{\zeta\wedge
\eta}
})=E^B\bigg(\frac{Z_{\theta^{(n)}_{\zeta\wedge\eta}}}{Z_{T\vee\theta^{(n)}_n}}
(\Psi_n-\tilde{V}^{\pi}_T)|\mathcal{F}^B_{\theta^{(n)}_{\zeta\wedge\eta}}\bigg).
\end{equation}
By (\ref{4.4}), (\ref{4.6}), (\ref{4.14}), the Cauchy-Schwarz and
Jensen inequalities,
\begin{equation}\label{4.15}
E^B\Lambda_1\leq
\bigg(E^B\bigg(\frac{Z_{\theta^{(n)}_{\zeta\wedge\eta}}}
{Z_{T\vee\theta^{(n)}_n}}\bigg)^2\bigg)^{1/2}
(E^B(\Psi_n-\tilde{V}^{\pi}_T)^2)^{1/2}\leq C(f_\delta)n^{-1/4}
\end{equation}
where $C(f_\delta)$ is a constant which depends only on $f_\delta$.
By using the same arguments as in (5.14)-(5.17) of \cite{DK2} we
obtain that
\begin{equation}\label{4.16}
E^B\Lambda_2\leq{\tilde{C}(f_\delta)n^{-1/2}}
\end{equation}
for some constant $\tilde{C}(f_\delta)$ which depends only on
$f_\delta$. Next, we estimate $\Lambda_3$. Set
\begin{eqnarray*}
&Q^B(s,t)=e^{-rt}G_t(S^B)\mathbb{I}_{t<s}+e^{-rs}F_s(S^B)\mathbb{I}_{s\leq{t}}, \ \ s,t\geq{0} \ \ \mbox{and}\\
&
Q^{B,n}(k,l)=(1+r_n)^{-k}G_{\frac{kT}{n}}(S^{B,n})\mathbb{I}_{k<l}+
(1+r_n)^{-l}F_{\frac{lT}{n}}(S^{B,n}) \mathbb{I}_{l\leq{k}}, \ \
k,l\leq{n}.
\end{eqnarray*}
From (\ref{2.3}) and (\ref{4.13}) we get
\begin{equation}\label{4.17}
\Lambda_3\leq
(Q^{B,n}(\frac{\zeta{T}}{n},\frac{{\eta}T}{n})-Q^{B}(\sigma,\theta^{(n)}_{\eta}
\wedge{T}))^{+}
+\mathbb{I}_\Theta(G_0(S_0)+\cL(T+2)(1+\max_{0\leq{k}
\leq{n}}S^{B,n}_{\frac{kT}{n}}))
\end{equation}
where
$\Theta=\{\zeta\wedge\eta<\tau^{B,n}_{I}\}\cap\{\sigma\wedge\theta^{(n)}_{\eta}
\geq\tau_{I_\epsilon}\}$. Similarly to Lemmas 3.2 and 3.3 in
\cite{Ki2} it follows that there exists a constant $C^{(2)}$  such
that
\begin{eqnarray}\label{4.18}
&\sup_{\zeta\in\mathcal{T}^{B,n}_{0,n}}\sup_{\eta\in\mathcal{T}^{B,n}_{0,n}}
E^{B}|Q^{B}(\theta^{(n)}_{\zeta},\theta^{(n)}_{\eta})-Q^{B,n}(\frac{\zeta{T}}
{n},\frac{\eta{T}}{n})|\leq\\
&{C^{(2)}n^{-1/4}(\ln{n})^{3/4}}.\nonumber
\end{eqnarray}
From (\ref{4.4}) and the Cauchy-Schwarz inequality it follows that
\begin{equation}\label{4.19}
E^B\bigg(\mathbb{I}_\Theta(G_0(S_0)+\cL(T+2)(1+\max_{0\leq{k}\leq{n}}
S^{B,n}_{\frac{kT}{n}}))\bigg)\leq C^{(3)}P(\Theta)^{1/2}
\end{equation}
for some constant $C^{(3)}$. By (\ref{4.10}) we see that
  ${\sigma}<\theta^{(n)}_{\eta}\wedge{T}$ provided $\zeta<\eta$. This
together with (\ref{4.17})--(\ref{4.19}) gives
\begin{eqnarray}\label{4.20}
&E^B\Lambda_3\leq
C^{(3)}P(\Theta)^{1/2}+C^{(2)}n^{-1/4}(\ln{n})^{3/4}+
E^B(Q^B(\theta^{(n)}_\zeta,\theta^{(n)}_\eta)-\\
&Q^{B}(\sigma,\theta^{(n)}_{\eta}\wedge{T}))^{+}\leq
C^{(3)}P(\Theta)^{1/2}+C^{(2)}n^{-1/4}(\ln{n})^{3/4}+\alpha_1+\alpha_2
\nonumber
\end{eqnarray}
where
\begin{eqnarray}\label{4.21}
&\alpha_1=E^{B}|e^{-r\theta^{(n)}_{\zeta\wedge\eta}}
G_{\theta^{(n)}_{\zeta\wedge\eta}}
(S^{B})-e^{-r{{\sigma}\wedge\theta^{(n)}_{\eta}}}
G_{\sigma\wedge\theta^{(n)}_{\eta}}(S^{B})|\\
&\mbox{and} \ \ \alpha_2=E^{B}|e^{-r\theta^{(n)}_{\zeta\wedge\eta}}
F_{\theta^{(n)}_{\zeta\wedge\eta}}
(S^{B})-e^{-r{{\sigma}\wedge\theta^{(n)}_{\eta}}}
F_{\sigma\wedge\theta^{(n)}_{\eta}}(S^{B})|. \nonumber
\end{eqnarray}
From Lemma 4.4 in \cite{DK2} it follows that there exists a
constants $C^{(4)},C^{(5)}$ such that
\begin{equation}\label{4.22}
\alpha_1+\alpha_2\leq
C^{(4)}(E^{B}(\theta^{(n)}_{\zeta\wedge\eta}-\theta^{(n)}_{\eta}
\wedge{\sigma})^2)^{1/2}+
C^{(5)}(E^{B}(\theta^{(n)}_{\zeta\wedge\eta}-\theta^{(n)}_{\eta}
\wedge{\sigma})^2)^{1/4}.
\end{equation}
By ({\ref{4.10}) we obtain that
$|\theta^{(n)}_{\zeta\wedge\eta}-\theta^{(n)}_{\eta}\wedge{\sigma}|\leq
|\theta^{(n)}_\zeta-\sigma|\leq |T-\theta^{(n)}_n|\leq{u_n}$. Thus
by (\ref{4.3}),
\begin{equation}\label{4.23}
\alpha_1+\alpha_2\leq{C^{(6)}n^{-1/4}}
\end{equation}
for some constant $C^{(6)}$. Finally, we estimate $P(\Theta)$.
Observe that $\sigma\wedge\theta^{(n)}_{\eta}\leq
\theta^{(n)}_{\zeta\wedge\eta}$, and so
\begin{eqnarray}\label{4.24}
&\Theta\subseteq\bigg\{
 {\frac{\sup_{0\leq{t}\leq{\sigma\wedge\theta^{(n)}_{\eta}}}
{S^B_t}}{\max_{0\leq{k}\leq{\zeta\wedge\eta}}S^{B,n}_{kT/n}}}>e^\epsilon\bigg\}
\bigcup \bigg\{
 {\frac{\inf_{0\leq{t}\leq{\sigma\wedge\theta^{(n)}_{\eta}}}
{S^B_t}}{\min_{0\leq{k}\leq{\zeta\wedge\eta}}S^{B,n}_{kT/n}}}<e^{-\epsilon}
\bigg\}\subseteq\\
&\bigg\{
 {\frac{\sup_{0\leq{t}\leq{\theta^{(n)}_{\zeta\wedge\eta}}}
{S^B_t}}{\max_{0\leq{k}\leq{\zeta\wedge\eta}}S^{B,n}_{kT/n}}}>e^\epsilon\bigg\}
\bigcup\bigg\{
 {\frac{\inf_{0\leq{t}\leq{\theta^{(n)}_{\zeta\wedge\eta}}}
{S^B_t}}{\min_{0\leq{k}\leq{\zeta\wedge\eta}}S^{B,n}_{kT/n}}}<e^{-\epsilon}
\bigg\}\subseteq\nonumber\\
&\bigg\{
 {{\max_{0\leq{k}\leq{n-1}}\sup_{\theta^{(n)}_k\leq{t}\leq\theta^{(n)}_{k+1}}
 \max(\frac{S^B_t}{S^{B,n}_{kT/n}},\frac{S^{B,n}_{kT/n}}{S^B_t})}}>e^\epsilon
\bigg\}\subseteq\nonumber\\
&\bigg\{
 {{\max_{0\leq{k}\leq{n-1}}\sup_{\theta^{(n)}_k\leq{t}\leq
 \theta^{(n)}_{k+1}}r|t-\frac{kT}{n}|+\kappa|B^{*}_t-B^{*}_{\theta^{(n)}_k}|>
 \epsilon
}} \bigg\}.\nonumber
\end{eqnarray}
Since $|B^{*}_t-B^{*}_{\theta^{(n)}_k}|\leq \sqrt{\frac{T}{n}}$ and
$|t-\frac{kT}{n}|\leq u_n+\frac{T}{n}$ for any $k<n$ and
$t\in{[\theta^{(n)}_k,\theta^{(n)}_{k+1}]}$ (where $u_n$ was defined
after (\ref{4.2})) then using the inequality
$(a+b)^{3}\leq{4(a^3+b^3)}$ for $a,b\geq{0}$ we obtain by
(\ref{4.3}) that
\begin{equation}\label{4.25}
E^B(
\max_{0\leq{k}\leq{n-1}}\sup_{\theta^{(n)}_k\leq{t}\leq\theta^{(n)}_{k+1}}r|t-
\frac{kT}{n}|+\kappa|B^{*}_t-B^{*}_{\theta^{(n)}_k}|)^{3}\leq
 C^{(7)} n^{-3/2}
\end{equation}
for some constant $C^{(7)}$. From (\ref{4.24}) and the Markov
inequality it follows that
$P(\Theta)\leq{C^{(7)}\frac{n^{-3/2}}{\epsilon^3}}$ and together
with (\ref{4.8}), (\ref{4.12}), (\ref{4.15}), (\ref{4.16}),
(\ref{4.20}) and (\ref{4.23}) we conclude that
\begin{eqnarray}\label{4.26}
&R^I_n(x)-R^{I_\epsilon}(x)\leq
3\delta+(C^{(6)}+C(f_\delta))n^{-1/4}+\tilde{C}(f_\delta)n^{-1/2}+\\
&{C^{(2)}n^{-1/4}(\ln{n})^{3/4}}+
C^{(3)}\sqrt{C^{(7)}\frac{n^{-3/2}}{\epsilon^3}}. \nonumber
\end{eqnarray}
Since the above constants do not depend on $n$ then
$R^{I_\epsilon}(x)\geq\limsup_{n\rightarrow\infty}R^I_n(x)-3\delta$.
Letting $\delta\downarrow 0$ we obtain that
$R^{I_\epsilon}(x)\geq\limsup_{n\rightarrow\infty}R^I_n(x)$ and by
Lemma \ref{lem3.3},
\begin{equation}\label{4.27}
R^{I}(x)=\lim_{\epsilon\rightarrow 0}R^{I_\epsilon}(x)\geq
\limsup_{n\rightarrow\infty}R^I_n(x).
\end{equation}

In order to compete the proof of Theorem \ref{thm2.2} we should
prove (\ref{2.19}). Fix an initial capital $x$, an open interval
$I=(L,R)$ and a natural number $n$. If
$\min(\frac{R}{S_0},\frac{S_0}{L})\leq e^{n^{-1/3}}$ then from Lemma
\ref{lem3.4} and the inequality $R^I_n(x)\geq (F_0(S_0)-x)^{+}$ it
follows
\begin{equation}\label{4.27+}
R^I(x)-R^I_n(x) \leq R^I(x)-(F_0(S_0)-x)^{+}\leq A_3n^{-1/4}.
\end{equation}
Next, we deal with the case where
$\min(\frac{R}{S_0},\frac{S_0}{L})>e^{n^{-1/3}}$ (which is true for
sufficiently large $n$). Introduce the open interval
$J_n=(L\exp(n^{-1/3}),R\exp(-n^{-1/3}))$. Set
$(\pi,\sigma)=(\psi_n(\pi^I_n),\phi_n(\sigma^I_n))$ where
$(\pi^I_n,\sigma^I_n)$ is the optimal hedge given by (\ref{3.15})
and the functions $\psi_n,\phi_n$ were defined in Section 2. We can
consider the portfolio $\pi=\psi_n(\pi_n)$ not only as an element in
$\mathcal{A}^{B,n}(x)$ but also as an element in
$\mathcal{A}^{B}(x)$ if we restrict the above portfolio to the
interval ${[0,T]}$. From Lemma \ref{lem3.2} we obtain that
\begin{equation}\label{4.28}
R^{J_n}(\pi,\sigma)-R^I_n(x)=R^{J_n}(\pi,\sigma)-R^{B,I}_n(\pi,\zeta^I_n)
\end{equation}
where, recall, $\zeta^I_n$ was defined in (\ref{3.16}). Since $I$
and $n$ are fixed we denote $\zeta=\zeta^I_n$. Recall that
$\Pi_n(\sigma^I_n)=\zeta$ and so from (\ref{2.20}) we get
$\sigma=(T\wedge\theta^{(n)}_{\zeta})\mathbb{I}_{\zeta<n}
+T\mathbb{I}_{\zeta=n}$. For a fixed $\delta>0$ choose a stopping
time $\tau$ such that
\begin{equation}\label{4.30}
R^{J_n}(\pi,\sigma)<\delta+E^{B}[(Q^{B,J_n}(\sigma,{\tau})-
\tilde{V}^{\pi}_{\sigma \wedge{{\tau}}})^{+}].
\end{equation}
 Observe that
$\min\{k|\theta^{(n)}_k\geq{\tau}\}\in\mathcal{T}^{B,n}$ since
$\{\min\{k|\theta^{(n)}_k\geq{\tau}\}\leq{j}\}=\{\theta^{(n)}_j\geq\tau\}
\in\mathcal{F}^{B}_{\theta^{(n)}_j}$ and set
$\eta=n\wedge\min\{k|\theta^{(n)}_k\geq{\tau}\}\in
\mathcal{T}^{B,n}_{0,n}$. Denote
\begin{eqnarray*}
&\Gamma_1=(Q^{B,J_n}(\sigma,{\tau})-Q^{B,J_n}(\sigma\wedge
\theta^{(n)}_n,{\tau}\wedge\theta^{(n)}_n))^+\\
&\mbox{and} \ \
\Gamma_2=(Q^{B,J_n}(\sigma\wedge\theta^{(n)}_n,\tau\wedge\theta^{(n)}_n)-
Q^{B,I,n}(\frac{\zeta{T}}{n},\frac{\eta{T}}{n}))^+.
\end{eqnarray*}
From (\ref{4.30}) it follows that
\begin{eqnarray*}
&R^{J_n}(\pi,\sigma)<\delta+E^{B}(Q^{B,J_n}(\sigma\wedge\theta^{(n)}_n,
{\tau}\wedge
\theta^{(n)}_n)-\tilde{V}^{\pi}_{\sigma\wedge{{\tau}}})^{+}+E^{B}\Gamma_1\\
&\mbox{and}  \ \ R^{B,I}_n(\pi,\zeta)\geq
E^{B}(Q^{B,J_n}(\sigma\wedge\theta^{(n)}_n,{\tau}
\wedge\theta^{(n)}_n)-\tilde{V}^{\pi}_{\theta^{(n)}_{\zeta\wedge\eta}})^{+}-
E^B{\Gamma_2}.
\end{eqnarray*}
Hence,
\begin{eqnarray}\label{4.31}
&R^{J_n}(\pi,\sigma)-R^{B,I}_n({\pi},\zeta)<E^{B}(Q^{B,J_n}(\sigma\wedge
\theta^{(n)}_n,\tau\wedge\theta^{(n)}_n)-\tilde{V}^{\pi}_{\sigma
\wedge{\tau}})^{+}\\
&-E^{B}(Q^{B,J_n}(\sigma\wedge\theta^{(n)}_n,{\tau}\wedge\theta^{(n)}_n)-
\tilde{V}^{\pi}_{\theta^{(n)}_{\zeta\wedge\eta}})^{+}+
\delta+E^{B}(\Gamma_1+ \Gamma_2). \nonumber
\end{eqnarray}
Observe that $\sigma\wedge\theta^{(n)}_n\leq\theta^{(n)}_\zeta$ and
$\tau\wedge\theta^{(n)}_n\leq\theta^{(n)}_\eta$, thus
\begin{equation}\label{4.32}
\sigma\wedge{\tau}\wedge{\theta^{(n)}_n}\leq\theta^{(n)}_{\zeta\wedge\eta}.
\end{equation}
Since $\pi\in{\mathcal{A}^{B,n}(x)}$ then by (\ref{2.21}),
$\tilde{V}^{\pi}_{\sigma\wedge\tau}=\tilde{V}^{\pi}_{\sigma\wedge\tau\wedge
\theta^{(n)}_n}=\tilde{E}^{B}(\tilde{V}^{\pi}_{\theta^{(n)}_{\zeta
\wedge\eta}}|\mathcal{F}^{B}_{\sigma\wedge\tau\wedge\theta^{(n)}_n})$.
This together with the Jensen inequality yields that
\begin{eqnarray}\label{4.33}
&(Q^{B,J_n}(\sigma\wedge\theta^{(n)}_n,{\tau}\wedge\theta^{(n)}_n)-
\tilde{V}^{\pi}_{\sigma\wedge{{\tau}}})^{+}\leq\\
&{\tilde{E}^{B}((Q^{B,J_n}(\sigma\wedge\theta^{(n)}_n,{\tau}\wedge
\theta^{(n)}_n)-
\tilde{V}^{\pi}_{\theta^{(n)}_{\zeta\wedge\eta}})^{+}|
\mathcal{F}^{B}_{\sigma\wedge\tau\wedge\theta^{(n)}_n})}=\nonumber\\
&{E}^{B}\bigg(\frac{Z_{\sigma\wedge\tau\wedge\theta^{(n)}_n}}
{Z_{\theta^{(n)}_{\zeta\wedge\eta}}}(Q^{B,J_n}(\sigma\wedge\theta^{(n)}_n,
{\tau}\wedge\theta^{(n)}_n)-
\tilde{V}^{\pi}_{\theta^{(n)}_{\zeta\wedge\eta}})^{+}
|\mathcal{F}^{B}_{\sigma\wedge\tau\wedge\theta^{(n)}_n}\bigg).\nonumber
\end{eqnarray}
Thus,
\begin{eqnarray}\label{4.34}
&E^{B}(Q^{B,J_n}(\sigma\wedge\theta^{(n)}_n,{\tau}\wedge\theta^{(n)}_n)-
\tilde{V}^{\pi}_{\sigma\wedge{{\tau}}})^{+}\leq\\
&{E}^{B}\bigg(\frac{Z_{\sigma\wedge\tau\wedge\theta^{(n)}_n}}
{Z_{\theta^{(n)}_{\zeta\wedge\eta}}}(Q^{B,J_n}(\sigma\wedge\theta^{(n)}_n,
{\tau}\wedge\theta^{(n)}_n)-
\tilde{V}^{\pi}_{\theta^{(n)}_{\zeta\wedge\eta}})^{+}\bigg).\nonumber
\end{eqnarray}
By (\ref{4.31}) and (\ref{4.34}) we obtain that
\begin{equation}\label{4.35}
R^{J_n}(\pi,\sigma)-R^{B,I}_n(\pi,\zeta)<\delta+E^{B}(\Gamma_1+\Gamma_2)+
\alpha_3
\end{equation}
where
\begin{equation*}
\alpha_3=E^B\bigg(\frac{Z_{\sigma\wedge\tau\wedge\theta^{(n)}_n}-
Z_{\theta^{(n)}_{\zeta\wedge\eta}}}{Z_{\theta^{(n)}_{\zeta\wedge\eta}}}
Q^{B,J_n}(\sigma\wedge\theta^{(n)}_n,{\tau}\wedge\theta^{(n)}_n)
\bigg).
\end{equation*}
Notice that $|\sigma-\theta^{(n)}_\zeta|\leq{w_n}$ and
$|\tau-\theta^{(n)}_\eta|\leq{w_n}$ (where $w_n$ was defined after
(\ref{4.2})). Thus by (\ref{4.32}) we obtain that
\begin{equation}\label{4.36}
0 \leq
\theta^{(n)}_{\zeta\wedge\eta}-\sigma\wedge{\tau}\wedge\theta^{(n)}_n\leq
\max(|\sigma-\theta^{(n)}_\zeta|,|\tau-\theta^{(n)}_\eta|)
\leq{w_n}.
\end{equation}
From Ito's formula it follows that
$dZ_t=\frac{\mu}{\kappa}Z_tdB_t+(\frac{\mu}{\kappa})^2Z_tdt$, and so
\begin{equation*}
Z_{\theta^{(n)}_{\zeta\wedge\eta}}-Z_{\sigma\wedge\tau\wedge\theta^{(n)}_n}=
\frac{\mu}{\kappa}\int_{\sigma\wedge\tau\wedge
\theta^{(n)}_n}^{\theta^{(n)}_{\zeta\wedge\eta}}Z_tdB_t
+(\frac{\mu}{\kappa})^2\int_{\sigma\wedge\tau\wedge
\theta^{(n)}_n}^{\theta^{(n)}_{\zeta\wedge\eta}}Z_tdt.
\end{equation*}
Set $E_n=\sup_{0\leq{t}\leq{\theta^{(n)}_n\vee{T}}}Z_t$. From
(\ref{4.3}), (\ref{4.4}), the Cauchy-Schwarz inequality and Ito's
isometry we obtain that
\begin{eqnarray}\label{4.37}
&E^{B}(Z_{\theta^{(n)}_{\zeta\wedge\eta}}-Z_{\sigma\wedge\tau\wedge
\theta^{(n)}_n})^2\leq{
2(\frac{\mu}{\kappa})^2E^{B}\int_{\sigma\wedge\tau\wedge
\theta^{(n)}_n}^{\theta^{(n)}_{\zeta\wedge\eta}}Z^2_tdt}+\\
&2(\frac{\mu}{\kappa})^4E^{B}(w_nE_n)^2\leq
{2(\frac{\mu}{\kappa})^2}E^{B}(w_nE^2_n)+2(\frac{\mu}{\kappa})^4E^{B}
(w_nE_n)^2\leq{C^{(8)}n^{-1/2}} \nonumber
\end{eqnarray}
for some constant $C^{(8)}$. By (\ref{2.3}) it follows that
$Q^{B,J_n}(\sigma\wedge\theta^{(n)}_n,{\tau}\wedge\theta^{(n)}_n)
\leq{G_0(S_0)+}$\\
$\cL(T+2)(1+\sup_{0\leq{t}\leq{T}}S^B_t)$, and so (\ref{4.4}) and
the Cauchy-Schwarz inequality yields that
\begin{equation}\label{4.38}
\alpha\leq{C^{(9)}n^{-1/4}}
\end{equation}
for some $C^{(9)}>0$ independent of $n$. Now we estimate
$E^{B}\Gamma_1$. Clearly $\Gamma_1\leq
(Q^{B}(\sigma,\tau)-Q^B(\sigma\wedge
\theta^{(n)}_n,{\tau}\wedge\theta^{(n)}_n))^+$. From the definitions
it follows easily that $\sigma<\tau$ is equivalent to
$\sigma\wedge\theta^{(n)}_n<\tau\wedge\theta^{(n)}_n$. Furthermore,
$\sigma\wedge\tau-\sigma\wedge\tau\wedge\theta^{(n)}_n\leq
|T-\theta^{(n)}_n|\leq u_n$ (with $u_n$ defined after (\ref{4.2})).
Thus from (\ref{4.3}) and Lemma 4.4 in \cite{DK2} we obtain that
there exists a constant $C^{(10)}$ such that for all $n\in\bbN$,
\begin{eqnarray}\label{4.39}
&E^{B}\Gamma_1\leq
E^B{|e^{-r\sigma\wedge\tau}G_{\sigma\wedge\tau}(S^B)-
e^{-r\theta^{(n)}_n\wedge\sigma\wedge{\tau}}G_{\theta^{(n)}_{n}
\wedge\sigma\wedge\tau}(S^{B})|}+\\
&E^B|e^{-r\sigma\wedge\tau}F_{\sigma\wedge\tau}(S^B)-
e^{-r\theta^{(n)}_n\wedge\sigma\wedge{\tau}}F_{\theta^{(n)}_{n}
\wedge\sigma\wedge\tau}(S^{B})|\leq \nonumber\\
&C^{(4)}(E^{B}(u_n)^2)^{1/2}+
C^{(5)}(E^{B}(u_n)^2)^{1/4}\leq{C^{(10)}n^{-1/4}}\nonumber
\end{eqnarray}
where the last inequality follows from (\ref{4.3}). Next, we
estimate $E^{B}\Gamma_2$. From (\ref{2.3}) it follows that
\begin{eqnarray}\label{4.40}
&\Gamma_2\leq
(Q^{B}(\sigma\wedge\theta^{(n)}_n,\tau\wedge\theta^{(n)}_n)-
Q^{B,n}(\frac{\zeta{T}}{n},\frac{\eta{T}}{n}))^+ \\
&+\mathbb{I}_{\tilde\Theta}(G_0(S_0)+\cL(T+2)(1+\sup_{0\leq{t}\leq{T}}S^B_t))
\nonumber
\end{eqnarray}
where
$\tilde\Theta={\{\eta\wedge\zeta\geq\tau^{B,n}_{I}\}}\cap{\{\sigma\wedge\tau
\wedge\theta^{(n)}_n<\tau_{J_n}\}}$. By the Cauchy-Schwarz
inequality,
\begin{equation}\label{4.41}
E^B\mathbb{I}_{\tilde\Theta}(G_0(S_0)+\cL(T+2)(1+\sup_{0\leq{t}\leq{T}}S^B_t))
\leq C^{(11)}(P(\tilde\Theta))^{1/2}
\end{equation}
for some constant $C^{(11)}$ independent of $n$. From (\ref{4.18}),
(\ref{4.40}) and (\ref{4.41}),
\begin{eqnarray}\label{4.41+}
&E^B\Gamma_2\leq
E^B(Q^{B}(\sigma\wedge\theta^{(n)}_n,\tau\wedge\theta^{(n)}_n)-
Q^{B,n}(\frac{\zeta{T}}{n},\frac{\eta{T}}{n}))^++\\
&C^{(11)}(P(\tilde\Theta))^{1/2} \leq
C^{(2)}n^{-1/4}(\ln{n})^{3/4}+C^{(11)}(P(\tilde\Theta))^{1/2}+\nonumber\\
&E^B(Q^{B}(\sigma\wedge\theta^{(n)}_n,\tau\wedge\theta^{(n)}_n)-
Q^{B}(\theta^{(n)}_{\zeta}, \theta^{(n)}_{\eta}))^{+}.\nonumber
\end{eqnarray}
From the definitions it follows easily that if
$\sigma\wedge\theta^{(n)}_n<\tau\wedge\theta^{(n)}_n$ then
$\zeta<\eta$. Hence, from (\ref{4.3}), (\ref{4.36}) and Lemma 4.4 in
\cite{DK2} we obtain that
\begin{eqnarray}\label{4.42}
&E^B(Q^{B}(\sigma\wedge\theta^{(n)}_n,\tau\wedge\theta^{(n)}_n)-Q^{B}
(\theta^{(n)}_{\zeta},
\theta^{(n)}_{\eta}))^{+}\leq\\
&E^B|e^{-r(\sigma\wedge\tau\wedge\theta^{(n)}_n)}G_{\sigma\wedge\tau\wedge
\theta^{(n)}_n}(S^B)-
e^{-r\theta^{(n)}_{\zeta\wedge\eta}}G_{\theta^{(n)}_{\zeta\wedge\eta}}
(S^{B})|+\nonumber\\
&E^B|e^{-r(\sigma\wedge\tau\wedge\theta^{(n)}_n)}F_{\sigma\wedge\tau
\wedge\theta^{(n)}_n}(S^B)-
e^{-r\theta^{(n)}_{\zeta\wedge\eta}}F_{\theta^{(n)}_{\zeta\wedge\eta}}
(S^{B})|\leq\nonumber\\
&C^{(4)}(E^{B}(w_n)^2)^{1/2}+ C^{(5)}(E^{B}(w_n)^2)^{1/4}
\leq{C^{(12)} n^{-1/4}}\nonumber
\end{eqnarray}
for some constant $C^{(12)}$. Finally, we estimate
$P(\tilde\Theta)$. Observe that
$\sigma\wedge\tau\wedge\theta^{(n)}_n
\geq\theta^{(n)}_{(\zeta\wedge\eta-1)^+}$. Indeed, from the
definitions it follows that $\tau\geq\theta^{(n)}_{(\eta-1)^+}$. If
$\sigma=T$ then
$\sigma\wedge\tau\wedge\theta^{(n)}_n=\tau\wedge\theta^{(n)}_n\geq
\theta^{(n)}_{(\eta-1)^+}\geq\theta^{(n)}_{(\zeta\wedge\eta-1)^+}$.
If $\sigma<T$ then $\sigma=\theta^{(n)}_{\zeta}$, and so
$\sigma\wedge\tau\wedge\theta^{(n)}_n\geq\theta^{(n)}_{\zeta}\wedge
\theta^{(n)}_{(\eta-1)^+} \geq\theta^{(n)}_{(\zeta\wedge\eta-1)^+}$.
 Thus
\begin{eqnarray}\label{4.43}
&\tilde\Theta\subseteq\bigg\{
 {\frac{\max_{0\leq{k}\leq{\zeta\wedge\eta}}S^{B,n}_{kT/n}}
 {\sup_{0\leq{t}\leq{\sigma\wedge\tau\wedge\theta^{(n)}_n}}
{S^B_t}}}>e^{n^{-1/3}}\bigg\}\bigcup  \bigg\{
{\frac{\min_{0\leq{k}\leq{\zeta\wedge\eta}}S^{B,n}_{kT/n}}
 {\inf_{0\leq{t}\leq{\sigma\wedge\tau\wedge\theta^{(n)}_n}}
{S^B_t}}}<e^{-n^{-1/3}}\bigg\}\\
 &\subseteq\bigg\{
 {\frac{\max_{0\leq{k}\leq{\zeta\wedge\eta}}S^{B,n}_{kT/n}}
 {\max_{0\leq{k}\leq{(\zeta\wedge\eta-1)^+}}
{S^B_{\theta^{(n)}_k}}}}>e^{n^{-1/3}}\bigg\}\ \bigcup \bigg\{
{\frac{\min_{0\leq{k}\leq{\zeta\wedge\eta}}S^{B,n}_{kT/n}}
 {\min_{0\leq{k}\leq{(\zeta\wedge\eta-1)^+}}
{S^B_{\theta^{(n)}_k}}}}<e^{-n^{-1/3}}\bigg\}\subseteq
\nonumber\\
&\bigg\{
 \max_{0\leq{k}\leq{n-1}}
 \max\bigg(\frac{S^B_{\theta^{(n)}_{k+1}}}{S^{B,n}_{kT/n}},
 \frac{S^{B,n}_{kT/n}}{S^B_{\theta^{(n)}_{k+1}}}\bigg)>e^{n^{-1/3}}
\bigg\}\subseteq \bigg\{
 \max_{0\leq{k}\leq{n-1}}\big(r|\theta^{(n)}_{k+1}-\frac{kT}{n}|\nonumber\\
 &+\kappa|B^{*}_{\theta^{(n)}_{k+1}}-B^{*}_{\theta^{(n)}_k}|\big)>n^{-1/3}
  \bigg\} \subseteq\{ r(u_n+w_n)+\kappa\sqrt\frac{T}{n}>n^{-1/3}\}. \nonumber
\end{eqnarray}
From (\ref{4.3}), (\ref{4.43}) and the Markov inequality it follows
that
\begin{equation}\label{4.44}
P(\tilde\Theta)\leq
nE^B\bigg(r(u_n+w_n)+\kappa\sqrt\frac{T}{n}\bigg)^3\leq
C^{(13)}n^{-1/2}
\end{equation}
for some constant $C^{(13)}$ independent of $n$. Since $\delta$ is
arbitrary then combining (\ref{4.28}), (\ref{4.35}), (\ref{4.38}),
(\ref{4.39}), (\ref{4.41+}) and (\ref{4.44}) we conclude that there
exists a constant $C^{(14)}$ such that
\begin{equation}\label{4.44+}
R^{J_n}(\pi,\sigma)-R^I_n(x)=R^{J_n}(\pi,\sigma)-R^{B,I}_n(\pi,\zeta)
\leq C^{(14)}n^{-1/4}(\ln{n})^{3/4}.
\end{equation}
By (\ref{4.44+}) and Lemma \ref{lem3.3} it follows that for $n$
which satisfy $\min(\frac{R}{S_0},\frac{S_0}{L})>e^{n^{-1/3}}$ we
have
\begin{equation}\label{4.45}
R^I(x)-R^I_n(x)\leq
R^I(x)-R^{J_n}(x)+R^{J_n}(\pi,\sigma)-R^I_n(x)\leq
A_1n^{-1/4}+C^{(14)}n^{-1/4}(\ln{n})^{3/4}.
\end{equation}
From (\ref{4.27+}) and (\ref{4.45}) we derive (\ref{2.19}) and
complete the proof of Theorem \ref{thm2.2}.

Next, we prove Theorem \ref{thm2.4}. Let $H=(L,R)$ be an open
interval as before and for any $n$ set
$H_n=(L\exp(-n^{-1/3}),R\exp(n^{-1/3}))$. Fix $n$ and let
$(\pi^{H_n}_n,\sigma^{H_n}_n)\in\mathcal{A}^{\xi,n}(x)\times
\mathcal{T}^{\xi}_{0n}$ be the optimal hedge given by (\ref{3.15}).
Using (\ref{4.44+}) for $I=H_n$ we obtain that
\begin{equation}\label{4.46}
R^H(\psi_n(\pi^{H_n}_n),\phi_n(\sigma^{H_n}_n))\leq
R^{H_n}_n(x)+A_1n^{-1/4}+C^{(14)}n^{-1/4}(\ln{n})^{3/4}.
\end{equation}
Thus
\begin{equation}\label{4.47}
\limsup_{n\rightarrow\infty}R^H(\psi_n(\pi^{H_n}_n),\phi_n(\sigma^{H_n}_n))
\leq \limsup_{n\rightarrow\infty}R^{H_n}_n(x).
\end{equation}
For any $\epsilon>0$ denote
$J_{\epsilon}=(Le^{-\epsilon},Re^\epsilon)$. Since $H_n\subseteq
J_\epsilon$
 for sufficiently large $n$ then from (\ref{4.47}) and
Theorem \ref{thm2.2} we obtain that for any $\epsilon>0$,
\begin{equation}\label{4.48}
\limsup_{n\rightarrow\infty}R^H(\psi_n(\pi^{H_n}_n),\phi_n(\sigma^{H_n}_n))
\leq
\limsup_{n\rightarrow\infty}R^{J_\epsilon}_n(x)=R^{J_\epsilon}(x).
\end{equation}
By (\ref{4.48}) and Lemma \ref{lem3.3},
\begin{equation}\label{4.49}
\limsup_{n\rightarrow\infty}R^H(\psi_n(\pi^{H_n}_n),\phi_n(\sigma^{H_n}_n))
\leq \lim_{\epsilon\rightarrow{0}}R^{J_\epsilon}(x)=R^{H}(x)
\end{equation}
which completes the proof of Theorem \ref{thm2.4}.

Next, we prove Theorem \ref{thm2.1}. Let $I=(L,R)$ be an open
interval as before. Assume that $\mu=0$. In this case
$P^B=\tilde{P}^B$ and $\tilde{P}^{\xi}_n=P^{\xi}_n$ for any $n$.
Thus $\cV^I=R^I(0)$ and $\cV^I_n=R^I_n(0)$. Hence, using the same
procedure as in first part of the proof of Theorem \ref{thm2.2} and
taking into account that the value of the portfolios $\pi,\pi'$ is
zero (which means that $C(f_\delta)=\tilde{C}(f_\delta)=0$ and so we
can let $\delta\downarrow{0}$ in (\ref{4.26})) we obtain that there
exist constants $C^{(15)}$ and $C^{(16)}$ such that for any
$\epsilon>0$,
\begin{equation}\label{4.50}
\cV^I_n-\cV^{I_\epsilon}\leq {C^{(15)}n^{-1/4}(\ln{n})^{3/4}}+
C^{(16)}\sqrt{\frac{n^{-3/2}}{\epsilon^3}}
\end{equation}
where $I_\epsilon=(Le^{-\epsilon},Re^{\epsilon})$. Taking
$\epsilon=n^{-1/3}$ we obtain by (\ref{3.19}) and (\ref{4.50}) that
\begin{equation}\label{4.51}
\cV^I_n-\cV^I\leq C^{(15)}n^{-1/4}(\ln{n})^{3/4}+
(C^{(16)}+A_2)n^{-1/4}.
\end{equation}
From Theorem \ref{thm2.2} it follows that there exists a constant
$C^{(17)}$ such that
\begin{equation*}
\cV^I-\cV^I_n=R^I(0)-R^I_n(0)\leq C^{(17)}n^{-1/4}(\ln{n})^{3/4}.
\end{equation*}
This together with (\ref{4.51}) completes the proof of Theorem
\ref{thm2.1}.

Finally, we prove Theorem \ref{thm2.3}. Let $H=(L,R)$ be an open
interval and $n$ be a natural number. Set
$H_n=(L\exp(-n^{-1/3}),R\exp(n^{-1/3}))$ and let
$(\pi_n,\sigma_n)\in\mathcal{A}^{\xi,n}(\cV^{H_n}_n)\times\mathcal{T}^\xi_{0n}$
be a perfect hedge for a double barrier option in the $n$--step CRR
market with the barriers $L\exp(-n^{-1/3}),R\exp(n^{-1/3})$, i.e.
for any $k\leq{n}$,
\begin{equation}\label{4.52}
\tilde{V}^{\pi_n}_{\sigma_n\wedge{k}}\geq Q^{H_n,n}(\sigma_n,k).
\end{equation}
Set
$(\pi,\zeta)=(\psi_n(\pi_n),\Pi_n(\sigma_n))\in\mathcal{A}^{B,n}(V^{H_n}_n)
\times\mathcal{T}^{B,n}_{0,n}$. From (\ref{4.52}) and the definition
of $\Pi_n$ we obtain that for any $k\leq{n}$,
\begin{equation}\label{4.53}
\tilde{V}^{\pi}_{\zeta\wedge{k}}=\Pi_n(\tilde{V}^{\pi_n}_{\sigma_n\wedge{k}})
\geq \Pi_n(Q^{H_n,n}(\sigma_n,k))=Q^{B,H_n,n}(\zeta,k)
\end{equation}
implying that $R^{B,H_n}_n(\pi,\zeta)=0$. Set
$\sigma=\phi_n(\sigma_n)\in\mathcal{T}^{B}_{0T}$ then
$\sigma=(T\wedge\theta^{(n)}_{\zeta})\mathbb{I}_{\zeta<n}+
T\mathbb{I}_{\zeta=n}$. Hence, using (\ref{4.44+}) for $I=H_n$ we
obtain that
\begin{equation}\label{4.54}
R^H(\pi,\sigma)\leq R^{B,H_n}_n(\pi,\zeta)+
C^{(14)}n^{-1/4}(\ln{n})^{3/4}=C^{(14)}n^{-1/4}(\ln{n})^{3/4}
\end{equation}
completing the proof.
\begin{rem}\label{rem4.2}
Consider another definition of the discounted payoff function where
in place of (\ref{2.8}) we set
\begin{equation}\label{4.55}
Q^{B,I}_1(t,s)=e^{-r(t\wedge{s})}(G_t(S^B)\mathbb{I}_{s<t}+
Y^I_t\mathbb{I}_{t\leq{s}})
\end{equation}
which means that the seller pays for cancellation an amount which
does
 not depend on the barriers. For such discounted payoff function
the option price will be equal to the original option price $\cV^I$
given by (\ref{2.8++}) and for any initial capital $x$ the shortfall
risk will be equal to $R^I(x)$ given by (\ref{2.11++}). Indeed, the
terms in the formula (\ref{4.55}) for the discounted payoff function
are not less than the corresponding terms for the payoff function
given by (\ref{2.8}). On the other hand, for any
$\pi\in\mathcal{A}^{B}$ and $\sigma\in\mathcal{T}^{B}_{0T}$,
\begin{equation*}
(Q^{B,I}(\sigma,\tau)-\tilde{V}^{\pi}_{\sigma\wedge\tau})^+=
(Q^{B,I}_1(\tilde\sigma,\tau)-\tilde{V}^{\pi}_{\tilde\sigma\wedge\tau})^+
\end{equation*}
where
$\tilde\sigma=\sigma\mathbb{I}_{\sigma<\tau_I}+T\mathbb{I}_{\sigma\geq\tau_I}$.
Thus for any portfolio $\pi\in\mathcal{A}^{B}$,
\begin{equation}\label{4.56}
R^I(\pi)=\inf_{\sigma\in\mathcal{T}^{B}_{0T}}\sup_{\tau\in\mathcal{T}^{B}_{0T}}
{E}^B(Q^{B,I}(\sigma,\tau)-\tilde{V}^{\pi}_{\sigma\wedge\tau})^+\geq
\inf_{\sigma\in\mathcal{T}^{B}_{0T}}\sup_{\tau\in\mathcal{T}^{B}_{0T}}
{E}^B(Q^{B,I}_1(\sigma,\tau)-\tilde{V}^{\pi}_{\sigma\wedge\tau})^+
\end{equation}
and we conclude that (\ref{4.56}) is, in fact, an equality which
proves that for a given portfolio the shortfall risk remains as
before. Since option prices can be represented as shortfall risk for
the case where $P=\tilde{P}$ and the initial capital is $0$ then it
follows that the option price remains as before, as well. The same
holds true for CRR markets. We note that the proof of our main
results for the discounted payoff function
 given by (\ref{4.55}) becomes a bit simpler than for the original definitions
 but the latter seem to more natural.
\end{rem}

\section{The knock--in case}\label{sec5}\setcounter{equation}{0}
In this section we present results similar to Theorems
\ref{thm2.1}--\ref{thm2.4}  (with a little bit different estimates)
for knock--in barrier options. For a given open interval $I=(L,R)$
the payoff processes in the BS model and the $n$--step CRR market
are defined in this case by
\begin{equation}\label{5.1}
\mathbb{X}_t=G_t(S^B), \
\mathbb{Y}^I_t=F_t(S^B)\mathbb{I}_{t\geq\tau_I} \ \ \mbox{and} \ \
\mathbb{X}^{(n)}_k=G_{\frac{kT}{n}}(S^B), \
\mathbb{Y}^{I,n}_k=F_{\frac{kT}{n}}(S^B)\mathbb{I}_{k\geq\tau^{(n)}_I},
\end{equation}
respectively. Notice that the seller will pay for cancellation an
amount which does not depend on the barriers. If we would define the
high payoff process $\mathbb{X}^I_t$, $t\geq{0}$ in a way similar
 to the low payoff process $\mathbb{Y}^I_t$, $t\geq{0}$, namely,
 $\mathbb{X}^I_t=G_t(S^B)\mathbb{I}_{t\geq\tau_I}$
then the seller could cancel the contract at the moment $t=0$
without paying anything to the buyer which would make such contract
 worthless.

Now, for the BS model we define the option price and the shortfall
risks by
\begin{eqnarray}\label{5.2}
&\tilde\cV^I=\inf_{\sigma\in\mathcal{T}^{B}_{0T}}
\sup_{\tau\in\mathcal{T}^{B}_{0T}}
\tilde{E}^B\tilde{Q}^{B,I}(\sigma,\tau),   \
\tilde{R}^{I}(\pi,\sigma)=\sup_{\tau\in\mathcal{T}^B_{0T}}
E^B(\tilde{Q}^{B,I}(\sigma,\tau)\quad\quad\\
&-\tilde{V}^{\pi}_{\sigma\wedge\tau})^+, \ \
\tilde{R}^{I}(\pi)=\inf_{\sigma\in\mathcal{T}^{B}_{0T}}\tilde{R}^{I}(\pi,\sigma)
\ \ \mbox{and} \ \ \tilde{R}^{I}(x)=\inf_{\pi\in\mathcal{A}^B(x)}
\tilde{R}^{I}(\pi)\nonumber
\end{eqnarray}
where $\tilde{Q}^I(t,s)=e^{-r(t\wedge
s)}(\mathbb{X}_t\mathbb{I}_{t<s}+\mathbb{Y}^I_t\mathbb{I}_{s\leq{t}})$
is the discounted payoff function. For the $n$--step CRR market the
corresponding definitions are
\begin{eqnarray}\label{5.3}
&\tilde\cV^{I}_n=\min_{\zeta\in\mathcal{T}^{\xi}_{0n}}\max_{\eta\in
\mathcal{T}^{\xi}_{0n}}\tilde{E}^{\xi}_n\tilde{Q}^{I,n}(\zeta,\eta),
\
\tilde{R}^I_n(\pi,\sigma)=\max_{\tau\in\mathcal{T}^\xi_{0n}}E^{\xi}_n(Q
^{I,n}(\sigma,\tau)\quad\quad\\
&-\tilde{V}^{\pi}_{\sigma\wedge\tau})^+, \ \
\tilde{R}^I_n(\pi)=\min_{\sigma\in\mathcal{T}^{\xi}_{0n}}\tilde{R}^I_n(\pi,
\sigma)\ \ \mbox{and} \ \ \tilde{R}^I_n(x)=
\inf_{\pi\in\mathcal{A}^{\xi,n}(x)}\tilde{R}^I_n(\pi) \nonumber
\end{eqnarray}
where $\tilde{Q}^{I,n}(k,l)=(1+r_n)^{-k\wedge
l}(\mathbb{X}^{(n)}_k\mathbb{I}_{k<l}+\mathbb{Y}^{
I,n}_l\mathbb{I}_{l\leq{k}})$ is the discounted payoff function.
Denote also by $Q^{(n)}(k,l)=(1+r_n)^{-k\wedge
l}(G_{\frac{kT}{n}}(S^{n})\mathbb{I}_{k<l}+
F_{\frac{lT}{n}}(S^{n})\mathbb{I}_{l\leq k})$ the regular payoff and
let $\cV_n=\min_{\zeta\in\mathcal{T}^{\xi}_{0n}}\max_{\eta\in
\mathcal{T}^{\xi}_{0n}}\tilde{E}^{\xi}_n{Q}^{(n)}(\zeta,\eta)$ be
the option price for this payoff.
\begin{thm}\label{thm5.1} Let $I=(L,R)$ be an open interval.

(i) For each $\epsilon>0$ there exists a constant
$\tilde{C}_{1,\epsilon}$ such that for any $n\in\mathbb{N}$,
\begin{equation}\label{5.4}
|\tilde\cV^I-\tilde\cV^{I}_n|\leq \tilde{C}_{1,\epsilon}
n^{-\frac{1}{4}+\epsilon}.
\end{equation}
(ii) For each initial capital $x$,
\begin{equation}\label{5.5}
lim_{n\rightarrow\infty}\tilde{R}^I_n(x)=\tilde{R}^I(x).
\end{equation}
Furthermore, for each $\epsilon>0$ there exists a constant
$\tilde{C}_{2,\epsilon}$ such that for any $x$ and $n\in\mathbb{N}$,
\begin{equation}\label{5.6}
\tilde{R}^I(x)\leq \tilde{R}^I_n(x)+\tilde{C}_{2,\epsilon}
n^{-\frac{1}{4}+\epsilon}.
\end{equation}
(iii) For each $n\in\bbN$ let
$(\pi^p_n,\sigma^p_n)\in\mathcal{A}^{\xi,n}(\tilde\cV^{I}_n)\times
\mathcal{T}^\xi_{0n}$ be a perfect hedge for a double barrier
knock-in option as above in the $n$--step CRR market with the
barriers $L,R$. Then for any $\epsilon>0$ and $n\in\mathbb{N}$,
\begin{equation}\label{5.7}
\tilde{R}^I(\psi_n(\pi^p_n),\phi_n(\sigma^p_n))\leq\tilde{C}_{2,\epsilon}
n^{-\frac{1}{4}+\epsilon}.
\end{equation}
(iv) For any $n\in\bbN$ let
$(\tilde\pi^I_n,\tilde\sigma^I_n)\in\mathcal{A}^{\xi,n}(x)\times
\mathcal{T}^{\xi}_{0n}$ be the optimal hedge which is given by
(\ref{5.13}) below. Then
\begin{equation}\label{5.8}
lim_{n\rightarrow\infty}\tilde{R}^{I}(\psi_n(\tilde\pi^I_n),\phi_n(\tilde
\sigma^I_n))=\tilde{R}^I(x).
\end{equation}
All the constants above are not depend on the interval $I$.
\end{thm}

In order to prove Theorem 5.1 we should establish a result similar
to Lemma \ref{lem3.2}. For each open interval $H$ set
$\mathbb{Y}^{B,H,n}_k=F_{\frac{kT}{n}}(S^{B,n})\mathbb{I}_{k\geq\tau^{B,n}_H}$,
$\mathbb{X}^{B,n}_k=G_{\frac{kT}{n}}(S^{B,n})$ and
$\tilde{Q}^{H,B,n}(k,l)=(1+r_n)^{-k\wedge l}(\mathbb{X}^{B,n}_k
\mathbb{I}_{k<l}+\mathbb{Y}^{ H,B,n}_l\mathbb{I}_{l\leq{k}})$,
$k,l\leq{n}$. Similarly to (\ref{3.8})--(\ref{3.10}) define the
shortfall risk by
\begin{eqnarray}\label{5.9}
&\tilde{R}^{B,H}_n(\pi,\zeta)=\sup_{\eta\in\mathcal{T}^{B,n}_{0n}}
E^B(\tilde{Q}^{B,H,n}
(\zeta,\eta)-\tilde{V}^{\pi}_{\theta^{(n)}_{\zeta\wedge\eta}}) ^+,\\
&\tilde{R}^{B,H}_n(\pi)=\inf_{\zeta\in\mathcal{T}^{B,n}_{0n}}
\tilde{R}^{B,H}_n(\pi,\zeta)\,\,\mbox{and}\,\,
\tilde{R}^{B,H}_n(x)=\inf_{\pi\in\mathcal{A}^{B,n}(x)}\tilde{R}^{B,H}_n(\pi).
\nonumber
\end{eqnarray}
Similarly to (\ref{3.9}) and (\ref{3.10}) for any
$\pi\in{\mathcal{A}^{B,n}}$ set
\begin{eqnarray}\label{5.10}
&\tilde{U}^{H,\pi}_n=((1+r_n)^{-n}\mathbb{Y}^{B,H,n}_n-
\tilde{V}^\pi_{\theta^{(n)}_n})^+,
 \
 \tilde{U}^{H,\pi}_k=\min\bigg(((1+r_n)^{-k}\mathbb{X}^{B,n}_k-\\
&\tilde{V}^\pi_{\theta^{(n)}_k})^+
,\max\bigg(((1+r_n)^{-k}\mathbb{Y}^{B,H,n}_k-\tilde{V}^{\pi}_{\theta^{(n)}_k})^+,
E^{B}(\tilde{U}^{H,\pi}_{k+1}|\mathcal{F}^{B}_{\theta^{(n)}_{k}})\bigg)\bigg),
\ \ k<n \nonumber\\
&\mbox{and} \ \
\tilde\zeta(H,\pi)=\min{\{k|((1+r_n)^{-k}\mathbb{X}^{B,n}_k-
\tilde{V}^{\pi}_{\theta^{(n)}_k})^{+}=\tilde{U}^{H,\pi}_{k}\}}\wedge{n}.
\nonumber
\end{eqnarray}
Similarly to (\ref{3.2}) and (\ref{3.4}) for any
$\pi\in\mathcal{A}^{\xi,n}$
 define
\begin{eqnarray}\label{5.11}
&\tilde{W}^{H,\pi}_n=((1+r_n)^{-n}\mathbb{Y}^{H,n}_n-\tilde{V}^\pi_n)^+,
\ \
\tilde{W}^{H,\pi}_k=\min\bigg(((1+r_n)^{-k}\mathbb{X}^{(n)}_k-\\
&\tilde{V}^{\pi}_k)^{+}
,\max\bigg(((1+r_n)^{-k}\mathbb{Y}^{H,n}_k-\tilde{V}^{\pi}_k)^{+},
E^{\xi}_n(\tilde{W}^{H,\pi}_{k+1} |\mathcal{F}^{\xi}_k)\bigg)\bigg),
\ \
k<n \nonumber\\
&\mbox{and} \ \
\tilde\sigma(H,\pi)=\min{\{k|(1+r_n)^{-k}\mathbb{X}^{(n)}_k-\tilde{V}^{\pi}_k)^{+}=
\tilde{W}^{H,\pi}_k\}}\wedge{n}. \nonumber
\end{eqnarray}
 For $k\leq{n}$ and $x_1,...,x_k$ set
\begin{equation*}
\tilde{q}^{H,n}_k(x_1,...,x_k)=1-\mathbb{I}_{{[\min_{0\leq{i}\leq{k}}
\psi^{x_1,...,x_i}(\frac{iT}{n}),
 \max_{0\leq{i}\leq{k}}\psi^{x_1,...,x_i} (\frac{iT}{n})]
\subset{(L,R)}}}
\end{equation*}
with the functions $\psi^{x_1,...,x_i}$ introduced after
(\ref{3.11}). Similarly to (\ref{3.13}) define a sequence
${\{\tilde{J}^{H,n}_k\}}_{k=0}^n$ of functions
$\tilde{J}^{H,n}_k:[0,\infty)\times\mathbb{R}^k\rightarrow
\mathbb{R}$ by the following backward recursion
\begin{eqnarray}\label{5.12}
&\tilde{J}^{H,n}_n(y,u_1,u_2...,u_n)=(f^n_n(u_1,...,u_n)\tilde{q}^{H,n}_n(u_1,...,u_n)-y)^+
\  \ \mbox{and}\\
&\tilde{J}^{H,n}_k(y,u_1,...,u_k)=\min\bigg((g^n_k(u_1,...,u_k)
-y)^+\nonumber,
\max\bigg((f^n_k(u_1,...,u_k)\tilde{q}^{H,n}_k(u_1,...,u_k)\nonumber\\
&-y)^+,\inf_{u\in K_n(y)}\big(p^{(n)}
\tilde{J}^{H,n}_{k+1}(y+ua^{(n)}_1,u_1,...,u_k,\sqrt{\frac{T}{n}})+(1-p^{(n)})\times\nonumber\\
&\tilde{J}^{H,n}_{k+1}(y+ua^{(n)}_2,u_1,...,u_k,-\sqrt{\frac{T}{n}})\big)
\bigg)\bigg)  \ \mbox{for}  \ k=n-1,n-2,...,0.\nonumber
\end{eqnarray}
Set also $\tilde{h}^{H,n}_k(y,x_1,...,x_k)=argmin_{u\in K_n(y)}
\big(p^{(n)}\tilde{J}^{H,n}_{k+1}(y+ua^{(n)}_1,u_1,
...,u_k,\sqrt{\frac{T}{n}})+
(1-p^{(n)})\tilde{J}^{H,n}_{k+1}(y+ua^{(n)}_2,u_1,...,u_k,-\sqrt{\frac{T}{n}})\big)$.
Finally, for any initial capital $x$ define the hedges
$(\tilde\pi^{H}_n,\tilde\sigma^{H}_n)\in{\mathcal{A}^{\xi,n}(x)\times
\mathcal{T}^{\xi}_{0n}}$ and $(\tilde\pi^{B,H}_n,\tilde\zeta^{H}_n)
\in\mathcal{A}^{B,n}(x)\times\mathcal{T}^{B,n}_{0,n}$ by
\begin{eqnarray}\label{5.13}
 &\tilde{V}^{\tilde\pi^{H}_n}_0=x, \ \  \tilde{V}^{\tilde\pi^{H}_n}_{k+1}=
 \tilde{V}^{\tilde\pi^{H}_n}_k+\tilde{h}^{H,n}_k(\tilde{V}^{\tilde
 \pi^{H}_n}_k,e^{\kappa\sqrt
 \frac{T}{n}
 \xi_1},...,e^{\kappa\sqrt\frac{T}{n}\xi_k})(e^{\kappa\sqrt\frac{T}{n}
 \xi_{k+1}}-1)\\
 & \mbox{for} \ \
 k>0, \ \ \tilde\sigma^{H}_n=\tilde\sigma(H,\tilde\pi^{H}_n) \ \ \mbox{and} \ \
 \tilde\pi^{B,H}_n=\psi_n(\tilde\pi^{H}_n),\ \
\tilde\zeta^{H}_n=\Pi_n(\tilde\sigma^{H}_n).  \nonumber
\end{eqnarray}
Using the same arguments as in Section 3 we obtain
\begin{equation}\label{5.14}
\tilde{R}^H_n(x)=\tilde{R}^H_n(\tilde\pi^{H}_n,\tilde\sigma^{H}_n)=
\tilde{J}^{H,n}_0(x)=
\tilde{R}^{B,H}_n(\tilde\pi^{B,H}_n,\tilde\zeta^{H}_n)=\tilde{R}^{B,H}_n(x).
\end{equation}
Next we derive estimates in the spirit of Lemmas \ref{lem3.3} and
\ref{lem3.4}.
\begin{lem}\label{lem5.1}
For any $\gamma>1$ there exists a constant $A_\gamma$ such that for
any open interval $H=(L,R)$, $\epsilon>0$ and a hedge
$(\pi,\sigma)\in\mathcal{A}^B\times\mathcal{T}^{B}_{0T}$
\begin{equation}\label{5.15}
 \tilde{R}^{H}(\pi,\sigma)-\tilde{R}^{H_\epsilon}(\pi,\sigma)\leq
A_\gamma \epsilon^{1/\gamma}
\end{equation}
where $H_\epsilon=(Le^{-\epsilon},Re^\epsilon)$.
\end{lem}
\begin{proof}
Choose an open interval $H=(L,R)$, $\epsilon>0$ and a hedge
$(\pi,\sigma)\in\mathcal{A}^{B}(x)\times\mathcal{T}^{B}_{0T}$.
 Since $(\tilde{Q}^H(\sigma,\tau)-\tilde{V}^{\pi}_{\sigma
\wedge\tau})^{+}\leq(\tilde{Q}^H(\sigma,\tau\vee(\tau_H\wedge{T}))-
\tilde{V}^{\pi}_{\sigma\wedge(\tau\vee(\tau_H\wedge{T}))})^{+}$ for
any $\tau\in\mathcal{T}^{B}_{0T}$ then for each $\delta>0$ there
exists a stopping time $\tau_1\in\mathcal{T}^B_{0T}$ such that
\begin{equation}\label{5.16}
\tilde{R}^H(\pi,\sigma)<E^B(\tilde{Q}^H(\sigma,\tau_1)-
\tilde{V}^{\pi}_{\sigma\wedge\tau_1})^{+}+\delta \ \ \mbox{and} \ \
\tau_1\geq\tau_H\wedge{T}.
\end{equation}
Set $\tau_2=\tau_1\vee(\tau_{H_\epsilon}\wedge{T})$ and
$\Gamma=(\tilde{Q}^H(\sigma,\tau_1)-\tilde{Q}^{H_\epsilon}(\sigma,\tau_2))^{+}$.
Since ${\{\tilde{V}^{\pi}_t\}}_{t=0}^T$ is a supermartingale (with
respect to the martingale measure) then by Jensen's inequality,
\begin{eqnarray*}
&(\tilde{Q}^H(\sigma,\tau_1)-\tilde{V}^{\pi}_{\sigma\wedge\tau_1})^{+}\leq
\tilde{E}^B(\tilde{Q}^H(\sigma,\tau_1)-\tilde{V}^{\pi}_{\sigma\wedge
\tau_2})^{+}|\mathcal{F}^B_{\sigma\wedge\tau_1})=\nonumber\\
&E^B\bigg(\frac{Z_{\sigma\wedge\tau_1}}{Z_{\sigma\wedge\tau_2}}
(\tilde{Q}^H(\sigma,\tau_1)-\tilde{V}^{\pi}_{\sigma\wedge\tau_2})^{+}|
\mathcal{F}^B_{\sigma\wedge\tau_1}\bigg).
\end{eqnarray*}
Thus, from (\ref{2.3}), (\ref{5.16}) and the H\"{o}lder inequality
it follows that for any $\beta>1$ there exists a constant
$c^{(1)}_\beta$ such that
\begin{eqnarray}\label{5.17}
&\tilde{R}^H(\pi,\sigma)\leq\delta+E^B(\frac{|Z_{\sigma\wedge\tau_1}-
Z_{\sigma\wedge\tau_2}|}{Z_{\sigma\wedge\tau_2}}
(\tilde{Q}^H(\sigma,\tau_1)-\tilde{V}^{\pi}_{\sigma\wedge\tau_2})^{+})+
E^B(\tilde{Q}^H(\sigma,\tau_1)\quad\quad\\
&-\tilde{V}^{\pi}_{\sigma\wedge\tau_2})^{+}\leq \delta+
E_W(D_T\tilde{Q}^H(\sigma,\tau_1)\frac{|Z_{\sigma\wedge\tau_1}-
Z_{\sigma\wedge\tau_2}|}{Z_{\sigma\wedge\tau_2}})+E^B(\tilde{Q}^{H_\epsilon}
(\sigma,\tau_2)-\nonumber\\
&\tilde{V}^{\pi}_{\sigma\wedge\tau_2})^{+}+E^B\Gamma\leq
\delta+c^{(1)}_\beta(E_W|Z_{\sigma\wedge\tau_1}-Z_{\sigma\wedge
\tau_2}|^\beta)^{1/\beta}+\tilde{R}^{H_\epsilon}(\pi,\sigma)+E^B\Gamma.
\nonumber
\end{eqnarray}
Observe that
\begin{equation}\label{5.18}
\Gamma\leq\Gamma_1+\Gamma_2+\Gamma_3
\end{equation}
where
\begin{eqnarray*}
&\Gamma_1=|e^{-r(\tau_1\wedge\sigma)}-e^{-r(\tau_2\wedge\sigma)}|
F_{\tau_1\wedge\sigma}(S^B), \ \
\Gamma_2=|F_{\tau_1\wedge\sigma}(S^B)-F_{\tau_2\wedge\sigma}(S^B)|\\
& \mbox{and} \ \  \Gamma_3=\mathbb{I}_{\tau_H\leq
T<\tau_{H_\epsilon}}\sup_{0\leq{t}\leq{T}}F_t(S^B).
\end{eqnarray*}
 In the same way as in (\ref{3.27})--(\ref{3.31}) for any
$\beta>\frac{1}{2}$ (and not necessarily $\beta=\frac{2}{3}$ as
there) there exists a constant $c^{(2)}_\beta$ such that
$E^B(\Gamma_1+\Gamma_2)\leq
c^{(2)}_\beta(E_W(\tau_2\wedge\sigma-\tau_1\wedge\sigma)^{\beta})^{\frac{1}
{2\beta}}$. Since $\tau_2\wedge\sigma-\tau_1\wedge\sigma\leq
T\wedge(\tau_{H_\epsilon}-\tau_H)$ then similarly to
(\ref{3.32})--(\ref{3.33}) it follows that
$E_W(\tau_2\wedge\sigma-\tau_1\wedge\sigma)^{\beta}\leq
c^{(3)}_\beta\epsilon$ for some constant $c^{(3)}_\beta$. We
conclude that for any $\beta>1$ there exists a constant
$c^{(4)}_\beta$ such that
\begin{equation}\label{5.19}
E^B(\Gamma_1+\Gamma_2)\leq c^{(4)}_\beta \epsilon^{1/\beta}.
\end{equation}
Next, we estimate $E^B\Gamma_3$. First assume that $L>0$ and
$R<\infty$. Set $x_1=({\ln{L}-\ln{S_0}})/{\kappa}$,
$x_2=({\ln{R}-\ln{S_0}})/{\kappa}$,
$y_1=x_1-\frac{\epsilon}{\kappa}$ and
$y_2=x_2+\frac{\epsilon}{\kappa}$ where we set $\ln 0=-\infty$ and
$\ln\infty=\infty$. Using the stopping times $\tau^{(x)}$ and the
probabilities $P_W$ introduced in the proof of Lemma \ref{lem3.3} we
 observe that ${\{\tau_H\leq T<\tau_{H_\epsilon}\}}\subseteq {\{\tau^{(x_1)}\leq
T<\tau^{(y_1)}\}}\cup{\{\tau^{(x_2)}\leq T<\tau^{(y_2)}\}}$, and so
\begin{eqnarray}\label{5.20}
&P_W{\{\tau_H\leq T<\tau_{H_\epsilon}\}}\leq
P_W{\{\tau^{(y_1)}>T\}}-P_W{\{\tau^{(x_1)}>T\}}+P_W{\{\tau^{(y_2)}>T\}}\\
&-P_W{\{\tau^{(x_2)}>T\}} =\int_{T}^\infty \frac{1}{\sqrt{2\pi
t^3}}\sum_{i=1}^2
(|y_i|\exp(-\frac{y^2_i}{2t})-|x_i|\exp(-\frac{x^2_i}{2t}))dt.\nonumber
\end{eqnarray}
Since
$\frac{d}{dx}(x\exp(-\frac{x^2}{2t}))=(1-\frac{x^2}{t})\exp(-\frac{x^2}
{2t})\leq 1$ then it follows from the mean value theorem that
$|y_i|\exp(-\frac{y^2_i}{2t})-|x_i|\exp(-\frac{x^2_i}{2t})\leq|y_i|-|x_i|=
\frac{\epsilon}{\kappa}$ for any $i$ which together with
(\ref{5.20}) gives
\begin{equation}\label{5.21}
P_W{\{\tau_H\leq T<\tau_{H_\epsilon}\}}\leq
\frac{2\sqrt{2}\epsilon}{\sqrt{\pi T}\kappa}.
\end{equation}
For the cases $L=0$ and $R=\infty$, $P_W{\{\tau_H\leq
T<\tau_{H_\epsilon}\}}\leq P_W{\{\tau^{(y_2)}>T\}}
-P_W{\{\tau^{(x_2)}>T\}}$ and $P_W{\{\tau_H\leq
T<\tau_{H_\epsilon}\}}\leq P_W{\{\tau^{(y_1)}>T\}}\
-P_W{\{\tau^{(x_1)}>T\}}$, respectively. Thus for the above cases
(\ref{5.21}) holds true. By (\ref{5.21}) and the H\"{o}lder
inequality we see that for any $\beta>1$ there exists a constant
$c^{(5)}_\beta$ such that
\begin{equation}\label{5.22}
E^B\Gamma_3=E_W(\mathbb{I}_{\tau_H\leq
T<\tau_{H_\epsilon}}D_T\sup_{0\leq{t}\leq{T}}F_t(S^B))\leq
c^{(5)}_\beta\epsilon^{1/\beta}.
\end{equation}
Finally, we estimate
$E_W|Z_{\sigma\wedge\tau_1}-Z_{\sigma\wedge\tau_2}|^\beta$. Set
$\Gamma_4=|\frac{\mu}{\kappa}(W_{\sigma\wedge\tau_2}-W_{\sigma\wedge\tau_1})
+(\frac{\mu}{2}-\frac{r\mu}{\kappa^2}-\frac{\mu^2}{2\kappa^2})
(\sigma\wedge\tau_2-\sigma\wedge\tau_1)|.$ From the
Burkholder–-Davis-–Gandy inequality it follows that there exists a
constant $c^{(7)}_\beta$ such that $E_W\Gamma^\beta_4\leq
c^{(7)}_\beta E_W(\sigma\wedge\tau_2-\sigma\wedge\tau_1)^{\beta/2}$.
By the mean value theorem we obtain that $(e^x-1)^\beta\leq \beta
e^\beta x$ provided $0\leq x\leq 1$ and since
$Z_t=\exp(\frac{\mu}{\kappa}W_t+(\frac{\mu}{2}-\frac{r\mu}{\kappa^2}-
\frac{\mu^2}{2\kappa^2}) t)$ it follows from the Markov and
H\"{o}lder inequalities that for any $\beta>1$ there exists a
constants $c^{(8)}_\beta$ such that
\begin{eqnarray}\label{5.23}
&\,\,\,E_W|Z_{\sigma\wedge\tau_2}-Z_{\sigma\wedge\tau_1}|^\beta \leq
E_W(\sup_{0\leq{t}\leq{T}}Z^\beta_t\mathbb{I}_{\Gamma_4>1})+\beta
e^\beta E_W(\sup_{0\leq{t}\leq{T}}Z^\beta_t\Gamma_4)\,\,\,\,\,\\
&\leq c^{(8)}_\beta (P_W\{\Gamma_4>1\})^{1/\beta} +\beta e^\beta
c^{(8)}_\beta (E_W\Gamma^{\beta}_4)^{1/\beta}\leq\ (1+\beta e^\beta
)c^{(8)}_\beta
(E_W\Gamma^{\beta}_4)^{1/\beta} \nonumber\\
&\leq(1+\beta e^\beta )c^{(8)}_\beta (c^{(7)}_\beta
c^{(3)}_{\frac{\beta}{2}}\epsilon)^{1/\beta}.\nonumber
\end{eqnarray}
Letting $\del\to 0$ we complete the proof by (\ref{5.17}),
(\ref{5.19}), (\ref{5.22}) and (\ref{5.23}).
\end{proof}
Repeating the proof of the last lemma with $\tau_H=0$ and a
portfolio $\pi$ satisfying $V^\pi\equiv{0}$ we arrive at the
following result.
\begin{cor}\label{cor5.1}
Let $\tilde{H}=(L,R)$ be an open interval satisfying
$\min(\frac{R}{S_0},\frac{S_0}{R})\leq e^\epsilon$. For any
$\gamma>1$ there exists a constant $\tilde{A}_\gamma$ such that
\begin{equation}\label{5.25}
\cV-\tilde\cV^{\tilde{H}}\leq \tilde{A}_\gamma \epsilon^{1/\gamma}
\end{equation}
where
$\cV=\inf_{\sigma\in\mathcal{T}^{B}_{0T}}\sup_{\tau\in\mathcal{T}^{B}_{0T}}
\tilde{E}^B{Q}^{B}(\sigma,\tau)$ is the option price for the regular
payoff function $Q^{B}(k,l)$.
\end{cor}
Now we are ready to prove Theorem \ref{thm5.1}. Let $I=(L,R)$ be an
open interval. We start with the proof of the second statement in
the above theorem. Let $x>0$ be an initial capital and choose
$\delta>0$. As before there exists $k$, $0<t_1<t_2...<t_k\leq{T}$
and $0\leq{f_\delta}\in{C^{\infty}_0(\mathbb{R}^k)}$ such that the
portfolio $\pi\in\mathcal{A}^{B}$ with
$\tilde{V}^{\pi}_t=\tilde{E}(f_{\delta}(B^{*}_{t_1},...,
B^{*}_{t_k})|\mathcal{F}^B_t)$ satisfies
\begin{equation}\label{5.26}
\tilde{R}^{I}(\pi)<\tilde{R}^{I}(x)+\delta \ \ \mbox{and} \ \
V^\pi_0<x.
\end{equation}
For any $n$ set
\begin{equation}\label{5.27}
\Psi_n=f_{\delta}(B^{*}_{\theta^{(n)}_{[nt_1/T]}},...,
B^{*}_{\theta^{(n)}_{[nt_k/T]}}).
\end{equation}
Using the same arguments as after the formula (\ref{4.6}) it follows
that for sufficiently large $n$ there exists a portfolio
$\pi'(n)\in\mathcal{A}^{B,n}$ with an initial capital less than $x$
satisfying $\tilde{V}^{\pi'}_{\theta^{(n)}_n}=\Psi_n$. For any
$\beta>0$ which satisfy $e^\beta<\min(\frac{R}{S_0},\frac{S_0}{L})$
introduce the open interval
$\tilde{I}_\beta=(Le^\beta,Re^{-\beta})$. From (\ref{5.14}),
(\ref{5.26}) and Lemma \ref{lem5.1} it follows that for any
$\gamma>1$,
\begin{equation}\label{5.28}
\tilde{R}^I_n(x)-\tilde{R}^I(x)\leq
\delta+\tilde{R}^{B,I}_n(\pi')-\tilde{R}^I(\pi)\leq
\delta+A_{\frac{1}{\gamma}}\beta^{1/\gamma}+\tilde{R}^{B,I}_n(\pi')-
\tilde{R}^{\tilde{I}_\beta}(\pi).
\end{equation}
Let $\sigma\in\mathcal{T}^{B}_{0T}$ and
$\eta\in\mathcal{T}^{B,n}_{0,n}$ be such that
\begin{eqnarray}\label{5.29}
&\tilde{R}^{\tilde{I}_\beta}(\pi)>{\sup_{\tau\in\mathcal{T}^{B}_{0T}}
E^{B}(\tilde{Q}^{B,I_\beta}(\sigma,\tau)-\tilde{V}^{\pi}_{\sigma\wedge
\tau})^{+}}-\delta\,\, \mbox{and}\,\,
E^{B}(\tilde{Q}^{B,I,n}(\frac{\zeta{T}}{n},\frac{{\eta}T}{n})\quad\quad\\
&-\tilde{V}^{\pi'}_{\theta^{(n)}_{\zeta\wedge\eta}})^{+}>
\sup_{\tilde{\eta}\in\mathcal{T}^{B,n}_{0,n}}E^{B}(\tilde{Q}^{B,I,n}
(\frac{\zeta{T}}{n},\frac{{\tilde\eta}T}{n})
-\tilde{V}^{\pi'}_{\theta^{(n)}_{\zeta\wedge\tilde{\eta}}})^{+}-\delta
\geq{\tilde{R}^{B,I}_n(\pi')-\delta}\nonumber
\end{eqnarray}
where
$\zeta=(n\wedge\min{\{i|\theta^{(n)}_i\geq{\sigma}\}})\mathbb{I}_{\sigma<T}+
n\mathbb{I}_{\sigma=T}.$ From (\ref{5.29}) we obtain that
\begin{eqnarray}\label{5.30}
&\tilde{R}^{B,I}_n(\pi')-\tilde{R}^{\tilde{I}_\beta}(\pi)<
2\delta+E^{B}(\tilde{Q}^{B,I,n}
(\frac{\zeta{T}}{n},\frac{{\eta}T}{n})-
\tilde{V}^{\pi'}_{\theta^{(n)}_{\zeta\wedge\eta}})^{+}\\
&-E^{B}(Q^{B,\tilde{I}_\beta}(\sigma,\theta^{(n)}_{\eta}\wedge{T})-
\tilde{V}^{\pi}_{\sigma\wedge\theta^{(n)}_{\eta}})^{+} \leq 2\delta+
E^B(\tilde\Lambda_1+\tilde\Lambda_2+\tilde\Lambda_3)\nonumber
\end{eqnarray}
where
\begin{eqnarray}\label{5.31}
&\tilde\Lambda_1=|\tilde{V}^{\pi'}_{\theta^{(n)}_{\zeta\wedge\eta}}-
\tilde{V}^{\pi}_{\theta^{(n)}_{\zeta\wedge\eta}\wedge{T}}|, \ \
\tilde\Lambda_2=|\tilde{V}^{\pi}_{\theta^{(n)}_{\zeta\wedge\eta}\wedge{T}}-
\tilde{V}^{\pi}_{\theta^{(n)}_{\eta}\wedge{\sigma}}| \\
&\mbox{and} \ \
\tilde\Lambda_3=(\tilde{Q}^{B,I,n}(\frac{\zeta{T}}{n},\frac{{\eta}T}{n})-
\tilde{Q}^{B,\tilde{I}_\beta}
(\sigma,\theta^{(n)}_{\eta}\wedge{T}))^+\nonumber.
\end{eqnarray}
The quantities $\tilde\Lambda_1$ and $\tilde\Lambda_2$ can be
estimated exactly as $\Lambda_1$ and $\Lambda_2$ in the formulas
(\ref{4.14})--(\ref{4.15}), i.e. for some constant $C'(f_\delta)$
depending only on $f_\delta$,
\begin{equation}\label{5.32}
E^B(\tilde\Lambda_1+\tilde\Lambda_2)\leq C'(f_\delta)n^{-1/4}.
\end{equation}
Using the quantities $Q^B(s,t)$ and $Q^{B,n}(k,l)$ (introduced
before the formula (\ref{4.17})) and observing that
${\sigma}<\theta^{(n)}_{\eta}\wedge{T}$ if $\zeta<\eta$ we obtain
from (\ref{2.3}) and (\ref{5.31}) that
\begin{equation}\label{5.33}
\Lambda_3\leq
(Q^{B,n}(\frac{\zeta{T}}{n},\frac{{\eta}T}{n})-Q^{B}(\sigma,
\theta^{(n)}_{\eta} \wedge{T}))^{+}
+\mathbb{I}_\Xi(G_0(S_0)+\cL(T+2)(1+\max_{0\leq{k}
\leq{n}}S^{B,n}_{\frac{kT}{n}}))
\end{equation}
where
$\Xi=\{\eta\geq\tau^{B,n}_{I}\}\cap\{\theta^{(n)}_{\eta}\wedge{T}<
\tau_{\tilde{I}_\beta}\}$. Similarly to (\ref{4.18})--(\ref{4.23})
we see that there exists a constant $\tilde{C}^{(1)}$ such that
\begin{equation}\label{5.34}
E^B\Lambda_3\leq C^{(2)}n^{-1/4}(\ln{n})^{3/4}+C^{(6)}n^{-1/4}+
\tilde{C}^{(1)}P(\Xi)^{1/2}.
\end{equation}
Similarly to (\ref{4.24}) we observe that
\begin{eqnarray}\label{5.35} &\Xi\subseteq\bigg\{
 {\frac{\max_{0\leq{k}\leq{\eta}}S^{B,n}_{kT/n}}
 {\sup_{0\leq{t}\leq{\theta^{(n)}_\eta \wedge{T}}}
{S^B_t}}}>e^{\beta}\bigg\}\bigcup  \bigg\{
{\frac{\min_{0\leq{k}\leq{\eta}}S^{B,n}_{kT/n}}
 {\inf_{0\leq{t}\leq{\theta^{(n)}_\eta \wedge{T}}}
{S^B_t}}}<e^{-\beta}\bigg\}\subseteq\\
 &\bigg\{
 {\frac{\max_{0\leq{k}\leq{\eta}}S^{B,n}_{kT/n}}
 {\max_{0\leq{k}\leq{\eta}}
{S^B_{\theta^{(n)}_k\wedge{T}}}}}>e^{\beta}\bigg\}\ \bigcup \bigg\{
{\frac{\min_{0\leq{k}\leq{\eta}}S^{B,n}_{kT/n}}
 {\min_{0\leq{k}\leq{\eta}}
{S^B_{\theta^{(n)}_k
\wedge{T}}}}}<e^{-\beta}\bigg\}\subseteq\bigg\{\max_{0\leq{k}\leq{n}}\max
\nonumber\\
& \bigg(\frac{S^B_{\theta^{(n)}_k \wedge{T}}}{S^{B,n}_{kT/n}},
 \frac{S^{B,n}_{kT/n}}{S^B_{\theta^{(n)}_k \wedge{T}}}\bigg)>e^{\beta}
\bigg\}\subseteq
\{|r+\mu-\frac{\kappa^2}{2}|u_n+\kappa\sup_{T\wedge\theta^{(n)}_n\leq{t}
\leq{\theta^{(n)}_n}}|B_t -B_{\theta^{(n)}_n \wedge{T}}|>\beta\}
\nonumber
\end{eqnarray}
where the term $u_n$ was defined before formula (\ref{4.3}). Using
the Burkholder–-Davis-–Gandy inequality for the martingale
$B_t-B_T$, $t\geq{T}$ it follows that for any $m>1$ there exists a
constant $\lambda_m$ such that
$E^B(\sup_{T\wedge\theta^{(n)}_n\leq{t}\leq{\theta^{(n)}_n}}|B_t
-B_{\theta^{(n)}_n \wedge{T}}|)^m\leq
\lambda_mE^B|\theta^{(n)}_n-T|^{m/2}$. Thus from (\ref{4.3}),
(\ref{5.35}) and the Markov inequality we derive that for any $m>1$
there exists a constant $\tilde{K}^{(m)}$ such that
$P(\Xi)\leq\frac{\tilde{K}^{(m)}n^{-m/4}}{\beta^m}$. This together
with (\ref{5.28}), (\ref{5.30}), (\ref{5.32}) and (\ref{5.34})
yields that for any $\gamma,m>1$,
\begin{eqnarray}\label{5.36}
&\tilde{R}^I_n(x)-\tilde{R}^I(x)\leq 3\delta+
A_{\gamma}\beta^{1/\gamma}+(C'(f_\delta)+C^{(6)})n^{-1/4}+\\
&C^{(2)}n^{-1/4}(\ln{n})^{3/4}+\tilde{C}^{(1)}\sqrt{\frac{\tilde{K}^{(m)}
n^{-m/4}}{\beta^m}}.\nonumber
\end{eqnarray}
Thus
$\tilde{R}^I(x)\geq\limsup_{n\rightarrow\infty}\tilde{R}^I_n(x)-3\delta-
A_{\frac{1}{\gamma}}\beta^{1/\gamma}$ and by letting
$\beta,\delta\downarrow{0}$ we get that
\begin{equation}\label{5.37}
\tilde{R}^I(x)\geq\limsup_{n\rightarrow\infty}\tilde{R}^I_n(x).
\end{equation}

In order to compete the proof of the second statement in Theorem
\ref{thm5.1} we should prove (\ref{5.6}). Fix $\beta>0$ and
$n\in\mathbb{N}$. Set $J^{(n,\beta)}=(L\exp(-2n^{-1/4+\beta})
,R\exp(2n^{-1/4+\beta}))$ and let
$(\pi,\sigma)=(\psi_n(\tilde\pi^I_n),\phi_n(\tilde\sigma^I_n))$
where $(\tilde\pi^I_n,\tilde\sigma^I_n)$ is the optimal hedge given
by (\ref{5.13}). Once again we consider the portfolio
$\pi=\psi_n(\tilde\pi^I_n)$ not only as an element in
$\mathcal{A}^{B}(x)$ but also as an element in
$\mathcal{A}^{B,n}(x)$. From (\ref{5.14}) we obtain that
\begin{equation}\label{5.38}
\tilde{R}^{J^{(n,\beta)}}(\pi,\sigma)-\tilde{R}^I_n(x)=
\tilde{R}^{J^{(n,\beta)}}
(\pi,\sigma)-\tilde{R}^{B,I}_n(\pi,\tilde\zeta^I_n)
\end{equation}
where, recall, $\tilde\zeta^I_n$ was defined in (\ref{5.13}). Set
 $\zeta=\tilde\zeta^I_n$ then from (\ref{5.13}) it follows that
$\sigma=(T\wedge\theta^{(n)}_{\zeta})\mathbb{I}_{\zeta<n}
+T\mathbb{I}_{\zeta=n}$. Fix $\delta>0$ and let
$\tau\in\mathcal{T}^B_{0T}$ be such that
\begin{equation}\label{5.39}
\tilde{R}^{J^{(n,\beta)}}(\pi,\sigma)<\delta+E^{B}(\tilde{Q}^{B,J^{(n,\beta)}}
(\sigma,\tau)- \tilde{V}^{\pi}_{\sigma \wedge\tau})^{+}.
\end{equation}
Set $\eta=n\wedge\min\{k|\theta^{(n)}_k\geq{\tau}\}\in
\mathcal{T}^{B,n}_{0,n}$ and let $I^{(n,\beta)}=
(L\exp(-n^{-1/4+\beta}),R\exp(n^{-1/4+\beta}))$. Denote
\begin{eqnarray*}
&\tilde\Gamma_1=(\tilde{Q}^{B,J^{(n,\beta)}}(\sigma,{\tau})-
\tilde{Q}^{B,I^{(n,\beta)}}(\sigma\wedge
\theta^{(n)}_n,{\tau}\wedge\theta^{(n)}_n))^+\\
&\mbox{and} \ \
\tilde\Gamma_2=(\tilde{Q}^{B,I^{(n,\beta)}}(\sigma\wedge\theta^{(n)}_n,
\tau\wedge\theta^{(n)}_n)-
\tilde{Q}^{B,I}(\frac{\zeta{T}}{n},\frac{\eta{T}}{n}))^+.
\end{eqnarray*}
From (\ref{5.39}) it follows that
\begin{eqnarray}\label{5.40}
&\tilde{R}^{J^{(n,\beta)}}(\pi,\sigma)-R^{B,I}_n(\pi,\zeta)<E^{B}
(\tilde{Q}^{B,I^{(n,\beta)}}(\sigma\wedge
\theta^{(n)}_n,\tau\wedge\theta^{(n)}_n)-\tilde{V}^{\pi}_{\sigma
\wedge{\tau}})^{+}\\
&-E^{B}(\tilde{Q}^{B,I^{(n,\beta)}}(\sigma\wedge\theta^{(n)}_n,{\tau}\wedge
\theta^{(n)}_n)-
\tilde{V}^{\pi}_{\theta^{(n)}_{\zeta\wedge\eta}})^{+}+
\delta+E^{B}(\tilde\Gamma_1+ \tilde\Gamma_2). \nonumber
\end{eqnarray}
In the same way as in the formulas (\ref{4.32})--(\ref{4.38}) we
derive that
\begin{eqnarray}\label{5.41}
&E^{B}(\tilde{Q}^{B,I^{(n,\beta)}}(\sigma\wedge
\theta^{(n)}_n,\tau\wedge\theta^{(n)}_n)-\tilde{V}^{\pi}_{\sigma
\wedge{\tau}})^{+}-\\
&E^{B}(\tilde{Q}^{B,I^{(n,\beta)}}(\sigma\wedge\theta^{(n)}_n,{\tau}\wedge
\theta^{(n)}_n)-
\tilde{V}^{\pi}_{\theta^{(n)}_{\zeta\wedge\eta}})^{+}\leq{C}^{(9)}n^{-1/4}
\nonumber
\end{eqnarray}
where $C^{(9)}$ is the same constant as in formula (\ref{4.38}).

Next, we estimate $E^{B}\tilde\Gamma_1$. Since in our case
$\sigma<\tau$ is equivalent to
$\sigma\wedge\theta^{(n)}_n<\tau\wedge\theta^{(n)}_n$ then from
(\ref{2.3}) it follows that
\begin{equation}\label{5.42}
\tilde\Gamma_1\leq  (Q^{B}(\sigma,\tau)-Q^B(\sigma\wedge
\theta^{(n)}_n,{\tau}\wedge\theta^{(n)}_n))^+ +\mathbb{I}_{\Xi_1}
(G_0({S^B_0})+\cL(T+2)(1+\sup_{0\leq{t}\leq{T}}S^B_t))
\end{equation}
where
$\Xi_1=\{\tau\geq\tau_{J^{(n,\beta)}}\}\cap\{\tau\wedge\theta^{(n)}_n<
\tau_{I^{(n,\beta)}}\}.$ The term
$E^B(Q^{B}(\sigma,\tau)-Q^B(\sigma\wedge
\theta^{(n)}_n,{\tau}\wedge\theta^{(n)}_n))^+$ can be estimated by
the right hand side of (\ref{4.39}). Hence, by (\ref{4.39}) and the
Cauchy--Schwarz inequality we obtain that
\begin{equation}\label{5.42+}
E^B\tilde\Gamma_1\leq C^{(10)}n^{-1/4}+C^{(11)}(P(\Xi_1))^{1/2}
\end{equation}
where $C^{(10)}$ and $C^{(11)}$ are the same constants as in the
formulas (\ref{4.39}) and (\ref{4.41}), respectively. Similarly to
(\ref{4.43}) we see that
\begin{eqnarray}\label{5.43}
&\Xi_1\subseteq\bigg\{
 {\frac{\sup_{0\leq{t}\leq{\tau}}S^{B}_{t}}
 {\sup_{0\leq{t}\leq{\tau\wedge\theta^{(n)}_n}}
{S^B_t}}}>e^{n^{-1/4+\beta}}\bigg\}\bigcup  \bigg\{
{\frac{\inf_{0\leq{t}\leq{\tau}}S^{B}_{t}}
 {\inf_{0\leq{t}\leq{\tau\wedge\theta^{(n)}_n}}
{S^B_t}}}<e^{-n^{-1/4+\beta}}\bigg\}\\
 &\subseteq\bigg\{
\sup_{\theta^{(n)}_n\wedge{T}\leq{t}\leq{T}}\max(\frac{S^B_t}
{S^B_{\theta^{(n)}_n}},\frac{S^B_{\theta^{(n)}_n}}{S^B_t})
>e^{n^{-1/4+\beta}}
 \bigg\}\nonumber\\
 &\subseteq\{|r+\mu-\frac{\kappa^2}{2}||T-\theta^{(n)}_n|+\kappa
 \sup_{\theta^{(n)}_n\wedge{T}\leq{t}\leq{T}}|B_t-B_{\theta^{(n)}_n\wedge{T}}|>
 n^{-1/4+\beta}
 \}.\nonumber
\end{eqnarray}
Employing the Burkholder–-Davis-–Gandy inequality for the martingale
$B_t-B_T\wedge\theta^{(n)}_n$, $t\geq{T\wedge\theta^{(n)}_n}$ we
obtain that $E^B(\sup_{T\wedge\theta^{(n)}_n\leq{t}\leq{T}}|B_t
-B_{\theta^{(n)}_n \wedge{T}}|)^m\leq \lambda_m
E^B|\theta^{(n)}_n-T|^{m/2}$ for any $m>1$.  Thus, by (\ref{4.3}),
(\ref{5.43}) and the Markov inequality it follows that
$P(\Xi_1)\leq\frac{\tilde{K}^{(m)}n^{-m/4}}{n^{-m(1/4-\beta)}}$ for
any $m>1$. This together with (\ref{5.42+}) gives that
\begin{equation}\label{5.44}
E^B\tilde\Gamma_1\leq
C^{(10)}n^{-1/4}+C^{(11)}\sqrt{\tilde{K}^{(m)}n^{-m\beta}}.
\end{equation}
Finally, we estimate $E^{B}\tilde\Gamma_2$. Since $\zeta<\eta$
provided $\sigma\wedge\theta^{(n)}_n<\tau\wedge\theta^{(n)}_n$ then
by (\ref{2.3}),
\begin{eqnarray}\label{5.45}
&\tilde\Gamma_2\leq
(Q^{B}(\sigma\wedge\theta^{(n)}_n,\tau\wedge\theta^{(n)}_n)-
Q^{B,n}(\frac{\zeta{T}}{n},\frac{\eta{T}}{n}))^+\\
&+\mathbb{I}_{\Xi_2}(G_0(S_0)+\cL(T+2)(1+\sup_{0\leq{t}\leq{T}}S^B_t))
\nonumber
\end{eqnarray}
where $\Xi_2={\{\tau\wedge\theta^{(n)}_n\geq
\tau_{I^{(n,\beta)}}\}}\cup{\{\eta<\tau^{B,n}_I\}}$. The term
$E^B(Q^{B}(\sigma\wedge\theta^{(n)}_n,\tau\wedge\theta^{(n)}_n)-
Q^{B,n}(\frac{\zeta{T}}{n},\frac{\eta{T}}{n}))^+$ can be estimated
applying (\ref{4.18}) and (\ref{4.42}) which gives
\begin{equation*}
E^B\tilde\Gamma_2 \leq
C^{(2)}n^{-1/4}(\ln{n})^{3/4}+C^{(12)}n^{-1/4}+E^B\big(\mathbb{I}_{\Xi_2}
(G_0(S_0)+\cL(T+2)(1+\sup_{0\leq{t}\leq{T}}S^B_t))\big).
\end{equation*}
This together with the Cauchy-Schwarz inequality yields that
\begin{equation}\label{5.46}
E^B\tilde\Gamma_2 \leq
C^{(2)}n^{-1/4}(\ln{n})^{3/4}+C^{(12)}n^{-1/4}+C^{(11)}(P(\Xi_2))^{1/2}.
\end{equation}
Since $\tau\wedge\theta^{(n)}_n\geq\theta^{(n)}_{(\eta-1)^+}$ then
similarly to (\ref{4.43}) we obtain that $\Xi_2\subseteq\{
r(u_n+w_n)+\kappa\sqrt\frac{T}{n}>n^{-1/4+\beta}\}$. Thus $P(\Xi_2)$
can be estimated by the right hand side of (\ref{4.44}) for
$\beta>-\frac{1}{12}$, and so
\begin{equation}\label{5.47}
P(\Xi_2)\leq  C^{(13)}n^{-1/2}.
\end{equation}
Since $\delta$ is arbitrary then combining (\ref{5.38}),
(\ref{5.40}), (\ref{5.41}), (\ref{5.44}), (\ref{5.46}) and
(\ref{5.47}) we conclude that there exists a constant
$\tilde{C}^{(2)}$ such that
\begin{eqnarray}\label{5.48}
&\tilde{R}^{J^{(n,\beta)}}(\pi,\sigma)-\tilde{R}^{B,I}_n(\pi,\zeta)=
\tilde{R}^{J^{(n,\beta)}}(\pi,\sigma)-\tilde{R}^I_n(x)\leq
\tilde{C}^{(2)}n^{-1/4}(\ln{n})^{3/4}\\
&+C^{(11)}\sqrt{\tilde{K}^{(m)}n^{-m\beta}}.\nonumber
\end{eqnarray}
From (\ref{5.48}) and Lemma \ref{lem5.1} it follows that for any
$\gamma>1$,
\begin{eqnarray}\label{5.49}
&\tilde{R}^{I}(\pi,\sigma)-\tilde{R}^{B,I}_n(\pi,\zeta)=\tilde{R}^{I}(\pi,
\sigma)-\tilde{R}^I_n(x)\leq
\tilde{C}^{(2)}n^{-1/4}(\ln{n})^{3/4}+\\
&C^{(11)}\sqrt{\tilde{K}^{(m)}n^{-m\beta}}+A_\gamma
2^{1/\gamma}n^{\frac{-1/4+\beta}{\gamma}}.\nonumber
\end{eqnarray}
Let $0<\epsilon<\frac{1}{4}$ and set $\beta=\frac{\epsilon}{2}$,
$\gamma=\frac{1/4-\epsilon/2}{1/4-\epsilon}>1$ and
$m=\frac{1}{\epsilon}$. From (\ref{5.49}) we obtain that there
exists a constant $\tilde{C}_{2,\epsilon}$ such that
\begin{equation}\label{5.50}
\tilde{R}^{I}(\pi,\sigma)-\tilde{R}^I_n(x)=\tilde{R}^{I}(\pi,\sigma)-
\tilde{R}^{B,I}_n(\pi,\zeta)\leq \tilde{C}_{2,\epsilon}
n^{-\frac{1}{4}+\epsilon}.
\end{equation}
Combining (\ref{5.37}) and (\ref{5.50}) we complete the proof of the
second and the fourth statements in Theorem \ref{thm5.1}.

Next, we prove the first statement in Theorem \ref{thm5.1}. Assume
that $\mu=0$. In this case $\tilde\cV^I=\tilde{R}^I(0)$ and
$\tilde\cV^I_n=\tilde{R}^I_n(0)$. Let $0<\epsilon<\frac{1}{4}$ and
fix $n$ assuming, first, that $\exp(n^{-1/4+\epsilon/2})\geq
\min(\frac{R}{S_0},\frac{S_0}{L})$. Using Corollary \ref{cor5.1} for
$\gamma=\frac{1/4-\epsilon/2}{1/4-\epsilon}>1$ we get that
\begin{equation}\label{5.51}
\cV-\tilde\cV^I\leq \tilde{A}_\gamma n^{-1/4+\epsilon}.
\end{equation}
From Theorem 2.1 in \cite{Ki2} it follows that there exists a
constant $C$ such that $|\cV_n-\cV|\leq Cn^{-1/4}(\ln{n})^{3/4}$.
This together with (\ref{5.51}) yields that for $n$ as above,
\begin{equation}\label{5.52}
\tilde\cV^I_n-\tilde\cV^I\leq \cV_n-\tilde\cV^I\leq
Cn^{-1/4}(\ln{n})^{3/4}+\tilde{A}_{\frac{1/4-\epsilon/2}{1/4-\epsilon}}
n^{-1/4+\epsilon}.
\end{equation}
Next, assume that $\exp(n^{-1/4+\epsilon/2)}<
\min(\frac{R}{S_0},\frac{S_0}{R})$. In this case we can apply
(\ref{5.36}) for $\beta=n^{-1/4+\epsilon/2}$,
$\gamma=\frac{1/4-\epsilon/2}{1/4-\epsilon}>1$ and
$m=\frac{1}{\epsilon}$, with $C'(f_\delta)=0$ since portfolios with
zero initial capital will preserve zero value, and so the left hand
side of (\ref{5.32}) is zero. Thus we can let $\delta\downarrow 0$
in (\ref{5.36}) and obtain that for some constant $C^{(\epsilon)}$
\begin{equation}\label{5.53}
\tilde\cV^I_n-\tilde\cV^I \leq C^{(\epsilon)}n^{-1/4+\epsilon}.
\end{equation}
From (\ref{5.6}) we obtain that there exists a constant
$\tilde{C}^{\epsilon}$ such that for any $n$,
\begin{equation}\label{5.54}
\tilde\cV^I-\tilde\cV^I_n\leq
\tilde{C}^{(\epsilon)}n^{-1/4+\epsilon}.
\end{equation}
Combining (\ref{5.52}), (\ref{5.53}) and (\ref{5.54}) we complete
the proof of the first statement in Theorem \ref{thm5.1}.

Finally, we prove the third statement in Theorem \ref{thm5.1}. Fix
$n$ and $\epsilon>0$. Clearly,
$\tilde{V}^{\pi^p_n}_{\sigma^p_n\wedge{k}}
\geq\tilde{Q}^{I,n}(\sigma^p_n,k)$ for any $k$, and so
\begin{equation}\label{5.55}
\tilde{V}^{\pi}_{\zeta\wedge{k}}=\Pi_n(\tilde{V}^{\pi^p_n}_{\sigma^p_n
\wedge{k}}) \geq
\Pi_n(\tilde{Q}^{I,n}(\sigma^p_n,k))=\tilde{Q}^{B,I,n}(\zeta,k)
\end{equation}
where
$(\pi,\zeta)=(\psi_n(\pi^p_n),\Pi_n(\sigma^p_n))\in\mathcal{A}^{B,n}(\tilde
\cV^I_n)\times\mathcal{T}^{B,n}_{0,n}$. Thus,
$\tilde{R}^{B,I}_n(\pi,\zeta)=0$. Set
$\sigma=\phi_n(\sigma_n)\in\mathcal{T}^{B}_{0T}$ then
$\sigma=(T\wedge\theta^{(n)}_{\zeta})\mathbb{I}_{\zeta<n}+
T\mathbb{I}_{\zeta=n}$ and applying (\ref{5.50}) we obtain that
\begin{equation}\label{5.56}
\tilde{R}^I(\pi,\sigma)\leq
\tilde{R}^{B,I}_n(\pi,\zeta)+\tilde{C}_{2,\epsilon}
n^{-\frac{1}{4}+\epsilon}=\tilde{C}_{2,\epsilon}
n^{-\frac{1}{4}+\epsilon}
\end{equation}
completing the proof.\\
\\
\textbf{\bfseries\large Acknowledgements:}\\
\\
Partially supported by ISF grant no. 130/06.

\end{document}